\documentclass[pra,twocolumn,showpacs,floatfix,aps,nofootinbib]{revtex4-2}
\usepackage{amsmath,amssymb,epsfig,amsfonts,bbm,graphicx,times,palatino,verbatim,tensor}
\usepackage{subfigure}
\usepackage{color}
\usepackage[latin1]{inputenc}
\usepackage[english]{babel}
\usepackage{url}
\usepackage{textcase}
\usepackage{bm}
\usepackage[hidelinks]{hyperref}
\usepackage{easybmat}
\usepackage{dsfont}
\usepackage{amsthm}
\usepackage{braket}
\usepackage{ulem}
\usepackage{enumitem}

\newcommand{\norm}[1]{\left\|#1\right\|}
\newcommand{\abs}[1]{\left|#1\right|}
\DeclareMathOperator{\Tr}{Tr}
\newcommand{\id}{\mathbbm{1}}
\def\diracspacing{0.7pt}
\newcommand{\ketbra}[2]{| \hspace{\diracspacing} #1 \rangle \langle #2 \hspace{\diracspacing} |}

\newtheorem{theorem}{Theorem}

\newtheorem{lemma}[theorem]{Lemma}

\newtheorem*{remark}{Remark}

\addto\captionsenglish{}
\addto\captionsenglish{}

\usepackage{algorithm}

\makeatletter
\newenvironment{protocol}
{
		\renewcommand{\ALG@name}{Protocol}
		\refstepcounter{algorithm}
		\hrule height.8pt depth0pt \kern2pt
		\renewcommand{\caption}[2][\relax]{
			{\raggedright\textbf{\fname@algorithm~\thealgorithm} ##2\par}
			\ifx\relax##1\relax 
			\addcontentsline{loa}{algorithm}{\protect\numberline{\thealgorithm}##2}
			\else
			\addcontentsline{loa}{algorithm}{\protect\numberline{\thealgorithm}##1}
			\fi
			\kern2pt\hrule\kern2pt
		}
	}{
	\kern2pt\hrule\relax
}
\makeatother

\makeatletter
\let\cat@comma@active\@empty
\makeatother

\begin{document}
\normalem

\title{Overcoming fundamental bounds on quantum conference key agreement}

\author{Giacomo Carrara}
\author{Gl\'aucia Murta}
\author{Federico Grasselli}
\email[corresponding author:]{ federico.grasselli@hhu.de}
\affiliation{Institut f\"ur Theoretische Physik III, Heinrich-Heine-Universit\"at D\"usseldorf, Universit\"atsstra\ss{}e 1, D-40225 D\"usseldorf, Germany}

\begin{abstract}
Twin-Field Quantum Key Distribution (TF-QKD) enables two distant parties to establish a shared secret key, by interfering weak coherent pulses (WCPs) in an intermediate measuring station. This allows TF-QKD to reach greater distances than traditional QKD schemes and makes it the only scheme capable of beating the repeaterless bound on the bipartite private capacity. Here, we generalize TF-QKD to the multipartite scenario. Specifically, we propose a practical conference key agreement (CKA) protocol that only uses WCPs and linear optics and prove its security with a multiparty decoy-state method. Our protocol allows an arbitrary number of parties to establish a secret conference key by single-photon interference, enabling it to overcome recent bounds on the rate at which conference keys can be established in quantum networks without a repeater.

\end{abstract}

\maketitle

\section{Introduction}
Quantum Key Distribution (QKD) allows two parties to take advantage of quantum mechanical properties to share a common secret key with information-theoretic security. In the past decades, QKD developed at an increasingly high pace and today represents one of the most mature applications of quantum information science, both in terms of theoretical development and experimental implementation \cite{review1,review2}. More recently, in view of building quantum communication networks, a lot of effort has been put in generalizing QKD to the multipartite scenario with Conference Key Agreement (CKA) \cite{chen2004,Epping2017,Federico2018,Ottaviani2019,Zhang-CVCKA,CKArev}, which has already seen the first experimental implementations \cite{CKAexp,CKAexp2}. CKA exploits the correlations offered by multipartite entanglement to deliver the same conference key to a set of parties and it has recently been extended to guarantee anonymity of the communicating parties in a larger network \cite{ACKA-Hahn,ACKA-Grasselli,ACKA-clusterstates}. However, most CKA protocols are faced with the difficulty of establishing multipartite entanglement over large distances, limiting their applicability in real-world scenarios.

In the bipartite case, a variant of QKD, named Twin-Field QKD (TF-QKD) \cite{TF1,TF2,TF3,TF4,TF5}, enables two parties to share keys at much longer distances than most other QKD protocols. The founding idea of TF-QKD \cite{TF3,TF4,TF5} consists in a Measurement-Device Independent (MDI) scheme where a single photon sent by either of the parties interferes in an intermediate untrusted relay, thus halving the communication distance. This enables TF-QKD to beat the well-known repeaterless bound on the secret key capacity \cite{PLOB1,PLOB2}, as demonstrated by several experiments \cite{TFexp1,TFexp2,TFexp3,TFexp4,TFexp5,TFexp6,TFexp7,TFexp8,TFexp9,TFexp10,TFexp11}.

In an effort to extend the range of CKA, Ref.~\cite{Grasselli2019} introduces a CKA protocol based on single-photon interference that is inspired by the TF-QKD setup. This protocol, however, is highly unpractical as it requires each party to entangle solid-state qubits with the optical signals sent to the relay. Moreover, each party must store their qubit until the relay announces the interference outcome and then measure the qubit accordingly. 

Alternatively, more practical generalizations of TF-QKD were devised in \cite{tripartite-TFQKD1,tripartite-TFQKD2,tripartite-TFQKD3}, where the parties are only required to send weak coherent pulses or interfere the pulses with linear optics. However, the protocols in \cite{tripartite-TFQKD1,tripartite-TFQKD2,tripartite-TFQKD3} are not MDI and, more importantly, are limited to tripartite configurations and cannot be scaled to an arbitrary number of parties.

In this work, we introduce an MDI CKA protocol that does not present such drawbacks. Our protocol can be realized using only weak coherent pulses interfered with linear optics at an untrusted relay and allows an arbitrary number of parties to establish a conference key through single-photon interference. This enables our CKA protocol to operate at much higher losses than previous CKA schemes \cite{chen2004,Epping2017,Federico2018,CKAexp,CKAexp2}, which typically require the simultaneous distribution of photonic multipartite entangled states.

We prove the security of our protocol against collective attacks in the asymptotic regime by developing a multiparty decoy-state analysis \cite{Decoy1,Decoy2,Decoy3}, through which we derive analytical upper bounds on multipartite yields. We simulate the performance of our protocol with a realistic channel model that accounts for photon loss, dark counts in the detectors as well as phase and polarization misalignment.

Furthermore, we benchmark the protocol's conference key rate with recent upper bounds that apply to arbitrary quantum networks, namely the single-message multicast bound derived in \cite{PLOBN1}, adopting a similar approach used to benchmark bipartite TF-QKD setups. In particular, we consider network architectures where the relay is removed and compute their single-message multicast bounds. Our simulations show that our CKA protocol can overcome such bounds for certain noise regimes and number of parties, thus paving the way for long-distance CKA in quantum networks.

As a final remark, we note that the correlations post-selected by our protocol and used to establish the conference key belong, in essence, to W-class states \cite{Wstate}. Thus, our protocol demonstrates that conference keys can be practically established even without resorting to Greenberger-Horne-Zeilinger (GHZ) correlations. 

The paper is structured as follows. In Sec.~\ref{sec:theprotocol} we describe our CKA protocol and in Sec.~\ref{sec:security-proof} we prove its security. In Sec.~\ref{sec:decoy} we detail our multipartite decoy-state method. We simulate the protocol's performance in Sec.~\ref{sec:simulations} and conclude in Sec.~\ref{sec:conclusion}. Appendix~\ref{sec:BBSnet} describes the optical setup in the untrusted relay. In Appendix~\ref{sec:idealprotocol} we draw the connection between our protocol and the correlations of W states. The analytical upper bounds on multipartite yields are derived in Appendix~\ref{section:yieldbound}. Appendix~\ref{sec:channelmodel} contains details on the channel model and related calculations, while Appendix~\ref{section:numsim} provides details on the numerical simulations.

\section{Protocol} \label{sec:theprotocol}

In this section we present our CKA protocol based on single-photon interference, which is schematically represented in Fig.~\ref{scheme}. We limit the description to the asymptotic regime, where the effects due to finite detection statistics are negligible.

In the following, the symbol $\vec{v}$ stands for the binary representation of the integer $v$, with components $v_i\in\{0,1\}$, and $|\vec{v}|$ is the Hamming weight of the vector $\vec{v}$.
\begin{figure}[!htbp]
\centering
\includegraphics[width=\columnwidth, keepaspectratio]{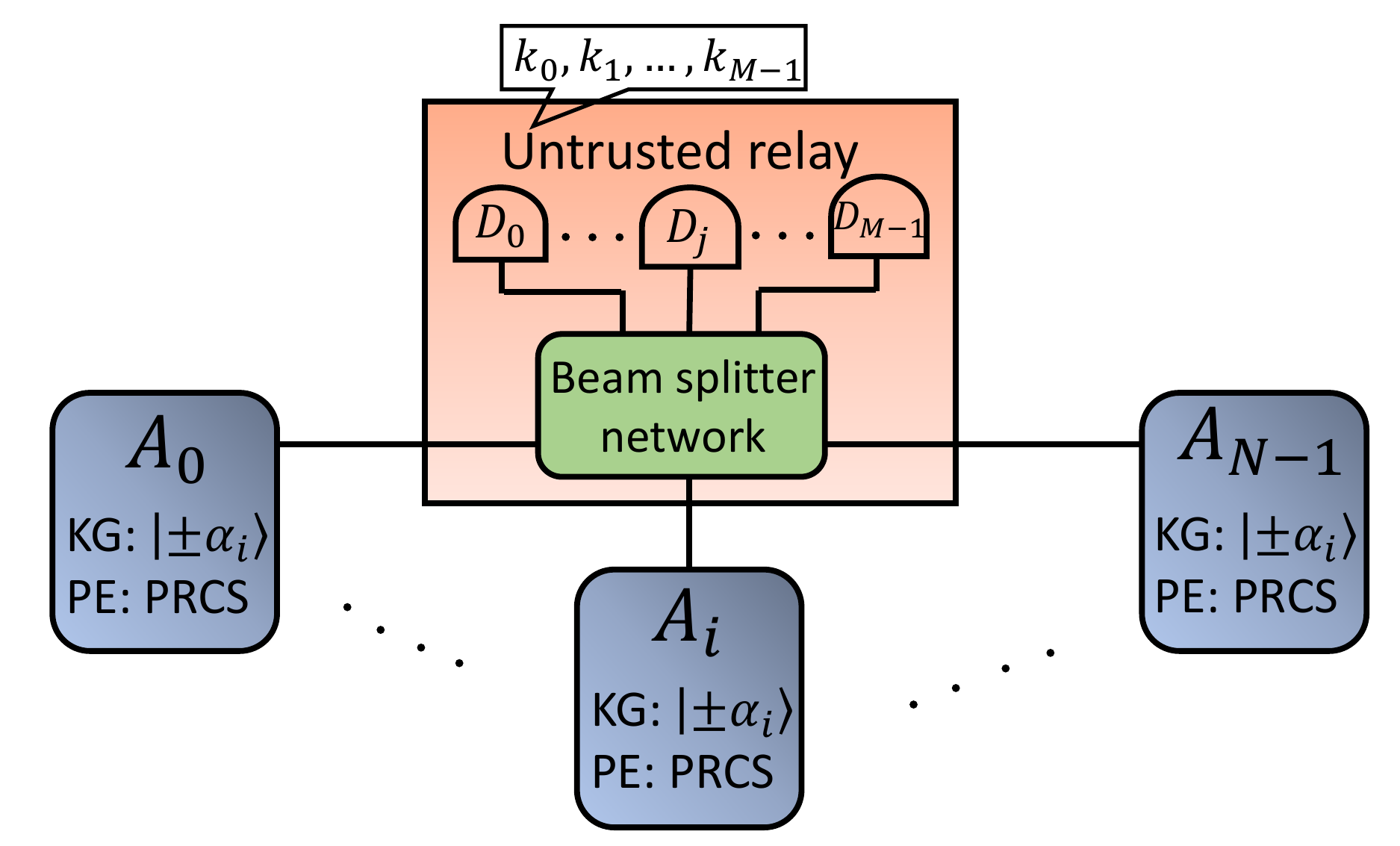}
\caption{Schematic representation of our CKA protocol. In a key generation round (KG), each party sends one of two coherent states $\ket{\pm \alpha_i}$ at random. In a parameter estimation round (PE), they send a phase-randomized coherent state (PRCS). In an honest implementation of the protocol, the relay combines the signals from each party with a beam splitter network with $M$ inputs and a threshold detector at each of the $M$ outputs (see Fig.~\ref{fig:BBSnet4.2} for the case $M=4$ and Appendix \ref{sec:BBSnet} for general $M$). The relay announces the detection pattern $\vec{k}=(k_0,k_1,\dots,k_{M-1})$.}
\label{scheme}
\end{figure}

The CKA protocol is run by $N$ parties, which we denote $A_0,A_1\dots,A_{N-1}$.\\

\begin{protocol} \label{prot:CKA}
		\caption{CKA protocol}
\begin{enumerate}[wide, labelwidth=!, labelindent=0pt]
\item Quantum state distribution and measurement: repeat the following steps a sufficiently large number of times.
\begin{enumerate}[label*=\arabic*.]
    
    \item Each party $A_i$ prepares an optical mode $a_i$ in a state that depends on whether the round is labelled as a parameter estimation (PE) round or key generation (KG) round\footnote{The type of round could be pre-determined, e.g., by a short preshared key held by every party \cite{Epping2017,Federico2018}.}. In a PE round, they prepare a phase-randomized coherent state (PRCS):
    \begin{align} \label{PEstates}
    \rho_{a_i}(\beta_i) &= e^{-\beta_i}\sum_{n=0}^{\infty}\frac{\beta_i^n}{n!}|n\rangle \langle n |\,,
    \end{align}
    where the intensity $\beta_i$ of the coherent state is chosen at random from a finite set $\mathcal{S}_i$ and where $\ket{n}$ is a Fock state. They record the intensity $\beta_i$. In a KG round, each party $A_i$ prepares the coherent state $\ket{x_i \alpha_i}_{a_i} $ for a fixed $\alpha_i \in\mathbbm{R}$, where $x_i =\pm 1$ is randomly chosen. They record the outcome $x_i$.
    
    \item Every party sends their optical pulse to an untrusted relay through an insecure channel.
    
    \item The untrusted relay performs an arbitrary operation on the $N$ optical signals and announces the pattern $\vec{k}\in\{0,1\}^M$, with $M \geq N$ [In an honest implementation of the protocol, $k_j=1$ ($k_j=0$) corresponds to a click (no click) in threshold detector $D_j$]. The round gets discarded if $|\vec{k}| \neq 1$ and we label $\Omega_j$ the event where $k_j=1$ and $k_{\neq j}=0$.
\end{enumerate}

    \item Parameter estimation: the parties partition their outcomes and intensities in $M$ sets, where each set corresponds to the event $\Omega_j$ (for $j=0,\dots,M-1$). For each partition, the parties reveal a fraction of their outcomes in order to estimate the probabilities $\Pr(\Omega_j | x_0, x_{i}, \mathrm{KG})$ that event $\Omega_j$ occurs in a KG round, given that parties $A_0$ and $A_i$ prepared coherent states $\ket{x_0 \alpha_0}$ and $\ket{x_i \alpha_i}$, respectively. With the estimated probabilities, the parties calculate the Quantum Bit Error Rate (QBER) with respect to reference party $A_0$, for every party pair and every partition ($Q^j_{X_0,X_i}$).
    
    Similarly, for each partition the parties reveal the intensities $\beta_i$ used in the PE rounds and estimate the so-called \textit{gains}, $G^j_{\beta_0,\dots,\beta_{N-1}}:=\Pr\left( \Omega_j | \beta_0,\dots,\beta_{N-1} \right)$, i.e. the probability of the event $\Omega_j$ in a PE round, given that the parties prepared PRCSs in \eqref{PEstates} with intensities $\beta_0,\dots,\beta_{N-1}$, respectively. Using the gains, the parties compute an upper bound ($\overline{Q}^j_Z$) on the phase error rate ($Q_Z^j$) of the protocol.
    
    \item Classical post-processing: The parties extract a secret conference key from the remaining (undisclosed) KG outcomes. To do so, for each partition labelled by $\Omega_j$, party $A_i$ flips their outcomes $x_i$ when $(-1)^{\vec{j}\cdot\vec{i}}=-1$. The parties then perform error correction and privacy amplification. The asymptotic conference key rate of the protocol is:
    \begin{align}\label{keyrate-lowerbound}
    r &= \sum_{j=0}^{M-1} \Pr(\Omega_j|\textrm{KG}) \left[1 - h(\overline{Q}^j_Z) -\max_{i \geq 1} h(Q^j_{X_0 X_i})\right],
    \end{align}
    where  $h(x)=-x\log_2 (x) -(1-x) \log_2 (1-x)$ is the binary entropy and where $\Pr(\Omega_j|\mathrm{KG})=(1/4)\sum_{x_0,x_i=\pm 1}\Pr(\Omega_j | x_0, x_{i}, \mathrm{KG})$ is the probability of event $\Omega_j$ in a KG round.
\end{enumerate}
\end{protocol}\vspace{0.5cm}

We prove the security of the CKA protocol in Sec.~\ref{sec:security-proof}. We remark that the security holds for any implementation of the quantum channels and of the relay, as far as the relay announces  a pattern in every round.

In an honest implementation of the protocol, the optical signals are sent through potentially noisy and lossy channels to the relay, where they interfere in a Balanced Beam Splitter (BBS) network of $M$ inputs and $M$ outputs, with $M \geq N$ and M being a power of two. The BBS network for $M=4$ is depicted in Fig.~\ref{fig:BBSnet4.2}, while the structure for generic $M$ is reported in Appendix~\ref{sec:BBSnet}.
We note that the total number of beam splitters required by the BBS network scales favourably with the number $N$ of parties, as $\mathcal{O}(N\log_2 N)$. The network transforms the input modes ($a_i$) in a balanced combination of the output modes ($d_j$), i.e.
\begin{equation}\label{transf}
    \hat{a}_i^{\dag}\rightarrow \frac{1}{\sqrt{M}}\sum_{j=0}^{M-1} (-1)^{\vec{j}\cdot \vec{i}}\,  \hat{d}_j^\dag \,.
\end{equation}
Then, the relay measures each output mode $d_j$ with a threshold detector $D_j$, for $j=0,\dots,M-1$, and announces the detection pattern $\vec{k}\in\{0,1\}^M$, where $k_j=1$ if detector $D_j$ clicked and $k_j=0$ otherwise. The round is retained only when exactly one detector clicks (event $\Omega_j$ for some $j$).

In the following, we provide the formulas to compute the QBER ($Q^j_{X_0,X_i}$) and the upper bound on the phase error rate ($\overline{Q}^j_Z$). The QBER is defined for every pair of parties $(A_0,A_i)$ and for every partition labelled by $\Omega_j$, as follows:
\begin{equation}\label{QBER}
    Q^j_{X_0,X_i}=\Pr( X_0 \neq (-1)^{\vec{j}\cdot \vec{i}} X_{i} | \Omega_j, \textrm{KG}),
\end{equation}
where $X_i$ is the binary random variable with outcomes $x_i=\pm 1$. The QBER is computed through Bayes' theorem:
\begin{align}
    Q^j_{X_0,X_i} & =  \sum_{x_0 \neq (-1)^{\vec{j}\cdot \vec{i}} x_{i}} \frac{ \Pr\left( \Omega_j | x_0, x_{i},\mathrm{KG}\right)}{4 \Pr(\Omega_j|\mathrm{KG})}.\label{QBERcomp}
\end{align}
The computation of the upper bound on the phase error rate is more involved. Indeed, it requires the derivation of upper bounds on quantities called \textit{yields} and defined as:
\begin{equation}
    Y^j_{n_0,\dots,n_{N-1}}:=\Pr(\Omega_j|n_0,\dots,n_{N-1}) ,\label{yields}
\end{equation}
i.e. the probability of the event $\Omega_j$ given the hypothetical scenario where the parties send Fock states with photon numbers $n_0, \dots, n_{N-1}$. In \eqref{upbound}, we provide analytical upper bounds ($\overline{Y}^j_{ n_0,\dots,n_{N-1}}$) on the yields as a function of the estimated gains $G^j_{\beta_0,\dots,\beta_{N-1}}$. Then, one can compute the upper bound on the phase error rate as follows:
\begin{align}\label{phase-error-rate-bound}
    &\overline{Q}^j_Z   = \frac{1}{\Pr(\Omega_j|\mathrm{KG})} \sum_{v\in\mathcal{V}} \nonumber\\
     &\left( \sum_{n_0+\dots +n_{N-1}\leq \overline{n}} \prod_{i=0}^{N-1} c_{i,n_i}^{(v_i)} \sqrt{\overline{Y}^j_{n_0,\dots,n_{N-1}}} + \Delta_{v,\overline{n}} \right)^2 ,
\end{align}
where $\overline{n}$ is a positive even number, while the set $\mathcal{V}$, the coefficients $c_{i,n_i}^{(v_i)}$ and the quantity $\Delta_{v,\overline{n}}$ are defined as follows:
\begin{align}
    \mathcal{V} &= \left\lbrace v\in\{0,2^N-1\} : |\vec{v}| \mod 2=0 \right\rbrace, \label{setV} \\[1.5ex]
    c_{i,n}^{(l)} &= \left\lbrace\begin{array}{ll}
    e^{-\alpha_i^2/2}\frac{\alpha_i^n}{\sqrt{n!}} & \mbox{if} \; n+l \; \mbox{is even} \\[1.5ex]
    0 & \mbox{if} \; n+l \; \mbox{is odd} 
   \end{array} \right. \label{coeffcat} \\[1.5ex]
  \Delta_{v,\overline{n}}&=\sum_{n_0+\dots+n_{N-1} \geq \overline{n}+2}\prod_{i=0}^{N-1} c_{i,n_i}^{(v_i)} . \label{Delta}
\end{align}
The full derivation of the upper bound \eqref{phase-error-rate-bound} on the phase error rate is provided Sec.~\ref{sec:security-proof}.

We remark that the protocol presented here uses the correlations of post-selected W-like states to obtain a secret conference key. In Appendix~\ref{sec:idealprotocol} we clarify the connection between the correlations generated in the CKA protocol and the W state. Moreover, we note that, for two parties ($N=2$), our protocol reduces to the TF-QKD protocol introduced in \cite{TF3} (see Appendix~\ref{sec:BBSnet}).

\begin{figure}[!htbp]
\centering
\includegraphics[width=\columnwidth, keepaspectratio]{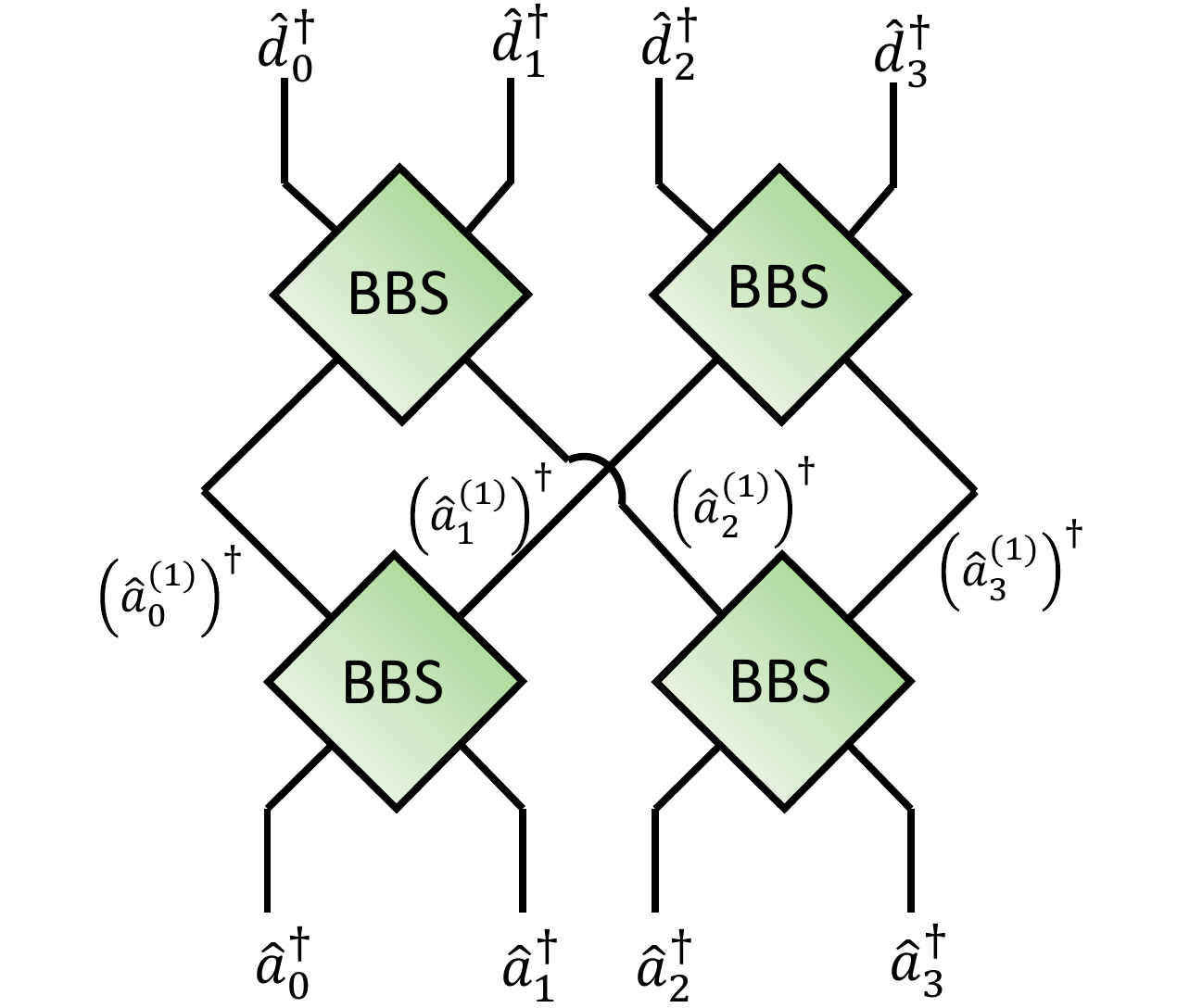}
\caption{Balanced Beam Splitter (BBS) network for $M=4$ inputs, which can be used by $N=2,3,4$ parties. The network for a general number of inputs ($M=2^s$) is described in Appendix~\ref{sec:BBSnet}.}
\label{fig:BBSnet4.2}
\end{figure}

\section{Security proof} \label{sec:security-proof}

Here we prove the security of the CKA protocol presented in Sec.~\ref{sec:theprotocol} under the assumption of collective attacks.
\begin{theorem}
The CKA protocol (Protocol~\ref{prot:CKA}), under collective attacks by the eavesdropper and in the asymptotic limit, generates a conference key with rate $r$, given by Eq. \eqref{keyrate-lowerbound}.
\end{theorem}

\begin{proof}
In the asymptotic limit and under collective attacks, the achievable conference key rate $r$ of a CKA protocol with one-way reconciliation is lower bounded by the following \cite{CKArev}:
\begin{equation}\label{DWrate}
    r \geq  H(X_0|E) - \max_{i \geq 1} H(X_0|X_i),
\end{equation}
where $H(X_0|E)$ ($H(X_0|X_i)$) is the von Neumann (Shannon) entropy of the KG outcome of reference party $A_0$, conditioned on the eavesdropper's total side information (party $A_i$'s KG outcome), and it is evaluated on the state shared by the parties in a KG round. Note that the probability of a KG round is set to one in \eqref{DWrate}, since, asymptotically, the fraction of PE rounds becomes negligible.

In the case of our protocol, we post-select the KG rounds where event $\Omega_j$ occurred and discard all the other rounds. And for each event $\Omega_j$, we independently extract a conference key. Hence, the asymptotic conference key rate of the whole protocol is bounded by:
\begin{align}\label{DWrate-oneclick}
    r &\geq \sum_{j=0}^{M-1} \Pr(\Omega_j|\textrm{KG}) \left[H(X_0|E)_{\Omega_j} - \max_{i \geq 1} H(X_0|X_i)_{\Omega_j}\right],
\end{align}
where the entropies are computed on the state shared by the parties in a KG round, conditioned on event $\Omega_j$.

Recall that, in an honest implementation, $\Omega_j$ corresponds to the event where only detector $D_j$ clicks. Although our proof holds regardless of the physical details associated to the event $\Omega_j$, in the following we often refer to $\Omega_j$ in terms of detector clicks for concreteness.

The second term in \eqref{DWrate-oneclick} is the conditional Shannon entropy between the KG outcomes of parties $A_0$ and $A_i$, when only detector $D_j$ clicked. Thus, it can be readily bounded through Fano's inequality with the corresponding QBER in \eqref{QBER} as follows:
\begin{align}
    H(X_0|X_i)_{\Omega_j} \leq h(Q^j_{X_0,X_i}). \label{shannon-upperbound}
\end{align}

In order to lower bound the first conditional entropy in \eqref{DWrate-oneclick}, we employ the entropic uncertainty relation \cite{entropicUncert}. To apply the uncertainty relation, we need to view the outcome $X_0$ corresponding to the coherent state prepared by party $A_0$ as the result of a fictitious measurement. To this aim, we consider an equivalent formulation of the   protocol where each party, in a KG round, first prepares the following entangled state between their optical mode $a_i$ and a virtual qubit $Q_i$:
\begin{equation}
        |\psi_i \rangle_{Q_i a_i} = \frac{1}{\sqrt{2}}\left( \ket{+}_{Q_i}\ket{\alpha_i}_{a_i} + \ket{-}_{Q_i}\ket{-\alpha_i}_{a_i}\right),
\end{equation}
where $\ket{\pm}=(\ket{0}\pm\ket{1})/\sqrt{2}$, and then measures their qubit in the X-basis. Note that, from the eavesdropper's point of view, the fictitious protocol is completely equivalent to the actual protocol, even in the case that the parties delay their X-basis measurement until after the relay's announcement. This allows us to consider the state of the $N$ qubits and optical modes, conditioned on detector $D_j$ clicking in a KG round, prior to the X-basis measurements. The state reads:
\begin{equation}
    \ket{\chi_j}_{Q_0 Q_1\dots Q_{N-1}E} := \frac{\hat{K}_j\left(\bigotimes_{i=0}^{N-1} \ket{\psi_i}_{Q_i a_i}\right)}{\sqrt{\Pr(\Omega_j|\mathrm{KG})}}, \label{chi-j}
\end{equation}
where $\hat{K}_j$ is the Kraus operator\footnote{Note that we can assume, without loss of generality, that there is only one Kraus operator for each announcement $\Omega_j$. Indeed, if there were more, the eavesdropper would be able to distinguish which operator acted on the optical modes with an ancillary classical flag. The flag would be part of the side information $E$, allowing the expansion of the entropy $H(X_0|E)_{\Omega_j}$ over each value of the flag, which coincides with the scenario of having only one Kraus operator.} that models the action of the untrusted relay, i.e. the eavesdropper, when it announces the event $\Omega_j$. The operator $\hat{K}_j$ acts between the Fock space of optical modes $a_0 \dots a_{N-1}$ and a generic Hilbert space $\mathcal{H}_E$, i.e., $\hat{K}_j:\;\mathcal{H}_{a_0}\otimes \dots \otimes \mathcal{H}_{a_{N-1}} \rightarrow \mathcal{H}_E$.

We remark that, due to the assumption of collective attacks, the operator $\hat{K}_j$ remains the same in every KG and PE round. Nevertheless, due to the partial distinguishability of the states prepared in KG and PE rounds, $\hat{K}_j$ could model the attempt to guess the type of round followed by an operation which is specific to KG and PE rounds. This implies that, in general, $\Pr(\Omega_j|\mathrm{KG}) \neq \Pr(\Omega_j|\mathrm{PE})$.

With the pure state in \eqref{chi-j}, we can apply the entropic uncertainty relation by considering the hypothetical scenario where party $A_0$ performs either an X-basis or a Z-basis measurement on their qubit. We thus obtain the following lower bound on the first entropy in \eqref{DWrate-oneclick}:
\begin{align}
    H(X_0|E)_{\Omega_j} \geq 1 - H(Z_0|Q_1 \dots Q_{N-1})_{\Omega_j}, \label{uncertrel}
\end{align}
where both conditional entropies are computed on the state \eqref{chi-j}. We then derive an upper bound on the entropy on the right hand side of \eqref{uncertrel} by using the fact that quantum maps on the conditioning systems can only increase the entropy \cite{NC2010}:
\begin{align}
    H(Z_0|Q_1 \dots Q_{N-1})_{\Omega_j} &\leq H(Z_0|\textstyle{\prod_{i=1}^{N-1}Z_i})_{\Omega_j} \nonumber\\
    &\leq h(Q^j_Z)  \label{uncertrel-upperbound}.
\end{align}
In the second line, we used Fano's inequality and the definition of phase error rate:
\begin{equation}
            Q^j_Z=\Pr({\textstyle\prod}_{i=0}^{N-1} Z_i=1 | \Omega_j, \textrm{KG}), \label{phase_error_rate}
\end{equation}
which expresses the probability that, in the hypothetical scenario where each party measures in the Z-basis their virtual qubit, the product of the outcomes is one. By employing \eqref{shannon-upperbound} and \eqref{uncertrel-upperbound} in \eqref{DWrate-oneclick}, we obtain the following expression for the asymptotic conference key rate of our CKA protocol:
\begin{align}\label{keyrate}
    r &\geq \sum_{j=0}^{M-1} \Pr(\Omega_j|\textrm{KG}) \left[1 - h(Q^j_Z) -\max_{i \geq 1} h(Q^j_{X_0 X_i})\right].
\end{align}
To complete the security proof, we still need to bound the phase error rate ($Q_Z^j$) with the statistics collected by the parties in the PE rounds. The derivation of the bound is inspired by the security proof in \cite{TF3} for a bipartite TF-QKD protocol.

By definition \eqref{phase_error_rate}, the phase error rate is the probability that an even number of parties obtains  $-1$ as the outcome of their Z-basis measurement, in the hypothetical scenario in which all parties measured their virtual qubit in the Z-basis in a KG round and detector $D_j$ clicks. Through the $N$-qubit state \eqref{chi-j}, which describes the state of the virtual qubits in a KG round conditioned on the click of detector $D_j$, we are able to express the phase error rate as follows:
\begin{equation}
    Q_Z^j=\sum_{v \in\mathcal{V}} \norm{\bra{ \vec{v} }_{_{Q_0 \dots Q_{N-1}}} \ket{\chi_j }}^2 \label{phase-error-rate-explicit}
\end{equation}
where the set $\mathcal{V}$ is defined in \eqref{setV}, i.e., the set of binary strings with parity zero. In order to bound the expression in \eqref{phase-error-rate-explicit}, we observe that, for $l=0,1$, we have: $_{Q_i}\braket{l|\psi_i}_{Q_i a_i} = \ket{C_i^{(l)}}_{a_i}$, where $\ket{C_i^{(l)}}_{a_i}$ are unnormalized ``cat states'': 
\begin{align}
    \ket{C_i^{(l)}}_{a_i} &= \frac{\ket{\alpha_i}+(-1)^l\ket{-\alpha_i}}{2} = \sum_{n=0}^{\infty}c^{(l)}_{i,n}\ket{n}_{a_i} \label{catstates},
\end{align}
with $c_{i,n}^{(l)}$ defined in \eqref{coeffcat}. By employing the states in \eqref{catstates}, we can derive an upper bound on each term in the sum of \eqref{phase-error-rate-explicit} as follows:
\begin{align}
    &\Pr(\Omega_j|\mathrm{KG})\norm{\bra{\vec{v}}_{Q_0 \dots Q_{N-1}} \ket{\chi_j}}^2  =  \norm{\hat{K}_j \bigotimes_{i=0}^{N-1} \ket{C_i^{(v_i)}}_{a_i} }^2 \nonumber \\
    &\quad= \norm{{\sum_{n_0,\dots,n_{N-1}=0}^{\infty}} \hat{K}_j \bigotimes_{i=0}^{N-1} c_{i,n_i}^{(v_i)}  \ket{n_i}}^2 \nonumber\\
    &\quad\leq \left( {\sum_{n_0,\dots,n_{N-1}=0}^{\infty}}  \norm{\hat{K}_j\bigotimes_{i=0}^{N-1}  c_{i,n_i}^{(v_i)} \ket{n_i}} \right)^2 \nonumber\\
    &\quad= \left( \sum_{n_0,\dots,n_{N-1}=0}^{\infty} \prod_{i=0}^{N-1}c_{i,n_i}^{(v_i)} \sqrt{Y^j_{ n_0,\dots,n_{N-1}}} \right)^2, \label{ineq}
\end{align}
where we used the fact that $\hat{K}_j$ only acts on the optical systems in the first equality and the triangle inequality in the third line. Moreover, we identified:
\begin{align} \label{yields-security}
\norm{\hat{K}_j |n_0 \rangle_{a_0}  \dots  |n_{N-1} \rangle_{a_{N-1}}}^2 &=\Pr\left(\Omega_j| n_0,\dots,n_{N-1} \right) \nonumber\\
&=: Y^j_{ n_0,\dots,n_{N-1}},  
\end{align}
as the yields. We derive an upper bound on the phase error rate by employing the inequality \eqref{ineq} in \eqref{phase-error-rate-explicit}. We obtain:
\begin{align}  \label{phase-error-rate-upp}
     Q_Z^j \leq \overline{Q}_Z^{j} &= \frac{1}{\Pr(\Omega_j|\mathrm{KG})} \sum_{v\in\mathcal{V}} \nonumber\\
     &\left( \sum_{n_0,\dots,n_{N-1}=0}^{\infty} \prod_{i=0}^{N-1}c_{i,n_i}^{(v_i)} \sqrt{Y^j_{ n_0,\dots,n_{N-1}}} \right)^2,
\end{align}
where the set $\mathcal{V}$ is given in \eqref{setV} and the coefficients $c^{(v_i)}_{i,n_i}$ are given in \eqref{coeffcat}.

The bound in \eqref{phase-error-rate-upp} is not yet sufficient to obtain a computable lower bound on the key rate \eqref{keyrate} of our CKA protocol, i.e. an expression that can be evaluated from the observed statistics. Indeed, the yields in \eqref{phase-error-rate-upp} are not directly observed and must be estimated through a multipartite decoy-state method.

From the detection statistics of PE rounds, the parties can estimate the gains. By recalling that, under the assumption of collective attacks, the Kraus operator $\hat{K}_j$ corresponding to the event $\Omega_j$ is the same in every round, we can express the gains as follows:
\begin{align}\label{PEprob}
    &G^j_{\beta_0,\dots,\beta_{N-1}} = \nonumber\\ &=\sum_{n_0,\dots,n_{N-1}=0}^{\infty} \Tr\left[\hat{K}_j \bigotimes_{i=0}^{N-1} e^{-\beta_i} \frac{\beta_i^{n_i}}{n_i!} \ket{n_i}\bra{n_i} \hat{K}^\dag_j\right] \nonumber\\
    &=\sum_{n_0,\dots,n_{N-1}=0}^{\infty} \prod_{i=0}^{N-1}P_{\beta_i}(n_{i}) \Tr\left[\hat{K}_j \bigotimes_{i=0}^{N-1} \ket{n_i}\bra{n_i} \hat{K}^\dag_j\right] \nonumber\\
    &=\sum_{n_0,\dots,n_{N-1}=0}^{\infty} \prod_{i=0}^{N-1}P_{\beta_i}(n_{i}) Y^j_{ n_0,\dots,n_{N-1}},
\end{align}
where we used \eqref{yields-security} in the last equality and defined the Poisson distribution $P_{\lambda}(n)=e^{-\lambda}\lambda^n/n!$. The last expression links the observed gains to the yields and forms the basis of our multipartite decoy-state method, which we detail in Sec.~\ref{sec:decoy}. Our method allows us to obtain analytical upper bounds $\overline{Y}^j_{n_0,\dots,n_{N-1}}$ on any yield. 

Although our method is general and works for any choice of photon numbers $n_0, \dots, n_{N-1}$, in practice it is not necessary to bound every yield appearing in \eqref{phase-error-rate-upp} with a non-trivial upper bound. This is because the product of the coefficients defined in \eqref{coeffcat} satisfies:
\begin{align}
    \prod_{i=0}^{N-1} c_{i,n_i}^{(v_i)} \neq 0 \quad\iff\quad n_{\rm tot}:=\sum_{i=0}^{N-1} n_i \,\,\mbox{is even}.
\end{align}
Therefore, the only yields contributing to the phase error rate upper bound in \eqref{phase-error-rate-upp} are those with $n_{\rm tot}=0,2,4,\dots$ and so on. Moreover, the product of the coefficients rapidly decreases with $n_{\rm tot}$, implying that it is sufficient to non-trivially bound only the yields corresponding to the first few values of $n_{\rm tot}$, while the rest of the yields can be bounded by one.  

With the yields' bounds, we can further bound the quantity in \eqref{phase-error-rate-upp} and obtain the following upper bound on the phase error rate:
\begin{align}\label{phase-error-rate-upp2}
    Q^j_Z \leq &\overline{Q}^j_Z  = \frac{1}{\Pr(\Omega_j|\mathrm{KG})} \sum_{v\in\mathcal{V}} \nonumber\\
     &\left( \sum_{n_0+\dots +n_{N-1}\leq \overline{n}} \prod_{i=0}^{N-1} c_{i,n_i}^{(v_i)} \sqrt{\overline{Y}^j_{n_0,\dots,n_{N-1}}} + \Delta_{v,\overline{n}} \right)^2 ,
\end{align}
where $\overline{Y}^j_{n_0,\dots,n_{N-1}}$ are the non-trivial bounds derived in Sec.~\ref{sec:decoy} and $\Delta_{v,\overline{n}}$ is the residual term obtained by bounding by one all the remaining yields. We have:
\begin{align}
    \Delta_{v,\overline{n}}&=\sum_{n_0+\dots+n_{N-1} \geq \overline{n}+2}\prod_{i=0}^{N-1} c_{i,n_i}^{(v_i)} , 
\end{align}
where $\overline{n}$ is an even number.

By employing \eqref{phase-error-rate-upp2} in \eqref{keyrate}, we recover the computable lower bound on the conference key rate in \eqref{keyrate-lowerbound}. This concludes the security proof.
\end{proof}

As a final remark, we stress that the assumption on collective attacks, i.e. the operator $\hat{K}_j$ being constant in every round, is instrumental in our proof. Extending the security proof to coherent attacks would mean that $\hat{K}_j$ could not only guess the type of the current round, but also depend on the sequence of previous guesses, thus not remaining constant throughout the protocol run. The security of our protocol under coherent attacks could be proved by adapting the technique in \cite{finitekeyTFQKDGuillermo}. Indeed, in \cite{finitekeyTFQKDGuillermo} the authors perform a full finite-key analysis against coherent attacks for the TF-QKD protocol in \cite{TF3}, which is recovered by our protocol when $N=2$. We conjecture that the asymptotic key rate of our protocol would not be affected by coherent attacks, as suggested by taking the asymptotic limit of the finite key rate in \cite{finitekeyTFQKDGuillermo} and realizing that it coincides with our asymptotic key rate \eqref{keyrate-lowerbound}, when $N=2$.

\section{Multipartite decoy-state method} \label{sec:decoy}

In this section we present a technique that generalizes the decoy-state method to the multipartite scenario and provides an analytical upper bound on any yield $Y^j_{n_0,\dots,n_{N-1}}$, when an arbitrary number of parties $N$ use the same set of two decoy intensities: $\mathcal{S}=\{\beta_0,\beta_1\}$.

The starting point of the multipartite decoy-state method is the equation that relates the observed gains with the yields, Eq.~\eqref{PEprob}, which we report here for clarity:
\begin{align}\label{PEprob2}
    &G^j_{\vec{f}} &=\sum_{n_0,\dots,n_{N-1}=0}^{\infty} Y^j_{ n_0,\dots,n_{N-1}} \prod_{i=0}^{N-1} \frac{e^{-\beta_{f_i}}\beta_{f_i}^{n_{i}}}{n_{i}!},
\end{align}
where we introduced the binary vector $\vec{f}$ that fixes the choice of intensity to $\beta_{f_i}$ for party $A_i$.

Importantly, the yields are independent of $\vec{f}$, i.e. of the selected intensities. Thus, Eq.~\eqref{PEprob2} can be interpreted as a system of $2^N$ linear equations, each one labelled by $\vec{f}$, where the yields are the unknowns. By performing appropriate linear combinations of the system of equations, one can derive equalities where only a subset of yields survive, thus reducing the number of unknowns. However, the number of unknowns is infinite, implying that such a technique cannot generate the exact solution for each yield. Nevertheless, from the linear combinations presenting a reduced number of yields, one can still obtain non-trivial upper bounds.

For concreteness, consider the following toy example of an equality linking a function $B$ of the observed statistics to a (possibly infinite) subset of yields, $Y$ and $Y_i$,
\begin{equation}
    B = c Y + \sum_i c_i Y_i, \label{toy}
\end{equation}
where $c$ and $c_i$ are real coefficients. Suppose that our goal is to derive an upper bound on the yield $Y$. To do so, we first split the sum of the other yields in two sums, one in which the coefficients $c_i$ have the same sign as $c$ and another where they have opposite sign. By labelling $s_i:=\mathrm{sign}(c_i)$ ($s$) the sign of coefficient $c_i$ ($c$), we have:
\begin{equation}
    B = c Y + \sum_{i:\,s_i=s} c_i Y_i + \sum_{i:\,s_i\neq s} c_i Y_i \label{toy2}.
\end{equation}
Now, by multiplying both sides by $s$ and isolating $Y$, we get:
\begin{equation}
    Y|c|  = sB - \sum_{i:\,s_i=s} |c_i| Y_i + \sum_{i:\,s_i\neq s} |c_i| Y_i \label{toy3}.
\end{equation}
Then, it is straightforward to obtain an upper bound on $Y$ by minimizing the yields $Y_i$ whose coefficients have the same sign as the coefficient of $Y$ ($s_i=s$) and by maximizing the other yields ($s_i\neq s$). In case we do not have non-trivial bounds on the yields $Y_i$, we simply set the former to zero and the latter to one. In many cases, the described procedure can lead to a non-trivial bound on $Y$:
\begin{equation}
    Y \leq \min\left\lbrace B/c  + \sum_{i:\,s_i\neq s} |c_i/c| ,\, 1\right\rbrace \label{toy4},
\end{equation}
where the minimum is taken to ensure that the bound is never greater than 1.

In Appendix~\ref{section:yieldbound}, we apply this method on the system in \eqref{PEprob2} and obtain a non-trivial upper bound on the generic yield $Y^j_{n_0,\dots,n_{N-1}}$, given by:
\begin{widetext}
\begin{align}
    \overline{Y}^j_{n_0,\dots,n_{N-1}}&=\min\{U^j_{n_0,\dots,n_{N-1}},1\}\,, \nonumber \\
    U^j_{n_0,\dots,n_{N-1}} &= \prod_{\underset{n_i\neq 0}{i \; \mathrm{s.t.}}}\frac{n_i!}{\beta_0^{n_{i}}-\beta_1^{n_{i}}}\left[\frac{B^j_{\vec{h}}\,\,(-1)^{N-m}}{(\beta_0-\beta_1)^{N-m}} + \left( e^{\beta_0}-e^{\beta_1}\right)^m \sum_{k=0}^{\left\lfloor \frac{N-m-1}{2} \right\rfloor}\binom{N-m}{2k+1}\left(\frac{\beta_1 e^{\beta_0}- \beta_0 e^{\beta_1} + \beta_0 -\beta_1}{\beta_0-\beta_1}\right)^{2k+1}\right] \label{upbound},
\end{align}
\end{widetext}
where $\vec{h}$ is the binary vector with components:
\begin{align}
    h_i = \left\lbrace \begin{array}{ll}
      1   & \mbox{if }n_i\geq 1  \\
    0  &  \mbox{if } n_i=0,
    \end{array}\right.
\end{align}
while $m=|\vec{h}|$, $\lfloor x \rfloor$ is the floor function, and $B^j_{\vec{h}}$ is given by:
\begin{equation}
    B^j_{\vec{h}}=\sum_{f=0}^{2^N-1} (-1)^{|\vec{f}|} \beta_0^{(\vec1-\vec{h})\cdot\vec{f}}\beta_1^{(\vec1-\vec{h})\cdot(\vec1-\vec{f})} \frac{G^j_{\vec{f}}}{\prod_{i=0}^{N-1}e^{-\beta_{f_i}}}.
\end{equation}

As a final remark, our analytical technique can be generalized to scenarios with different/more intensities for each party. Besides, we point out that the calculation of the yields' bounds required by the phase error rate bound in \eqref{phase-error-rate-bound} can also be done numerically by using linear programming techniques \cite{TF2}.

\section{Simulations} \label{sec:simulations}

In order to assess the performance of our   protocol, we simulate its key rate \eqref{keyrate-lowerbound} under a channel model that includes different sources of noise. First, we model the losses between each party and the detectors at the relay with the same pure loss channel with transmittance $\eta$. We also account for a polarization and phase misalignment of $2\%$ between the reference party $A_0$ and each other party. Moreover, we account for dark counts in the detectors by computing the key rates considering different dark count probabilities, namely: $10^{-8},10^{-9}$ and $10^{-10}$. In Appendix \ref{sec:channelmodel} we describe the channel model in detail and provide the calculations of the protocol's statistics under such model. 

In our symmetric channel model each party experiences the same loss. Thus, the optimal signal intensities are independent of the party, implying that we can set $\alpha_i=\alpha$ and $\mathcal{S}_i=\mathcal{S}$ for every $i$, without losing in performance. Under these conditions, we analytically verify (see Appendix \ref{sec:channelmodel}) that the detection statistics, i.e. $\Pr(\Omega_j | x_0, x_{i}, \mathrm{KG})$ and $\Pr(\Omega_j | \beta_0,\dots,\beta_{N-1})$, are independent of which detector clicks ($j$) and of the party ($i$).

This readily implies that the QBER in \eqref{QBERcomp} is independent of the party and of the detector and we can indicate it as: $Q^j_{X_0 X_i}=Q_X$. Similarly, the analytical upper bounds on the yields presented in Sec.~\ref{sec:decoy} are independent of $j$ since the gains are independent of $j$. We employ our yields bounds, \eqref{upbound}, in the calculation of the bound $\overline{Q}_Z$ on the phase error rate \eqref{phase-error-rate-bound}, where we choose $\overline{n}=4$ as the cutoff number above which every yield is trivially bounded by one. The choice is motivated by the fact that, for $\overline{n}=4$, the residual term $\Delta_{v,\overline{n}}$ in \eqref{Delta} becomes negligible.

By considering the discussed symmetries, the asymptotic conference key rate of our simulations simplifies to:
\begin{align}\label{keyrate-simulations}
    r &\geq M\Pr(\Omega|\textrm{KG}) \left[1 - h(\overline{Q}_Z) - h(Q_X)\right],
\end{align}
where $\Pr(\Omega|\textrm{KG})$ is the probability that a fixed detector clicks in a KG round and $M\geq N$ is the number of detectors in the relay.

In order to benchmark the performance of our protocol, we follow a similar approach to the one used for TF-QKD schemes. Typically, the key rate of a TF-QKD protocol is benchmarked against the repeaterless bound \cite{PLOB1}, i.e. the bound on the private capacity between Alice and Bob when the relay between the two parties is removed. In our multipartite setting, we consider the hypothetical scenario where the central relay is removed and compare our conference key rate with the resulting single-message multicast bound \cite{PLOBN1}, i.e. the ultimate rate at which conference keys can be established in the quantum network without the relay. However, the multicast bound depends on the network's architecture and there are at least two network configurations (star network and fully connected network) that can arise when removing the relay, which we depict in Fig.~\ref{fig:netarch}.

\begin{figure}[!htbp]
\centering
\includegraphics[width=\columnwidth, keepaspectratio]{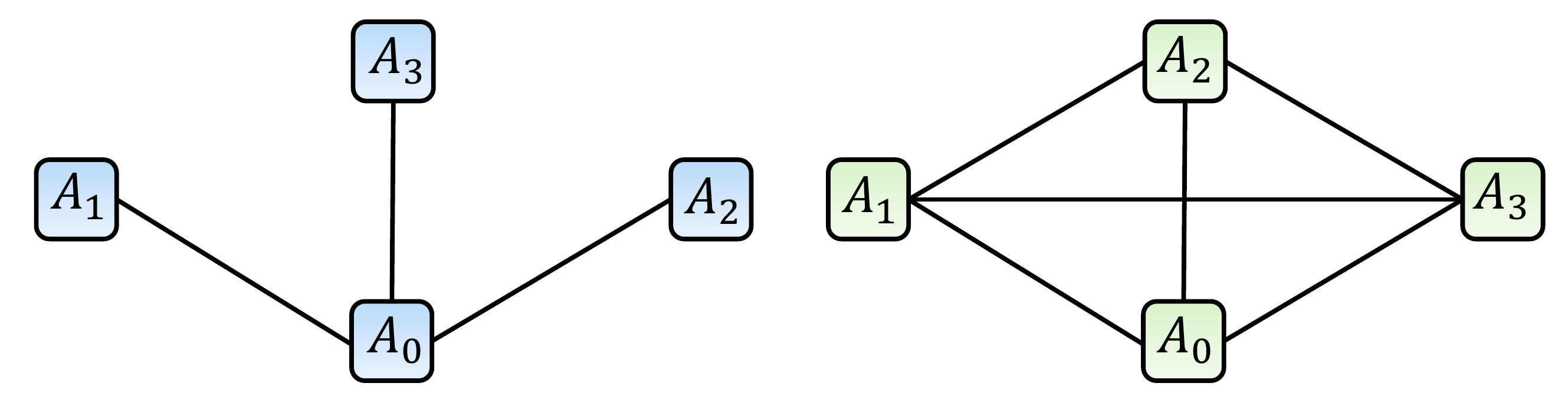}
\caption{Two possible network configurations that arise when the relay is removed, for $N=4$ parties. In the left configuration, there is a bipartite link between $A_0$ and each other party (star network). In the right configuration, each party is connected with each other (fully connected network). The transmittance of the channel connecting any two parties is $\eta^2$.}
\label{fig:netarch}
\end{figure}

In the star network (left configuration in Fig.~\ref{fig:netarch}), there is a pure-loss bosonic channel with transmittance $\eta^2$ between party $A_0$ and each other party $A_i$ (for $i=1,\dots,N-1$). In this case, the single-message multicast bound is independent of the number of parties $N$ and coincides with the bipartite repeaterless bound \cite{PLOB1} used for TF-QKD protocols: 
\begin{equation}
    r \leq -\log_2{(1-\eta^2)}=: R_1 \label{R1}.
\end{equation}
In the right configuration of Fig.~\ref{fig:netarch}, the resulting network is fully connected, such that each party is linked to each other with the same pure-loss bosonic channel with transmittance $\eta^2$. In this case, the single-message multicast bound reads \cite{PLOBN1}:
\begin{equation}
    r \leq -(N-1)\log_2{(1-\eta^2)}=: R_2(N). \label{R2}
\end{equation}
It is important to emphasize that, in order to obtain the network configurations of Fig.~\ref{fig:netarch}  when removing the relay, additional pure-loss channels need to be added on top of the existing channels used by our protocol. For instance, the star network can be seen as the result of a combination of six channels with transmittance $\eta$: three channels connect $A_0$ to the point where the relay was located and are subsequently linked to the three channels connecting to parties $A_1$, $A_2$ and $A_3$. While our CKA protocol requires only four such channels (from the relay to each of the parties) when $N=4$. This contrasts with the benchmarking of bipartite TF-QKD against the repeaterless bound, where the relay is removed and the two original channels are linked together without the need to add further channels. Therefore, when comparing the multicast bounds \eqref{R1} and \eqref{R2} with the CKA rate of our protocol, one should consider that the multicast bounds can only be attained if additional channels are used.

In Fig.~\ref{fig:rate}, we plot the key rate \eqref{keyrate-simulations} of our protocol for $N=3$, $N=4$ and $N=5$ parties, together with the multicast bounds \eqref{R1} and \eqref{R2}. In the top plot of Fig.~\ref{fig:rate}, we compute the phase error rate bound in \eqref{phase-error-rate-bound} with our analytical upper bounds on the yields \eqref{upbound} obtained with two decoy intensities fixed to $\beta_0=0.5$ and $\beta_1=0$, respectively. In the bottom plot, instead, we assume that the relevant yields in the phase error rate bound \eqref{phase-error-rate-bound} are known and use their exact analytical expression \eqref{exact-yields} (see Appendix~\ref{sec:channelmodel} for the calculation). This corresponds to the limit where the parties have an infinite number of decoy intensities and can estimate the yields exactly. In both plots, we optimize the key rate at each level of loss over the signal amplitude $\alpha$. Further details on the numerical simulations and on the optimal values for $\alpha$ are reported in Appendix~\ref{section:numsim}.

\begin{figure}[!htbp]
\centering
\textbf{Two decoys}\par\medskip
    \includegraphics[width=\columnwidth, keepaspectratio]{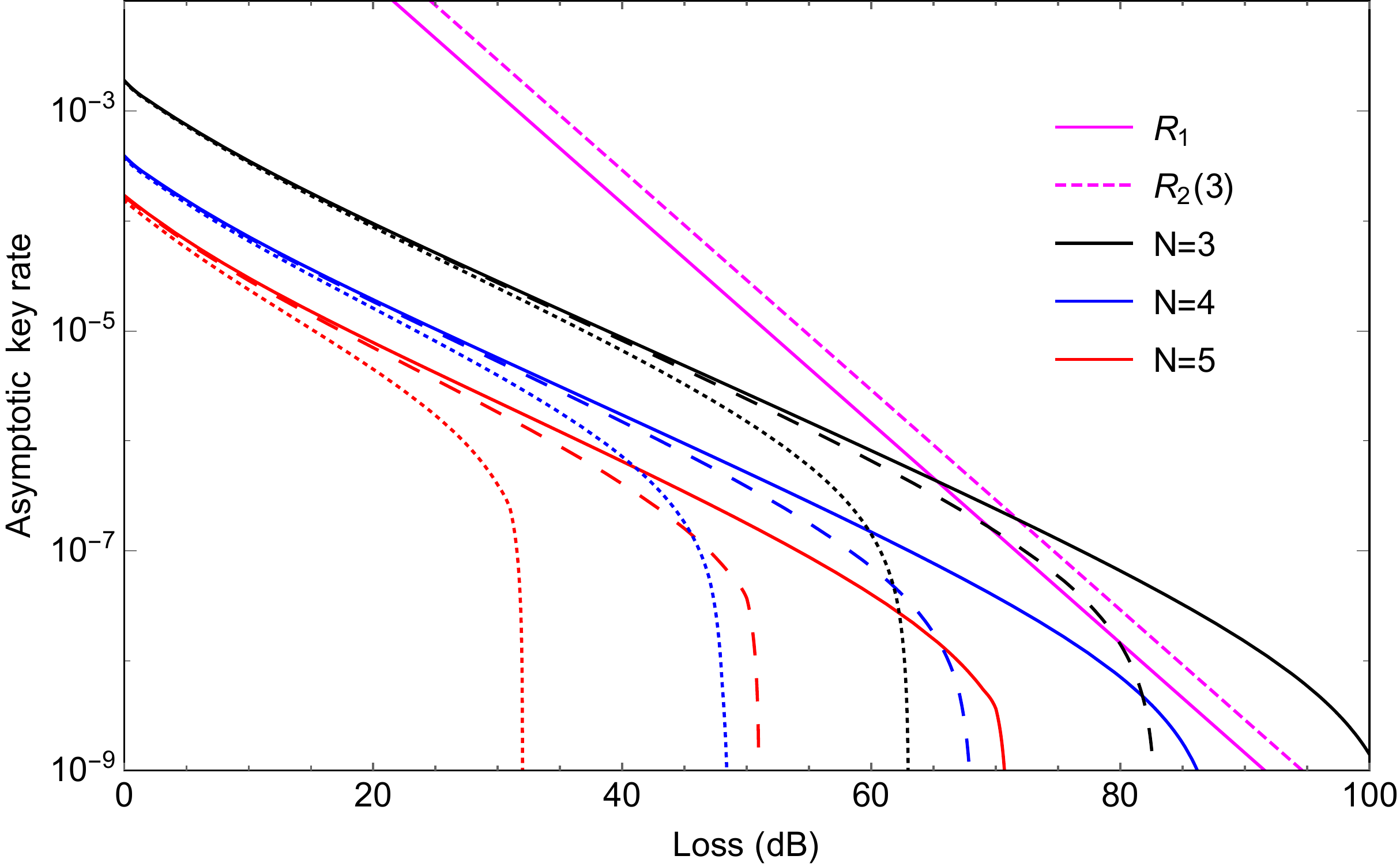}
    \textbf{Exact yields}\par\medskip
    \includegraphics[width=\columnwidth, keepaspectratio]{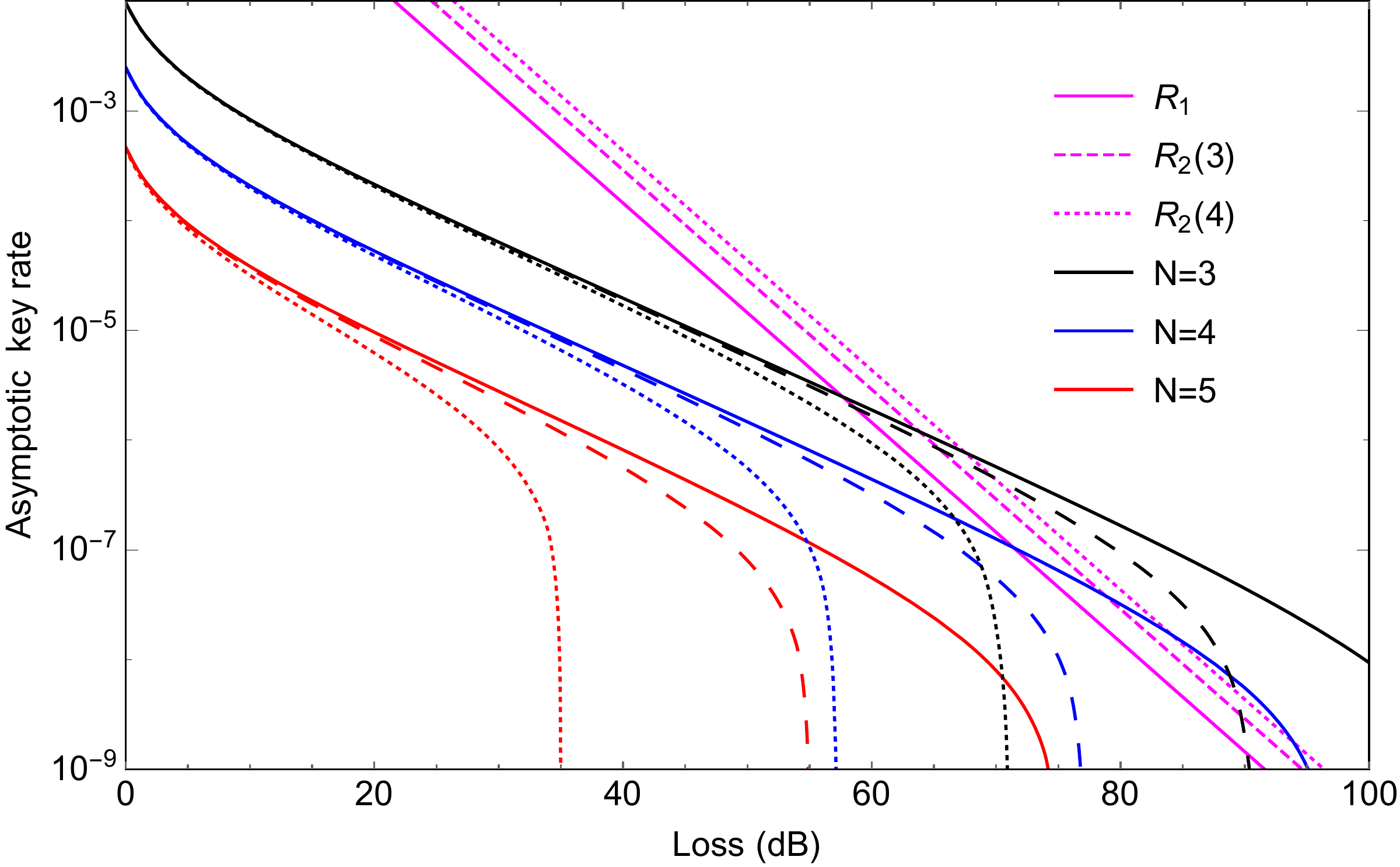}
\caption{Asymptotic conference key rate of our protocol (Eq.~\eqref{keyrate-simulations}) as a function of the loss (in dB) between any two parties, when each party uses two decoy intensities (top) and when the parties can perfectly estimate the yields (bottom). We plot the key rate for different dark count probabilities: $p_d=10^{-10}$ (solid lines), $p_d=10^{-9}$ (dashed lines) and $p_d=10^{-8}$ (dotted lines) and different number of parties $N$, while we fix the polarization and phase misalignment to $2\%$. We report the single-message multicast bound $R_1$ (Eq.~\eqref{R1}, solid magenta line) for the star network and the bounds $R_2(3)$ and $R_2(4)$ (Eq.~\eqref{R2}, dashed and dotted magenta line) for the fully connected network. For sufficiently high loss, our protocol with two decoys can overcome the multicast bounds for both configurations when $N=3$. A tighter estimation of the yields (e.g. by adding decoy intensities) would allow our protocol to overcome both multicast bounds for $N=3$ and $N=4$.}
\label{fig:rate}
\end{figure}

From the top plot in Fig.~\ref{fig:rate}, we observe that our protocol, already with two decoy intensities per party, is capable of overcoming both the single-message multicast bounds $R_1$ and $R_2$, for three parties and in the high loss regime. This is explained by the fact that our protocol relies on single-photon interference events, regardless of the number of parties, hence its key rate scales with the transmittance between one party and the relay: $r\sim \eta$. Conversely, the multicast bounds in \eqref{R1} and \eqref{R2} for a quantum network without a relay cannot scale better than $r\sim\eta^2$.

However, as the number of parties increases, the key rate of our protocol drops due to the unavoidable QBER inherited from W state correlations and cannot beat the multicast bounds. This can be mitigated by increasing the number of decoy intensities per party, as suggested by the bottom plot in Fig.~\ref{fig:rate} that represents the best-case scenario of infinite decoys. Indeed, we observe a significant improvement of the key rate, especially in the high-loss regime, allowing it to overcome the multicast bounds $R_1$ and $R_2$ for three and four parties. We also note that the multicast bounds used in our comparison are not proven to be tight and might be quite loose \cite{PLOBN1}, implying that our protocol could extend its advantage to more than four parties, when compared to tighter multicast bounds.

The improvement of the key rate in the high-loss regime occurs because adding decoy intensities to the multipartite decoy-state method allows for tighter yields' bounds when evaluating the phase error rate through \eqref{phase-error-rate-bound}. As a consequence, the optimal value of the signal intensity ($\alpha^2$) can increase without severely affecting the phase error rate bound, as shown in Appendix~\ref{section:numsim}. In turn, higher signal intensities increase the probability that exactly one detector clicks (up to the limit where multiple-photon contributions become dominant), thus increasing the key rate. The gain in the key rate is particularly visible in the high-loss regime, where the effect of dark counts on the detector clicks is comparable to the arrival of a signal.

In parallel, the key rate computed with the exact yields is higher than the the one computed with two decoys even in the low-loss regime. This is due to the fact that the latter is not optimized over the decoy intensities. In particular,  the value of $\beta_0=0.5$ is chosen such that it is close-to-optimal only for high losses, thus explaining the sub-optimal behavior of the key rate with two decoys at low losses.

\section{Conclusion} \label{sec:conclusion}

We designed a practical, measurement-device-independent, conference key agreement (CKA) protocol that delivers a shared conference key to an arbitrary number of parties. In the protocol, each party only has to transmit coherent pulses to an untrusted relay, which interferes the pulses in a network of balanced beam splitters and performs threshold measurements. Our protocol harnesses single-photon interference at the relay in order to establish a common key. This can be understood by realizing that the correlations post-selected by our protocol correspond to the correlations of a W state, which can indeed generate conference keys \cite{CKArev}.

We prove security of our protocol against collective attacks and derive an analytical expression for the asymptotic key rate, by combining the entropic uncertainty relation \cite{entropicUncert} with a novel multipartite decoy-state analysis. We emphasize that our protocol and its security proof are general and can account for scenarios with arbitrary asymmetric losses. Moreover, we provided extended numerical simulations with a realistic channel model that accounts for phase and polarization misalignment, photon loss and dark counts in the detectors. We showed that our protocol is capable, in certain regimes, to overcome the ultimate conference key rates achievable in a quantum network without a relay, by comparing it to single-message multicast bounds \cite{PLOBN1}. 

To the best of our knowledge, our protocol applies for the first time the decoy-state method in a setting with an arbitrary number of parties. In particular, we obtain analytical upper bounds on the yields of any combination of Fock states sent by the parties, which can be employed in any multipartite protocol that applies the decoy-state method.

At the same time, our protocol represents the first example of a CKA protocol that can beat single-message multicast bounds in quantum networks \cite{PLOBN1}. This heralds an important step for long-distance CKA, similarly to how the introduction of TF-QKD \cite{TF1} allowed QKD to reach much longer distances by beating the repeaterless bound \cite{PLOB1}. Indeed, our results show that adding an untrusted relay, with a relatively simple optical setup, in a quantum network, can increase the rate at which the network users establish conference keys over long distances. In particular, the scaling improvement in the key rate (the key rate scales with $\eta$ instead of $\eta^2$) matches the one that could be achieved by future quantum repeaters.

In addition, our protocol is readily implementable with current technology as it does not add further experimental requirements compared to state-of-the-art experiments on TF-QKD protocols \cite{TFexp1,TFexp2,TFexp3,TFexp4,TFexp5,TFexp6,TFexp7,TFexp8,TFexp9,TFexp10,TFexp11}. As a matter of fact, for two parties our CKA scheme reduces to the bipartite TF-QKD protocol in \cite{TF3}, which has already been implemented in several experiments \cite{TFexp1,TFexp3,TFexp6,TFexp9}. In such experiments, phase-tracking and phase-locking techniques are required in order to ensure that the parties' signals remain in phase. This, however, might become more challenging when more parties are involved. A solution could be found by adapting the idea of time multiplexing as shown in \cite{time-multiplexed-TFQKD}. We remark that the implementation of our CKA protocol would represent the first instance of a multipartite conference key agreement which does rely on GHZ-type states.

The work presented in this manuscript can be further developed along different lines of research. From a security perspective, a complete finite-key analysis along the lines of the proof given in \cite{finitekeyTFQKDGuillermo} for bipartite TF-QKD is required, in order to prove the protocol secure in the presence of statistical fluctuations and coherent attacks.

Moreover, our decoy-state analysis assumes a highly symmetrical configuration where every party uses the same set of decoy intensities, which is optimal in the scenario of symmetric channel losses analyzed in this work. However, real-life scenarios would likely display asymmetric channel losses, which require a more general decoy analysis with independent decoy intensities for each party, as shown for the bipartite case in \cite{Ybound2}. On a similar note, our decoy analysis employs only two decoy settings per party, which is  not sufficient to achieve close-to-optimal key rates (i.e. key rates obtained with infinite decoy settings), as shown by Fig. \ref{fig:rate}. Hence, it is likely that using more than two decoy settings to derive numerical or analytical bounds on the yields appearing in the phase error rate could improve the resulting key rate.

Finally, the efficiency of the protocol at lower losses could be improved by retaining those rounds where more than one detector clicks and using them to extract extra conference key bits. The security proof presented in this work could be naturally extended to make use of such rounds.

We believe that our work constitutes a significant step towards increasing the practicality of multipartite cryptographic protocols.

\begin{acknowledgments}

\noindent We thank \'{A}lvaro Navarrete for insightful discussions and Hermann Kampermann and Dagmar Bru{\ss} for their useful comments. This work was funded by the Deutsche Forschungsgemeinschaft (DFG, German Research Foundation) under Germany's Excellence Strategy - Cluster of Excellence Matter and Light for Quantum Computing (ML4Q) EXC 2004/1 - 390534769. F.G. acknowledges support by a DFG Individual Research Grant. 
\end{acknowledgments}

\bibliography{Bibliography}
\bibliographystyle{apsrev4-2}

\setcounter{equation}{0}
\setcounter{figure}{0}
\setcounter{table}{0}

\makeatletter
\renewcommand{\theequation}{\thesection.\arabic{equation}}
\renewcommand{\thefigure}{\thesection.\arabic{figure}}

\appendix

\widetext

\section{The Balanced Beam Splitter network} \label{sec:BBSnet}

We provide a complete description of the Balanced Beam Splitter (BBS) network that describes the honest implementation of the untrusted relay. The network is composed of $s$ layers, labelled by  $r=0,\dots,s-1$, and each layer receives as input $M=2^s$ optical modes $\hat{a}^{(r)}_i$, for $0 \leq i \leq M-1$. Note that for $r=0$ the modes correspond to the modes arriving at the relay from the parties.

In a generic layer $r$, the optical mode $\hat{a}^{(r)}_i$ is mixed with the mode $\hat{a}^{(r)}_{i+2^r}$ in a BBS, for all modes $\hat{a}_i \in F_r$. The set $F_r$ for layer $r$ contains the modes:
\begin{align}
    F_r:= \bigcup_{k=0}^{2^{s-r-1}-1}\{\hat{a}_{k2^{r+1}},\hat{a}_{k2^{r+1}+1}\dots,\hat{a}_{k2^{r+1}+2^r-1}\} \label{F_r}.
\end{align}
For example, $F_0$ contains the even modes and $F_1$ contains modes $0,1,3,4$, and so on. This pattern repeats until the last layer, that contains the first half of the modes.
Each layer contains $M/2$ beam splitters. Hence, the total number of beam splitters in the BBS network, in terms of the number of inputs $M$, is: 
\begin{equation}
    n_{BS}=\frac{M}{2}\log_2{M}.
\end{equation}
We note that the BBS network, due to its structure, must be prepared for a number of inputs $M$ equal to a power of 2 but can be used by any number of parties $N \leq M$. We also remark that, for $s=1$, the BBS network reduces to a single beam splitter and coincides with the setup used in the TF-QKD protocol of \cite{TF3}.  

We are interested in the evolution of the creation operators in layer $r$ through a BBS, which is given by 
\begin{equation}\label{modes}
\begin{cases}
(\hat{a}_i^{(r)})^\dag \rightarrow \frac{1}{\sqrt{2}}[(\hat{a}_i^{(r+1)})^\dag+(\hat{a}_{i+2^r}^{(r+1)})^\dag] \;\;\;\;\; \forall i \in F_r\\
(\hat{a}_j^{(r)})^\dag \rightarrow  \frac{1}{\sqrt{2}}[(\hat{a}_{j-2^r}^{(r+1)})^\dag-(\hat{a}_{j}^{(r+1)})^\dag] \;\;\;\;\; \forall j \in \overline{F}_r
\end{cases}
\end{equation}
where $\overline{F}_r$ indicates the complement of $F_r$.

By going through all layers until $r=s-1$, we are able to transform each input mode in a balanced combination of all the output modes, whose coefficients are at most a minus sign. The global mode transformation, which includes the transformation of each layer, is given in the following Theorem.

\begin{theorem}
Given $M=2^s$ input modes in the BBS network described above where, in each layer $r$, the modes transform according to Eq.~\eqref{modes}. Then, the global evolution of the modes over all the $s$ layers is given by:
\begin{equation}\label{transf_app}
    \hat{a}_i^{\dag}\rightarrow \frac{1}{(\sqrt{2})^s}\sum_{k=0}^{2^s-1}f_{k,i}^{(s)} (\hat{a}^{(s)}_k)^\dag \;\;\;\;\; \forall i=0,\dots ,2^s-1,
\end{equation}
where the function $f_{k,i}^{(s)}$ is given by
\begin{equation}
    f_{k,i}^{(s)}=\prod_{l=0}^{s-1}(-1)^{\left \lfloor k/{2^l}\right \rfloor\left \lfloor i/2^l\right \rfloor}, \label{fki}
\end{equation}
and $\lfloor\cdot \rfloor$ is the floor function. Moreover, $f_{k,i}^{(s)}$ can be recast as follows: 
\begin{equation}\label{func}
        f_{k,i}^{(s)}=(-1)^{\vec{k}\cdot \vec{i}}\,\,,
    \end{equation}
    where $\vec{k}$ and $\vec{i}$ are the binary vectors of length $s$ representing the integers $k$ and $i$ in binary representation.
\end{theorem}

\begin{proof}
The Theorem is proved by induction on $s$. Hence, the first step of the proof is to prove the result for $s=1$, i.e. just two inputs. We thus have two optical modes $\hat{a}_0^\dag$ and $\hat{a}_1^\dag$ mixed in a BBS. The transformation of the modes is given in Eq. \eqref{modes}, for $r=0$, i.e.
\begin{equation}\label{modes0}
\begin{cases}
\hat{a}^\dag_0 \rightarrow \frac{1}{\sqrt{2}}[(\hat{a}^{(1)}_0)^\dag+(\hat{a}^{(1)}_1)^\dag] \\
\hat{a}^\dag_1 \rightarrow \frac{1}{\sqrt{2}}[(\hat{a}^{(1)}_0)^\dag-(\hat{a}^{(1)}_1)^\dag]
\end{cases}
\end{equation}
The formula provided in the Theorem's statement, \eqref{transf_app}, for $s=1$ reads 
\begin{equation}
 \hat{a}_i \rightarrow \frac{1}{\sqrt{2}} \sum_{k=0}^1 (-1)^{ki} (\hat{a}_k^{(1)})^\dag
\end{equation}
which is equivalent to the transformation of the modes of Eq. \eqref{modes0}. The Theorem is thus proved for $s=1$.

Now, in the inductive step we assume that the Theorem's statement in \eqref{transf_app} is correct for generic $s$ and show that it induces the same transformation for $s+1$, i.e. that the Theorem holds for $s+1$.

We start by adding to the modes labelled by $i$ another set of $2^s$ modes, labelled by $j=2^s,\dots,2^{s+1}-1$, that undergoes the same kind of transformations, i.e.
\begin{equation}\label{transf1}
    \hat{a}_j^{\dag}\rightarrow \frac{1}{(\sqrt{2})^s}\sum_{k=2^s}^{2^{s+1}-1}f_{k,j}^{(s)} (\hat{a}^{(s)}_k)^\dag \;\;\;\;\; \forall j=2^s,\dots ,2^{s+1}-1.
\end{equation}
We now follow the prescription in \eqref{modes} and combine the modes in the $s+1$ layer of the BBS network. This means that we combine the mode  $(\hat{a}^{(s)}_i)^\dag$ with the corresponding mode $(\hat{a}^{(s)}_{i+2^s})^\dag$, and obtain:
\begin{equation}\label{modes1}
\begin{cases}
(\hat{a}_i^{(s)})^\dag \rightarrow \frac{1}{\sqrt{2}}[(\hat{a}_i^{(s+1)})^\dag+(\hat{a}_{i+2^s}^{(s+1)})^\dag] \;\;\;\;\; \forall i=0,\dots, 2^s-1\\
(\hat{a}_j^{(s)})^\dag \rightarrow  \frac{1}{\sqrt{2}}[(\hat{a}_{j-2^s}^{(r+1)})^\dag-(\hat{a}_{j}^{(s+1)})^\dag] \;\;\;\;\; \forall j=2^s,\dots,2^{s+1}-1
\end{cases}
\end{equation}
We use the assumption in the inductive step. That is, we use \eqref{transf_app} and \eqref{transf1} to describe the transformations of the modes in the first $s$ layers. We separately address the transformations on the first $2^s$ modes $\hat{a}_i^\dag$, described by \eqref{transf_app}, and the transformations on the other $2^s$ modes $\hat{a}_j^\dag$, described by \eqref{transf1}.

\begin{enumerate}
    \item For the modes $\hat{a}_i^\dag$ with $i=0,\dots, 2^s-1$, we employ the first equation in \eqref{modes1} together with \eqref{transf_app}. We obtain the following transformation of the modes after $s+1$ layers:
    \begin{eqnarray}\label{sum1}
    \hat{a}_i^{\dag} & \rightarrow & \frac{1}{(\sqrt{2})^s}\sum_{k=0}^{2^s-1}f_{k,i}^{(s)} \frac{1}{\sqrt{2}}[(\hat{a}_k^{(s+1)})^\dag+(\hat{a}_{k+2^s}^{(s+1)})^\dag] \nonumber \\
    & = & \frac{1}{(\sqrt{2})^{s+1}}\left(\sum_{k=0}^{2^s-1}f_{k,i}^{(s)}(\hat{a}_k^{(s+1)})^\dag + \sum_{k=0}^{2^s-1}f_{k,i}^{(s)}(\hat{a}_{k+2^s}^{(s+1)})^\dag \right)
    \end{eqnarray}
    Now let us consider the coefficient $f_{k,i}^{(s+1)}$. By definition \eqref{fki}, we have: 
    \begin{eqnarray}
    f_{k,i}^{(s+1)} & = & \prod_{l=0}^{s}(-1)^{\left \lfloor k/2^l\right \rfloor\left \lfloor i/2^l\right \rfloor}=(-1)^{\left \lfloor k/2^s\right \rfloor \left \lfloor i/2^s\right \rfloor}\prod_{l=0}^{s-1}(-1)^{\left \lfloor k/2^l\right \rfloor\left \lfloor i/2^l\right \rfloor} \nonumber \\
    & = & (-1)^{\left \lfloor k/2^s\right \rfloor \left \lfloor i/2^s\right \rfloor} f_{k,i}^{(s)}
    \end{eqnarray}
    However, since $i=0,\dots, 2^s-1$ we have that $\left \lfloor i/2^s \right \rfloor=0 \;\; \forall i$, which in turn implies 
    \begin{equation}
    f_{k,i}^{(s+1)}=f_{k,i}^{(s)} \;\;\;\;\; \forall \,i=0,\dots,2^s -1, \,\,\forall \,k. \label{equality}
    \end{equation}
    We use this result in Eq. \eqref{sum1} combined with a rescaling of the second sum with $k \rightarrow k-2^s$ to write
    \begin{equation}
          \hat{a}_i^{\dag} \rightarrow \frac{1}{(\sqrt{2})^{s+1}}\left(\sum_{k=0}^{2^s-1}f_{k,i}^{(s+1)}(\hat{a}_k^{(s+1)})^\dag + \sum_{k=2^s}^{2^{s+1}-1}f_{k-2^s,i}^{(s)}(\hat{a}_{k}^{(s+1)})^\dag \right) \label{sum2},
    \end{equation}
    where
    \begin{eqnarray}
    f_{k - 2^s,i}^{(s)}= \prod_{l=0}^{s-1}(-1)^{\left \lfloor (k - 2^s)/2^l\right \rfloor\left \lfloor i/2^l\right \rfloor} = \prod_{l=0}^{s-1}(-1)^{\left \lfloor k/2^l - 2^{s-l}\right \rfloor\left \lfloor i/2^l\right \rfloor}.
    \end{eqnarray}
    Since $k \geq 2^s$ we have that $k/2^l\geq 2^{s-l}$. Moreover, $s>l$ for every $l$, which means that $2^{s-l}$ is a positive, even integer. We thus can write
       \begin{eqnarray}
    f_{k - 2^s,i}^{(s)} & = & \prod_{l=0}^{s-1}(-1)^{\left( \left \lfloor k/2^l \right \rfloor - 2^{s-l}\right) \left \lfloor i/2^l\right \rfloor} = \prod_{l=0}^{s-1}(-1)^{ \left \lfloor k/2^l \right \rfloor \left \lfloor i/2^l\right \rfloor}(-1)^{ -2^{s-l} \left \lfloor i/2^l\right \rfloor} \nonumber \\
    & = & \prod_{l=0}^{s-1}(-1)^{ \left \lfloor k/2^l \right \rfloor \left \lfloor i/2^l\right \rfloor} = f_{k,i}^{(s)} = f_{k,i}^{(s+1)}
    \end{eqnarray}
    where $(-1)^{ -2^{s-l} \left \lfloor i/2^l\right \rfloor}=1$ $\forall i$ because $2^{s-l}$ is even for all $l$ and where we used \eqref{equality} in the last equality. With the last expression, we can simplify \eqref{sum2} as follows:
     \begin{eqnarray}
          \hat{a}_i^{\dag} & \rightarrow & \frac{1}{(\sqrt{2})^{s+1}}\left(\sum_{k=0}^{2^s-1}f_{k,i}^{(s+1)}(\hat{a}_k^{(s+1)})^\dag + \sum_{k=2^s}^{2^{s+1}-1}f_{k,i}^{(s+1)}(\hat{a}_{k}^{(s+1)})^\dag \right) \nonumber \\
          & = & \frac{1}{(\sqrt{2})^{s+1}} \sum_{k=0}^{2^{s+1}-1}f_{k,i}^{(s+1)}(\hat{a}_k^{(s+1)})^\dag \label{sumfinal1}
    \end{eqnarray}
    which concludes the proof for $i=0,\dots, 2^s-1$.
    \item For the modes $\hat{a}_j^\dag$, with $j=2^s,\dots,2^{s+1}-1$, we combine the second equation in  \eqref{modes1} with the assumption \eqref{transf1} and obtain
    \begin{equation}
        \hat{a}_j^\dag \rightarrow \frac{1}{(\sqrt{2})^{s+1}}\left(\sum_{k=2^s}^{2^{s+1}-1}f_{k,j}^{(s)}(\hat{a}_{k-2^s}^{(s+1)})^\dag - \sum_{k=2^s}^{2^{s+1}-1}f_{k,j}^{(s)}(\hat{a}_k^{(s+1)})^\dag \right).
    \end{equation}
    Once again, we can rescale the first sum with $k \rightarrow k+2^s$ in the last expression and obtain 
    \begin{equation}
        \hat{a}_j^\dag \rightarrow \frac{1}{(\sqrt{2})^{s+1}}\left(\sum_{k=0}^{2^{s}-1}f_{k+2^s,j}^{(s)}(\hat{a}_{k}^{(s+1)})^\dag - \sum_{k=2^s}^{2^{s+1}-1}f_{k,j}^{(s)}(\hat{a}_k^{(s+1)})^\dag \right). \label{sum3}
    \end{equation}
    Since $2^{s-l}$ is a positive, even integer, we can simplify the coefficient in the first sum as follows:
    \begin{eqnarray}
    f_{k + 2^s,j}^{(s)} & = & \prod_{l=0}^{s-1}(-1)^{\left( \left \lfloor k/2^l \right \rfloor + 2^{s-l}\right) \left \lfloor j/2^l\right \rfloor} = \prod_{l=0}^{s-1}(-1)^{ \left \lfloor k/2^l \right \rfloor \left \lfloor j/2^l\right \rfloor}(-1)^{ 2^{s-l} \left \lfloor j/2^l\right \rfloor} \nonumber \\
    & = & \prod_{l=0}^{s-1}(-1)^{ \left \lfloor k/2^l \right \rfloor \left \lfloor j/2^l\right \rfloor} = f_{k,j}^{(s)}.
    \end{eqnarray}
    Moreover, since $k=0,\dots,2^s-1$, one has that $\left \lfloor k/2^s \right \rfloor=0$ and hence that: 
    \begin{eqnarray}
    f_{k,j}^{(s+1)} & = & \prod_{l=0}^{s}(-1)^{\left \lfloor k/2^l\right \rfloor\left \lfloor j/2^l\right \rfloor} = (-1)^{\left \lfloor k/2^s\right \rfloor \left \lfloor j/2^s\right \rfloor}\prod_{l=0}^{s-1}(-1)^{\left \lfloor k/2^l\right \rfloor\left \lfloor j/2^l\right \rfloor} \nonumber \\
    & = & (-1)^{\left \lfloor k/2^s\right \rfloor \left \lfloor i/2^s\right \rfloor} f_{k,i}^{(s)} = f_{k,i}^{(s)},
    \end{eqnarray}
    which means that we can replace the coefficient  $f_{k+2^s,j}^{(s)}$ in the first sum of \eqref{sum3}  with $f_{k,j}^{(s+1)}$. Regarding the second sum in \eqref{sum3}, we can write the coefficient as
    \begin{equation}
        (-1)f_{k,j}^{(s)}=(-1)^{g(k,j)}f_{k,j}^{(s)}
    \end{equation}
    where $ g(k,j)$ is a function that is odd for $j,k=2^s,\dots,2^{s+1}-1$. For instance, we can choose the function to be the following:
    \begin{equation}
        g(k,j)=\left \lfloor k/2^s\right \rfloor \left \lfloor j/2^s\right \rfloor.
    \end{equation}
     Then, the coefficient of the second sum in \eqref{sum3} becomes
     \begin{equation}
        (-1) f_{k,j}^{(s)}= (-1)^{\left \lfloor k/2^s\right \rfloor \left \lfloor j/2^s\right \rfloor}f_{k,j}^{(s)}\equiv  f_{k,j}^{(s+1)}.
    \end{equation}
    With the above expressions, we can recast \eqref{sum3} as follows and conclude the proof for $j=2^s,\dots,2^{s+1}-1$:
    \begin{eqnarray}
        \hat{a}_j^\dag & \rightarrow & \frac{1}{(\sqrt{2})^{s+1}}\left(\sum_{k=0}^{2^{s}-1}f_{k,j}^{(s+1)}(\hat{a}_{k}^{(s+1)})^\dag + \sum_{k=2^s}^{2^{s+1}-1}f_{k,j}^{(s+1)}(\hat{a}_k^{(s+1)})^\dag \right) \nonumber \\
        & = & \frac{1}{(\sqrt{2})^{s+1}} \sum_{k=0}^{2^{s+1}-1}f_{k,j}^{(s+1)}(\hat{a}_k^{(s+1)})^\dag \label{sumfinal2}.
    \end{eqnarray}
\end{enumerate}
The combination of the two results in \eqref{sumfinal1} and \eqref{sumfinal2} imply that the  global transformation of the modes, for $M=2^{s+1}$ inputs, is given by \eqref{transf_app} where $s$ is replaced by $s+1$. This proves the Theorem for $s+1$ and concludes the proof. 
\end{proof}

\section{W state correlations} \label{sec:idealprotocol}

In this Appendix we present the logical steps that brought us to design the protocol presented in Sec.~\ref{sec:theprotocol} and show the connection between the correlations generated by our protocol and the correlations of the W state \cite{Wstate}.

We start by describing an Ideal protocol, i.e. a protocol that is less practical than the one presented in the main text but has the merit of elucidating the core ideas that lead to the CKA protocol of Sec.~\ref{sec:theprotocol}. The protocol is run by $N$ parties, which we call $A_0,\dots,A_{N-1}$, and consists of the following steps.\\

\begin{protocol} \label{prot:ideal}
		\caption{Ideal protocol}
\begin{enumerate}
\item Quantum part: repeat what follows for a sufficient amount of iterations.
\begin{enumerate}[label*=\arabic*.]
    \item Every party holds an optical mode $a_i$ and a qubit $Q_i$ and prepares the following entangled state:
    \begin{equation}\label{startstate}
    |\phi_i \rangle = \sqrt{q_i}|0\rangle_{Q_i}|0\rangle_{a_i}+\sqrt{1-q_i}|1\rangle_{Q_i}|1\rangle_{a_i},
\end{equation}
    where $|0\rangle_{Q_i}$ and $|1\rangle_{Q_i}$ are two orthogonal states of the qubit, $|0\rangle_{a_i}$ and $|1\rangle_{a_i}$ are the vacuum and one-photon state of the optical mode, respectively, and $0< q_i<1$. 
    \item Every party sends their optical pulse through a noisy and lossy channel to an untrusted relay.
    \item In the untrusted relay, the optical signals interfere in a Balanced Beam Splitter (BBS) network of $M=2^s$ inputs and $M$ outputs, for some natural number $s$ with $M \geq N$. The BBS network is described in Appendix~\ref{sec:BBSnet}. The network transforms the input modes in a balanced combination of the output modes, i.e.
    \begin{equation}\label{transf-app}
    \hat{a}_i^{\dag}\rightarrow \frac{1}{\sqrt{M}}\sum_{j=0}^{M-1} (-1)^{\vec{j}\cdot \vec{i}}\,  \hat{d}_j^\dag \quad,
\end{equation}
   where $\hat{a}_i^{\dag}$ and $\hat{d}_j^\dag$ are the creation operators of the input and output modes, respectively, and $\vec{j}$ and $\vec{i}$ the binary representations of the integers $j$ and $i$ and $\vec{j}\cdot \vec{i}$ is their scalar product.
    
    \item The untrusted relay measures each output mode $d_j$ with a threshold detector $D_j$, for $j=0,\dots,M-1$. The relay announces the detection pattern $\vec{k}\in\{0,1\}^M$ for each detector, where $k_j=1$ if detector $D_j$ clicked and $k_j=0$ otherwise. The round gets discarded unless only one detector clicked, i.e. if $|\vec{k}| = 1$, where $|\vec{x}|$ is the Hamming weight of vector $\vec{x}$.
    
    \item Each party $A_i$ measures their qubit $Q_i$. If the round is labelled as a parameter estimation (PE) round, each party measures in the Z-basis and obtains an outcome $Z_i=\pm 1$. If the round is a key generation (KG) round, each party measures in the X-basis and obtains outcome $X_i=\pm 1$.
\end{enumerate}
    \item Parameter estimation: the parties partition their outcomes in $M$ sets, where each set corresponds to the event $\Omega_j$ where only detector $D_j$ clicks. For each partition, the parties reveal a fraction of their $X$-basis outcomes in order to compute the QBER, with respect to reference party $A_0$. The QBER is defined as:
        \begin{equation}\label{QBER-app}
            Q^j_{X_0,X_i}=\Pr( X_0 \neq (-1)^{\vec{j}\cdot \vec{i}} X_{i} \;|\; \Omega_j, \textrm{KG}).
        \end{equation}
    Similarly, for each partition of outcomes the parties reveal their $Z$-basis outcomes and evaluate the phase error rate, defined as follows: 
        \begin{equation}
            Q^j_Z=\Pr({\textstyle\prod}_{i=0}^{N-1} Z_i=1 \;|\; \Omega_j, \textrm{KG}). \label{phase_error_rate-app}
        \end{equation}
    \item Classical post-processing: The parties extract a secret conference key from the remaining undisclosed $X$-basis outcomes. To do so, for each partition labelled by $\Omega_j$, party $A_i$ flips their $X$-basis outcomes when $(-1)^{\vec{j}\cdot\vec{i}}=-1$. The parties then perform error correction and privacy amplification.
\end{enumerate}
\end{protocol}\vspace{0.5cm}

We remark that the probabilities defining the QBER \eqref{QBER-app} and the phase error rate \eqref{phase_error_rate-app} are conditioned on the event that only detector $D_j$ clicked and the round was chosen to be a KG round. While the QBER can be directly computed from the outcomes collected in KG rounds, the phase error rate refers to the hypothetical scenario where the parties measured in the Z-basis in a KG round. However, since the only difference between KG and PE rounds is the local qubit measurement, the choice of the type of round can be delayed until the qubit measurement is performed. Hence, the phase error rate, as defined in \eqref{phase_error_rate-app}, effectively coincides with the analogous quantity observed from the PE data: $Q^j_Z \equiv \Pr({\textstyle\prod}_{i=0}^{N-1} Z_i=1 \;|\; \Omega_j, \textrm{PE})$. As we discuss below, this fact does not hold in our CKA protocol (Sec.~\ref{sec:theprotocol}), where the phase error rate \eqref{phase_error_rate} is indirectly bounded with the PE statistics thanks to a multiparty decoy-state method.

The Ideal protocol is designed to exploit the correlations of a particular class of multipartite, W-type states, which are post-selected due to single-photon interference. As a matter of fact, a noisy version of such states is recovered as the conditional state of the qubit systems post-selected on the event that only detector $D_j$ clicks. In order to see this more clearly, one can derive such state under the ideal conditions of no losses in the channels and $q_i=q \to 1$ for every $i$ --- indeed, the optimal values of $q$ are close to one \cite{Grasselli2019}, hence  we approximate the state to first order in $(1-q)$. Under these simplifications, the state of the qubits $Q_0 \dots Q_{N-1}$ shared by the $N$ parties, once post-selected on the click of detector $D_j$, reads:
\begin{equation}\label{postselect_state}
    |W_j\rangle_{{Q_0,\dots} Q_{N-1}} :=\frac{1}{\sqrt{N}}\sum_{i=0}^{N-1} (-1)^{\vec{j}\cdot \vec{i}}\ket{\vec{b}_i}_{{Q_0,\dots} Q_{N-1}}\,,
\end{equation}
where the vector $\vec{b}_i$ is defined as the $N$-bit vector of all zeroes except for the $i$-th element that is one.

The state in \eqref{postselect_state} is a W-type state, where each term in the sum presents a real phase determined by the detector that clicked. The state  is post-selected from the events where only one photon was effectively sent by any of the parties with equal probability. Indeed, under the above approximations, the probability that the W-type state in \eqref{postselect_state} is post-selected is  $q^{N-1}(1-q)N/M$.

In this regard, the Ideal protocol resembles the CKA protocol of \cite{Grasselli2019} as it exploits the multipartite correlations of a W state to establish a shared conference key. As a matter of fact, we note that in the classical post-processing the parties flip their $X$-basis outcomes according to $(-1)^{\vec{j}\cdot \vec{i}}$, where $\vec{i}$ depends on the party and $\vec{j}$ on the detector that clicked. This can be equivalently seen as party $A_i$ applying a $Z$ gate on their qubit before the $X$-basis measurement, if $\vec{j}\cdot \vec{i}$ is odd. In other words, party $A_i$ applies the gate $Z^{\vec{j}\cdot \vec{i}}$ (note that $Z^2=\id$). Since such gate does not change the $Z$-basis outcomes in the PE rounds, we can assume, without loss of generality, that party $A_i$ applies the gate $Z^{\vec{j}\cdot \vec{i}}$ before measuring their qubit in any basis. If we now apply the gates in the post-selected state of the qubits \eqref{postselect_state}, we obtain:
\begin{align}\label{postselect_state2}
    |W_j\rangle_{{Q_0,\dots} Q_{N-1}} &= \frac{1}{\sqrt{N}}\sum_{i=0}^{N-1} (-1)^{\vec{j}\cdot \vec{i}} \bigotimes_{k=0}^{N-1} Z^{\vec{j}\cdot \vec{k}} \ket{\vec{b}_i}_{{Q_0,\dots} Q_{N-1}}  \nonumber\\
    &= \frac{1}{\sqrt{N}}\sum_{i=0}^{N-1} (-1)^{\vec{j}\cdot \vec{i}} (-1)^{\vec{j}\cdot \vec{i}}  \ket{\vec{b}_i}_{{Q_0,\dots} Q_{N-1}} \nonumber\\
    &= \frac{1}{\sqrt{N}}\sum_{i=0}^{N-1} \ket{\vec{b}_i}_{{Q_0,\dots} Q_{N-1}}\,,
\end{align}
where we used the fact that the operator $Z^{\vec{j}\cdot \vec{k}}$ has no effect on the ket $\ket{\vec{b}_i}$ except for $k=i$, i.e. when it acts on the $i$-th qubit that is in state $\ket{1}$. From \eqref{postselect_state2} we see that the post-selected state, after the local operations that simulate the classical post-processing of the outcomes, coincides with the W state, as claimed. The Ideal protocol presents a crucial difference from the protocol in \cite{Grasselli2019}, which is made explicit in the following remark:

\begin{remark}
In \cite{Grasselli2019} the parties needed to tailor their KG measurements depending on which detector clicks, in order to neutralize the effects of complex phases in their post-selected $W$-type state. In the Ideal protocol, thanks to the bespoke BBS network, the post-selected state \eqref{postselect_state} only presents real phases, which are corrected as discussed above by simply flipping the KG outcomes and without changing the measurement basis.
\end{remark}

This implies that, in the Ideal protocol, the parties' measurements are independent of the relay's announcements, hence they commute with the action of the relay. This enables us to reformulate the Ideal protocol in prepare-and-measure (PM) form. In the resulting PM protocol, the parties first measure their qubits and record the outcome. Then they send the optical mode, whose state is conditioned on the outcome, to the relay. Hence, the PM protocol coincides with the Ideal protocol except for Step~1.1.\\

\begin{protocol} \label{prot:PM}
		\caption{Prepare-and-measure (PM) protocol}
\begin{enumerate}
\item Quantum part: repeat what follows for a sufficient amount of iterations.
\begin{enumerate}[label*=\arabic*.]
    \item Each party $A_i$ prepares an optical mode $a_i$ in a state that depends on whether the round is labelled as a PE or KG round. 
    \begin{itemize}
        \item In a PE round, they prepare the vacuum state $|0\rangle_{a_i}$ with probability $q_i$, corresponding to the outcome $Z_i=+1$, and the one-photon state $|1\rangle_{a_i}$ with probability $1-q_i$, corresponding to the outcome $Z_i=-1$. 
        \item In a KG round, they prepare with equal probability either the state $|+\rangle_{a_i}=\sqrt{q_i} |0 \rangle_{a_i} + \sqrt{1-q_i} |1 \rangle_{a_i} $, corresponding to the outcome $X_{i}=+1$, or the state $|-\rangle_{a_i}=\sqrt{q_i} |0 \rangle_{a_i} - \sqrt{1-q_i} |1 \rangle_{a_i}$, corresponding to the outcome $X_{i}=-1$.
    \end{itemize}
    \item same as in Ideal prot.
    \item same as in Ideal prot.
    \item same as in Ideal prot.
\end{enumerate}
    \item same as in Ideal prot.
    \item same as in Ideal prot.
\end{enumerate}
\end{protocol}\vspace{0.5cm}

Note that, while the PM protocol is more practical than the Ideal protocol (e.g. it does not require qubit-photon entanglement), it is equivalent to the latter from the point of view of security, since an adversary could not distinguish which of the two protocols is run. Despite the increased practicality, the PM protocol still requires the preparation of single-photon states and their superposition with the vacuum. This prompts us to reduce even further the complexity of the protocol's implementation and obtain a practical, prepare-and-measure, CKA protocol.

In order to derive a practical CKA protocol, we observe that the states prepared in the KG rounds of the PM protocol ($\ket{\pm}_{a_i}$) can be approximated by coherent states of suitable amplitude ($\ket{\pm \alpha_i}$, for $\alpha_i\in\mathbbm{R}$), where the information about the $X$-basis outcome is encoded in the amplitude's sign. At the same time, the statistics collected in PE rounds and used to compute the phase error rate \eqref{phase_error_rate-app} are linked to the so-called yields, i.e. the probability that a detector clicks given that each party sent a fixed number of photons. This suggests us to prepare phase-randomized coherent states in PE rounds and use their detection statistics to apply the decoy-state method and compute the yields, with which we bound the phase error rate. This heuristic reasoning leads us to the practical CKA protocol presented in Sec.~\ref{sec:theprotocol}, where, we recall, each party $A_i$ prepares a coherent state $|x_i \alpha_i\rangle_{a_i}$ with $x_i = \pm 1$ in KG rounds and phase-randomized coherent states in PE rounds.

We emphasize that, in the protocol of Sec.~\ref{sec:theprotocol}, the choice of the type of round (KG or PE) cannot be delayed until after the action of the untrusted relay, contrary to the Ideal protocol. Indeed, the average state prepared by $A_i$ in KG rounds, $(1/2)(\ket{\alpha_i}\bra{\alpha_i}+\ket{-\alpha_i}\bra{-\alpha_i})$, differs from the average state prepared in PE rounds, $(1/|\mathcal{S}_i|) \sum_k \rho_{a_i}(\beta_k)$, due to the coherences of the former in the Fock basis. This means that an adversary controlling the relay could partially distinguish the type of round being executed and act accordingly. Another way to see this is that there is no equivalent entanglement-based version of the CKA protocol. That is, party $A_i$ cannot find two suitable POVMs (one for KG rounds and one for PE rounds) such that the state of their optical mode, conditioned on measuring with one of the two POVMs a fictitious system entangled with the optical mode, corresponds to the state that $A_i$ should prepare in that round \cite{TL17}.

One of the implications of the above fact is that the phase error rate \eqref{phase_error_rate-app} affecting the KG rounds cannot be directly observed from the statistics of the PE rounds, as instead happens in the Ideal protocol. Nevertheless, in the security proof provided in Sec.~\ref{sec:security-proof}, we show how to use the PE statistics to derive an upper bound on the phase error rate \eqref{phase_error_rate-app}. Specifically, we develop a multipartite decoy-state method that allows us to bound certain yields through the PE statistics. The yields, in turn, are needed to analytically upper bound the phase error rate with \eqref{phase-error-rate-bound}.

This concludes our connection between the Ideal protocol, which manifestly makes use of W-state correlations and whose phase error rate can be directly observed in PE rounds, and the CKA protocol discussed in the main text, whose phase error rate is bounded by PE statistics combined with the decoy-state method.

\section{Analytical upper bound on the yields}\label{section:yieldbound}

In this Appendix, we report the full derivation of the analytical bounds on the yields as a function of the observed gains, \eqref{upbound}. The bounds are derived with a multipartite decoy-state method in which each party is provided with the same set of two decoy intensities: $\mathcal{S}=\{\beta_0,\beta_1\}$.

We recall that the gains are probabilities that can be directly estimated from the observed data and are defined as: $G^j_{\vec{f}}:= \Pr\left( \Omega_j | \beta_{f_0},\dots,\beta_{f_{N-1}} \right)$, where $\vec{f}$ is an $N$-dimensional binary vector that covers all the possible choices of intensities by the parties. From Eq.~\eqref{PEprob} in Sec.~\ref{sec:security-proof}, we showed that the gains are related to the yields by the following equality:
\begin{align}
    G^j_{\vec{f}} &= \sum_{n_0,\dots,n_{N-1}=0}^{\infty}  Y^j_{n_0,\dots,n_{N-1}} \prod_{i=0}^{N-1}P_{\beta_{f_i}}(n_{i}) \nonumber\\
    &=\sum_{n_0,\dots,n_{N-1}=0}^{\infty}  Y^j_{n_0,\dots,n_{N-1}} \prod_{i=0}^{N-1}\frac{e^{-\beta_{f_i}}\beta_{f_i}^{n_{i}}}{n_{i}!} \nonumber\\
    &= \prod_{i=0}^{N-1}e^{-\beta_{f_i}} \sum_{n_0,\dots,n_{N-1}=0}^{\infty}  \frac{Y^j_{n_0,\dots,n_{N-1}}}{n_0!\dots n_{N-1}!} \prod_{i=0}^{N-1}\beta_{f_i}^{n_{i}}. \label{gain-app}
\end{align}
We remark that, in principle, the gains can depend on the detector $D_j$ that clicks and so can the yields. However, for simplicity of notation, in this section we drop the superscript $^j$ from the gains and yields. Moreover, in our simulations, due to the symmetric losses affecting each party, the gains and hence the yields are independent of which detector clicks (see Appendix~\ref{sec:channelmodel}). Hence their dependency on $j$ vanishes.

The last equality in \eqref{gain-app} brings us to define a rescaled gain, $\tilde{G}$, as follows:
\begin{equation}\label{gain}
   \tilde{G}_{\vec{f}} := \frac{G_{\vec{f}}}{\prod_{i=0}^{N-1}e^{-\beta_{f_i}}}=\sum_{n_0,\dots,n_{N-1}=0}^{\infty}  \frac{Y_{n_0,\dots,n_{N-1}}}{n_0!\dots n_{N-1}!} \prod_{i=0}^{N-1}\beta_{f_i}^{n_{i}} 
\end{equation}
We now define, from a fixed binary vector $\vec{h}$ of dimension $N$ and Hamming weight $|\vec{h}|=m$, the following quantity
\begin{align}
    B_{\vec{h}} & := \sum_{n_0,\dots,n_{N-1}=0}^{\infty} \frac{Y_{n_0,\dots,n_{N-1}}}{n_0!\dots n_{N-1}!} \prod_{i=0}^{N-1}\left(\beta_1^{1-h_i}\beta_0^{n_{i}}-\beta_0^{1-h_i}\beta_1^{n_{i}}\right) \label{Bh},
\end{align}
which can be recast as a combination of rescaled gains $\tilde{G}_{\vec{f}}$. To see this, we expand the product over $i$ in the last expression as a sum of $2^N$ products, each labelled by a binary vector $\vec{f}$, where each term in the sum is the product of either $\beta_1^{1-h_i}\beta_0^{n_{i}}$ or $-\beta_0^{1-h_i}\beta_1^{n_{i}}$ for every $i=0,\dots,N-1$. In particular, $f_i=0$ ($f_i=1$) indicates that the former (latter) quantity was picked. With this in mind, we can write:    
\begin{align}
    \prod_{i=0}^{N-1}\left(\beta_1^{1-h_i}\beta_0^{n_{i}}-\beta_0^{1-h_i}\beta_1^{n_{i}}\right) &= \sum_{f=0}^{2^N-1} \prod_{i=0}^{N-1} \beta_{f_i}^{n_i} (-1)^{f_i} \beta_0^{(1-h_i)f_i}\beta_1^{(1-h_i)(1-f_i)} \nonumber\\
    &= \sum_{f=0}^{2^N-1} (-1)^{|\vec{f}|} \beta_0^{(\vec1-\vec{h})\cdot\vec{f}}\beta_1^{(\vec1-\vec{h})\cdot(\vec1-\vec{f})} \prod_{i=0}^{N-1} \beta_{f_i}^{n_i} .
\end{align}
By replacing the last expression in \eqref{Bh}, we can employ \eqref{gain} to directly relate $B_{\vec{h}}$ and $\tilde{G}_{\vec{f}}$. We obtain:
\begin{equation} \label{Bh-gains}
    B_{\vec{h}}=\sum_{f=0}^{2^N-1} (-1)^{|\vec{f}|} \beta_0^{(\vec1-\vec{h})\cdot\vec{f}}\beta_1^{(\vec1-\vec{h})\cdot(\vec1-\vec{f})} \tilde{G}_{\vec{f}}.
\end{equation}
This expression is important as is constitutes the link between the quantity $B_{\vec{h}}$, which in the following is used to bound the yields, and the observed gains.

By recasting \eqref{Bh} as follows:
\begin{align}
    B_{\vec{h}} = \sum_{n_0,\dots,n_{N-1}=0}^{\infty} \frac{Y_{n_0,\dots,n_{N-1}}}{n_0!\dots n_{N-1}!} \prod_{\underset{h_i=1}{i \; \mathrm{s.t.}}}(\beta_0^{n_{i}}-\beta_1^{n_{i}}) \prod_{\underset{h_i=0}{i \; \mathrm{s.t.}}}(\beta_1 \beta_0^{n_{i}}-\beta_0 \beta_1^{n_{i}}), \label{Bh2}
\end{align}
we notice that, whenever $h_i=1$, the coefficient of the yields $Y_{n_0,\dots,0_{i} \dots n_{N-1}}$ (i.e. with with $n_i=0$) is null, implying that they do not contribute to the value of $B_{\vec{h}}$. Similarly, when $h_i=0$, all yields of the form $Y_{n_0,\dots,1_{i} \dots n_{N-1}}$ are removed. With this observation, we can now obtain a non-trivial upper bound on any yield  $Y_{n_0,\dots,n_{N-1}}$ in terms of a certain combination of $B_{\vec{h}}$.

To do so, we recast \eqref{Bh2} as follows, where in the first term we only sum over the indexes $n_i$ that correspond to $h_i=1$ and set all the other photon indexes to zero, while in the second term we account for all the other possibilities:
\begin{align}
    B_{\vec{h}}  = &(-1)^{N-m}(\beta_0-\beta_1)^{N-m}\sum_{(n_0,\dots,n_{N-1})\in\mathcal{N}(\vec{h})} \frac{Y_{n_0,\dots,n_{N-1}}}{n_0!\dots n_{N-1}!} \prod_{\underset{h_i=1}{i \; \mathrm{s.t.}}}(\beta_0^{n_{i}}-\beta_1^{n_{i}}) \nonumber\\
    &+\sum_{(n_0,\dots,n_{N-1})\in\tilde{\mathcal{N}}(\vec{h})} \frac{Y_{n_0,\dots,n_{N-1}}}{n_0!\dots n_{N-1}!} \prod_{\underset{h_i=1}{i \; \mathrm{s.t.}}}(\beta_0^{n_{i}}-\beta_1^{n_{i}})\prod_{\underset{h_i=0}{i \; \mathrm{s.t.}}}(\beta_1 \beta_0^{n_{i}}-\beta_0 \beta_1^{n_{i}}) \label{Bh3},
\end{align}
where the sets of indexes $\mathcal{N}(\vec{h})$ and $\tilde{\mathcal{N}}(\vec{h})$ are defined as:
\begin{align}
    \mathcal{N}(\vec{h}) &:= \left\lbrace(n_0,\dots,n_{N-1}) : n_i=h_i r_i,\,  r_i \geq 1  \right\rbrace \\
    \tilde{\mathcal{N}}(\vec{h}) &:= \left\lbrace(n_0,\dots,n_{N-1}) : n_i \geq 1 \mbox{ (if $h_i=1$)}\,;\,n_i\geq 2 \mbox{ or } n_i=0 \mbox{ (if $h_i=0$)}  \right\rbrace \setminus \mathcal{N}(\vec{h})
\end{align}
where $k_i$ are integers. Note that the sum over $n_i$ in the second term skips the case $n_i=1$ for $h_i=0$ since this contribution is null in \eqref{Bh2} (see observation above).

We observe that the yields in the first sum in \eqref{Bh3} contain exactly $m$ nonzero photon numbers, which allowed us factor out the quantities $(\beta_1\beta_0^{n_i} -\beta_0 \beta_1^{n_i})$. Now, we split the second sum in \eqref{Bh3} in a sum of $N-m$ terms, where each term only contains yields with $m+k$ nonzero photon numbers, for $k=1,\dots,N-m$. In this way, we can factor out the quantities $(\beta_1\beta_0^{n_i} -\beta_0 \beta_1^{n_i})$ even in the second sum. This becomes important later, when we want to evaluate the sign in front of each yield. In order to sum over the various combinations of yields with $m+k$ photon numbers, we introduce the binary vectors $\vec{h}^{(k)}$, which can be seen as ``expansions'' of the vector $\vec{h}$ obtained by flipping $k$ of its zeros to ones. Thus, we have that $|\vec{h}^{(k)}|=m+k$ and that $h_i^{(k)}=1$ whenever $h_i=1$, which can be formally stated as the condition: $\vec{h}^{(k)}\wedge \vec{h}=\vec{h}$, where $\wedge$ is the bitwise AND operation. Analogously to $\vec{h}$, when $h_i^{(k)}=0$ we fix the corresponding photon number $n_i$ to zero. Then, in analogy with $\mathcal{N}(\vec{h})$, we define a set of indexes $\mathcal{N}(\vec{h},\vec{h}^{(k)})$ for each expansion $\vec{h}^{(k)}$ that represents the combinations of photon numbers that are allowed by the chosen vector $\vec{h}^{(k)}$:
\begin{align}
    \mathcal{N}(\vec{h},\vec{h}^{(k)}) := \left\lbrace (n_0,\dots,n_{N-1}) : n_i=h^{(k)}_i r_i,\,  r_i \geq 1+h_i^{(k)}-h_i\right\rbrace,
\end{align}
where we accounted for the fact that each additional new bit equal to one in $\vec{h}^{(k)}$, which is a zero in $\vec{h}$, corresponds to an index $n_i$ that starts from two instead of one. According to this, we obtain: 
\begin{align}
    B_{\vec{h}}  = &(-1)^{N-m}(\beta_0-\beta_1)^{N-m}\sum_{(n_0,\dots,n_{N-1})\in\mathcal{N}(\vec{h})} \frac{Y_{n_0,\dots,n_{N-1}}}{n_0!\dots n_{N-1}!} \prod_{\underset{h_i=1}{i \; \mathrm{s.t.}}}(\beta_0^{n_{i}}-\beta_1^{n_{i}}) \nonumber\\
    &+\sum_{k=1}^{N-m}  (-1)^{N-m-k} (\beta_0-\beta_1)^{N-m-k}(\beta_0\beta_1)^k \sum_{(n_0,\dots,n_{N-1})\in\mathcal{N}_k(\vec{h})} \frac{Y_{n_0,\dots,n_{N-1}}}{n_0!\dots n_{N-1}!} \prod_{\underset{h_i=1}{i \; \mathrm{s.t.}}}(\beta_0^{n_{i}}-\beta_1^{n_{i}})\prod_{\underset{h^{(k)}_i -h_i=1}{i \; \mathrm{s.t.}}}( \beta_0^{n_{i}-1}-\beta_1^{n_{i}-1}) \label{Bh4},
\end{align}
where we defined the following set that accounts for all possible choices of $\vec{h}^{(k)}$, for a given $k$:
\begin{align} \label{Nk}
    \mathcal{N}_k(\vec{h}):= \bigcup_{\substack{\vec{h}^{(k)}\in\{0,1\}^N: \\ |\vec{h}^{(k)}|=m+k \\ \vec{h}^{(k)} \wedge \vec{h} = \vec{h}}}  \mathcal{N}(\vec{h},\vec{h}^{(k)}),
\end{align}
where the operation $\wedge$ represents the entry-wise product. Now, we wish to isolate a specific yield $Y_{u_0, \dots u_{N-1}}$ from the first sum in \eqref{Bh4} in order to derive an upper bound on it. Note that we can choose any combination of photon numbers $(u_0, \dots u_{N-1})$ such that: $u_i=h_i k_i$, for $k_i \geq 1$. Since the choice of the vector $\vec{h}$ is arbitrary, the photon numbers are also arbitrary. By isolating the yield $Y_{u_0, \dots u_{N-1}}$ in \eqref{Bh4}, we obtain:
\begin{align}
    B_{\vec{h}}  =&(-1)^{N-m}(\beta_0-\beta_1)^{N-m} Y_{u_0,\dots,u_{N-1}} \prod_{\underset{h_i=1}{i \; \mathrm{s.t.}}}\frac{(\beta_0^{u_{i}}-\beta_1^{u_{i}})}{u_i!} \nonumber\\
    &+(-1)^{N-m}(\beta_0-\beta_1)^{N-m}\sum_{(n_0,\dots,n_{N-1})\in\mathcal{N}(\vec{h})\setminus\{(u_0,\dots,u_{N-1})\}} \frac{Y_{n_0,\dots,n_{N-1}}}{n_0!\dots n_{N-1}!} \prod_{\underset{h_i=1}{i \; \mathrm{s.t.}}}(\beta_0^{n_{i}}-\beta_1^{n_{i}}) \nonumber\\
    &+\sum_{k=1}^{N-m}  (-1)^{N-m-k} (\beta_0-\beta_1)^{N-m-k}(\beta_0\beta_1)^k  \sum_{(n_0,\dots,n_{N-1})\in\mathcal{N}_k(\vec{h})} \frac{Y_{n_0,\dots,n_{N-1}}}{n_0!\dots n_{N-1}!} \prod_{\underset{h_i=1}{i \; \mathrm{s.t.}}}(\beta_0^{n_{i}}-\beta_1^{n_{i}})\prod_{\underset{h^{(k)}_i -h_i=1}{i \; \mathrm{s.t.}}}( \beta_0^{n_{i}-1}-\beta_1^{n_{i}-1}) \label{Bh5},
\end{align}
We now derive an upper bound on $Y_{u_0, \dots u_{N-1}}$. To this aim, we observe that the yield $Y_{u_0, \dots u_{N-1}}$ and each of the yields in the second term in \eqref{Bh5} are multiplied by coefficients of the same sign. Indeed, they are multiplied by the same number of terms of the form $(\beta_0^s-\beta_1^s)$. More quantitatively, the sign of the coefficient $C_{u_0, \dots u_{N-1}}$ of $Y_{u_0, \dots u_{N-1}}$ and of the coefficients of the yields in the second term is:
\begin{align}
    \mathrm{sign}(C_{u_0, \dots u_{N-1}}) = (-1)^{N-m} \left[\mathrm{sign}(\beta_0 -\beta_1)\right]^N.
\end{align}
By similar arguments, the yields in the third term in \eqref{Bh5} are multiplied by coefficients $C_{n_0, \dots n_{N-1}}$ with the following sign:
\begin{align}
    \mathrm{sign}(C_{n_0, \dots n_{N-1}}) = (-1)^{N-m-k} \left[\mathrm{sign}(\beta_0 -\beta_1)\right]^N.
\end{align}
In order to extract an upper bound on $Y_{u_0, \dots u_{N-1}}$, we need to minimize all the yields carrying the same sign as $Y_{u_0, \dots u_{N-1}}$ and maximize all the yields with opposite sign in \eqref{Bh5}. In our case, this means setting to zero all the yields in the first sum and all the yields in the second sum that correspond to even values of $k$. The other yields are set to one. By applying this reasoning to \eqref{Bh5}, we obtain the following expression satisfied by an upper bound  $U_{u_0, \dots u_{N-1}}$ on $Y_{u_0, \dots u_{N-1}}$:
\begin{align}
    B_{\vec{h}}  =&(-1)^{N-m}(\beta_0-\beta_1)^{N-m} U_{u_0,\dots,u_{N-1}} \prod_{\underset{h_i=1}{i \; \mathrm{s.t.}}}\frac{(\beta_0^{u_{i}}-\beta_1^{u_{i}})}{u_i!} \nonumber\\
    &+\sum_{\substack{k=1 \\ k \,\mathrm{odd}}}^{N-m}  (-1)^{N-m-k} (\beta_0-\beta_1)^{N-m-k}(\beta_0\beta_1)^k \sum_{(n_0,\dots,n_{N-1})\in\mathcal{N}_k(\vec{h})} \frac{1}{n_0!\dots n_{N-1}!} \prod_{\underset{h_i=1}{i \; \mathrm{s.t.}}}(\beta_0^{n_{i}}-\beta_1^{n_{i}})\prod_{\underset{h^{(k)}_i -h_i=1}{i \; \mathrm{s.t.}}}( \beta_0^{n_{i}-1}-\beta_1^{n_{i}-1}) \label{Bh6}.
\end{align}
In order to simplify the above expression, we first focus on the term with the sum over $k$, which we denote $B^{(2)}_{\vec{h}}$ and recast as follows:
\begin{align}
    B^{(2)}_{\vec{h}} &=\sum_{\substack{k=1 \\ k \,\mathrm{odd}}}^{N-m}  (-1)^{N-m-k} (\beta_0-\beta_1)^{N-m-k}(\beta_0\beta_1)^k \sum_{(n_0,\dots,n_{N-1})\in\mathcal{N}_k(\vec{h})} \prod_{\underset{h_i=1}{i \; \mathrm{s.t.}}}\frac{\beta_0^{n_{i}}-\beta_1^{n_{i}}}{n_i!}\prod_{\underset{h^{(k)}_i -h_i=1}{i \; \mathrm{s.t.}}}\frac{\beta_0^{n_{i}-1}-\beta_1^{n_{i}-1}}{n_i!} \nonumber\\
    &= \sum_{\substack{k=1 \\ k \,\mathrm{odd}}}^{N-m}  (-1)^{N-m-k} (\beta_0-\beta_1)^{N-m-k}(\beta_0\beta_1)^k \sum_{\substack{\vec{h}^{(k)}\in\{0,1\}^N: \\ |\vec{h}^{(k)}|=m+k \\ \vec{h}^{(k)} \wedge \vec{h} = \vec{h}}} \sum_{(n_0,\dots,n_{N-1})\in\mathcal{N}(\vec{h},\vec{h}^{(k)})} \prod_{\underset{h_i=1}{i \; \mathrm{s.t.}}}\frac{\beta_0^{n_{i}}-\beta_1^{n_{i}}}{n_i!}\prod_{\underset{h^{(k)}_i -h_i=1}{i \; \mathrm{s.t.}}}\frac{\beta_0^{n_{i}-1}-\beta_1^{n_{i}-1}}{n_i!}\label{B2h},
\end{align}
where in the second equality we split the second sum over all the different subsets $\mathcal{N}(\vec{h},\vec{h}^{(k)})$ in $\mathcal{N}_k(\vec{h})$ using \eqref{Nk}. We can now swap the innermost sum in the last expression with the products and obtain:
\begin{align}
    B^{(2)}_{\vec{h}} = &\sum_{\substack{k=1 \\ k \,\mathrm{odd}}}^{N-m}  (-1)^{N-m-k} (\beta_0-\beta_1)^{N-m-k}(\beta_0\beta_1)^k \nonumber\\
    &\cdot\sum_{\substack{\vec{h}^{(k)}\in\{0,1\}^N: \\ |\vec{h}^{(k)}|=m+k \\ \vec{h}^{(k)} \wedge \vec{h} = \vec{h}}}  \prod_{\underset{h_i=1}{i \; \mathrm{s.t.}}} \left(\sum_{n_i=1}^{\infty} \frac{\beta_0^{n_{i}}-\beta_1^{n_{i}}}{n_i!}\right)\prod_{\underset{h^{(k)}_i -h_i=1}{i \; \mathrm{s.t.}}}\left(\sum_{n_i=2}^\infty \frac{\beta_0^{n_{i}-1}-\beta_1^{n_{i}-1}}{n_i!}\right)\label{B2h2}.
\end{align}
It can now be easily seen, using the Taylor series of the exponential function, that the following identities hold:
\begin{align}
    \sum_{n=1}^{\infty} \frac{\beta_0^n-\beta_1^n}{n!}&=e^{\beta_0}-e^{\beta_1} ,\\
    \sum_{n=2}^{\infty} \frac{\beta_0^{n-1}-\beta_1^{n-1}}{n!}&=\frac{1}{\beta_0 \beta_1}\left( \beta_1 e^{\beta_0}- \beta_0 e^{\beta_1} + \beta_0 -\beta_1 \right).
\end{align}
By using the above identities in \eqref{B2h2}, we obtain:
\begin{align}
    B^{(2)}_{\vec{h}} = &\sum_{\substack{k=1 \\ k \,\mathrm{odd}}}^{N-m}  (-1)^{N-m-k} (\beta_0-\beta_1)^{N-m-k} \sum_{\substack{\vec{h}^{(k)}\in\{0,1\}^N: \\ |\vec{h}^{(k)}|=m+k \\ \vec{h}^{(k)} \wedge \vec{h} = \vec{h}}}  (e^{\beta_0}-e^{\beta_1})^m \left( \beta_1 e^{\beta_0}- \beta_0 e^{\beta_1} + \beta_0 -\beta_1 \right)^k \label{B2h3},
\end{align}
where we observe that the argument of the sum over $\vec{h}^{(k)}$ is independent of $\vec{h}^{(k)}$. Therefore, the sum reduces to counting all the possible choices of $\vec{h}^{(k)}$ for a given $k$. This number is given by the possible combinations of $k$ bits in $\vec{h}^{(k)}$ that are set to one, chosen among the $N-m$ elements that correspond to zeroes in $\vec{h}$. Hence, we have $\binom{N-m}{k}$ choices and we obtain:
\begin{align}
    B^{(2)}_{\vec{h}} &= \sum_{\substack{k=1 \\ k \,\mathrm{odd}}}^{N-m}  (-1)^{N-m-k} (\beta_0-\beta_1)^{N-m-k} \binom{N-m}{k}  (e^{\beta_0}-e^{\beta_1})^m \left( \beta_1 e^{\beta_0}- \beta_0 e^{\beta_1} + \beta_0 -\beta_1 \right)^k \nonumber\\
    &= \sum_{k=0}^{\left\lfloor \frac{N-m-1}{2} \right\rfloor} (-1)^{N-m-2k-1} (\beta_0-\beta_1)^{N-m-2k-1}  \binom{N-m}{2k+1} \left( e^{\beta_0}-e^{\beta_1}\right)^m \left( \beta_1 e^{\beta_0}- \beta_0 e^{\beta_1} + \beta_0 -\beta_1 \right)^{2k+1} \label{B2h4},
\end{align}
where $\lfloor x\rfloor$ is the floor function. Finally, by employing \eqref{B2h4} in \eqref{Bh6}, we obtain the following equality satisfied by the upper bound on the selected yield:
\begin{align}
    B_{\vec{h}}  =&(-1)^{N-m}(\beta_0-\beta_1)^{N-m} U_{u_0,\dots,u_{N-1}} \prod_{\underset{h_i=1}{i \; \mathrm{s.t.}}}\frac{(\beta_0^{u_{i}}-\beta_1^{u_{i}})}{u_i!} \nonumber\\
    &+\sum_{k=0}^{\left\lfloor \frac{N-m-1}{2} \right\rfloor} (-1)^{N-m-2k-1} (\beta_0-\beta_1)^{N-m-2k-1}  \binom{N-m}{2k+1} \left( e^{\beta_0}-e^{\beta_1}\right)^m \left( \beta_1 e^{\beta_0}- \beta_0 e^{\beta_1} + \beta_0 -\beta_1 \right)^{2k+1} \label{Bh7}.
\end{align}
By isolating the yield's upper bound and relabelling $u_i\to n_i$, we obtain the final expression of the yield bound: $\overline{Y}_{n_0,\dots,n_{N-1}}=\min\{U_{n_0,\dots,n_{N-1}},1\}$, where:
\begin{align}
    U_{n_0,\dots,n_{N-1}} = \prod_{\underset{n_i\neq 0}{i \; \mathrm{s.t.}}}\frac{n_i!}{\beta_0^{n_{i}}-\beta_1^{n_{i}}}\left[\frac{B_{\vec{h}}\,\,(-1)^{N-m}}{(\beta_0-\beta_1)^{N-m}} + \left( e^{\beta_0}-e^{\beta_1}\right)^m \sum_{k=0}^{\left\lfloor \frac{N-m-1}{2} \right\rfloor}\binom{N-m}{2k+1}\left(\frac{\beta_1 e^{\beta_0}- \beta_0 e^{\beta_1} + \beta_0 -\beta_1}{\beta_0-\beta_1}\right)^{2k+1}\right] \label{upbound-app},
\end{align}
where $B_{\vec{h}}$ is given in \eqref{Bh-gains} (and in principle can depend on the detector $D_j$ through the gains) and $m=|\vec{h}|$, while $\vec{h}$ is the binary vector with components $h_i$ defined by:
\begin{align}
    h_i = \left\lbrace \begin{array}{ll}
      1   & \mbox{if }n_i\geq 1  \\
    0  &  \mbox{if } n_i=0.
    \end{array}\right.
\end{align}

\section{Channel model}\label{sec:channelmodel}

In this Appendix we describe our channel model and compute the detection statistics of the protocol. The channel model includes the following sources of noise:
\begin{enumerate}
    \item Pure-photon loss: the optical mode of party goes through the same lossy channel. The lossy channel is modelled with a beam splitter with transmittance $\eta$, where the additional input port of the beam splitter is fed with the vacuum. 
    \item Polarization misalignment: the optical mode of each party undergoes a polarization misalignment modelled by a unitary operation that maps the creation operator of each mode according to 
    \begin{equation}\label{polarization}
         \hat{a}^\dag_i \rightarrow \cos{\theta_i}\hat{a}^\dag_{i,P}-\sin{\theta_i}\hat{a}^\dag_{i,P_{\perp}}
    \end{equation}
    where $\hat{a}^\dag_{i,P}$ is the creation operator on the original polarization and $\hat{a}^\dag_{i,P_{\perp}}$ is the creation operator on the orthogonal polarization.
    \item Phase shift: the optical mode of each party undergoes a phase shift $\phi_i$, modelled by multiplying the mode operator $\hat{a}^\dag_i$ by a phase $\phi_i$. 
    \item Dark counts in the detectors: each detector is affected by dark counts, with a probability $p_d$ that is equal for all detectors and independent on the state sent.
\end{enumerate}

In Sec.~\ref{sec:simulations} we argued that since the channel of each party is equally lossy, the optimal choice for the signal intensities is the same for each party. Hence, here we assume that the amplitudes of each party in KG and PE rounds coincide: $\alpha_i=\alpha$ and $\mathcal{S}_i=\mathcal{S}$, for every $i$. Moreover, we choose the same polarization misalignment between the reference party $A_0$ and each other party. This means that we choose a misalignment of $\theta_0$ for $A_0$ and $\theta_1$ for the other parties. Similarly for the phase shift, we set $\phi_0=0$ and $\phi_i=\phi$  for $i \neq 0$.\\

\textbf{Computation of $\Pr(\Omega_j|x_0,x_1,\dots,x_{N-1},\mathrm{KG})$}\\
We start by computing the detection probability $\Pr(\Omega_j|x_0,x_1,\dots,x_{N-1},\mathrm{KG})$, which is the probability that only detector $D_j$ clicks, given that party $A_i$ prepared, in a KG round, the coherent state $\ket{x_i\alpha}$, with $x_i=\pm 1$. This detection probability is needed to compute the QBER \eqref{QBER} through \eqref{QBERcomp}.

The state prepared by the $N$ parties in a KG round, before any noise or loss is applied, reads:
\begin{equation}
    |\psi_{in} \rangle = \bigotimes_{i=0}^{N-1} |x_i \alpha \rangle .
\end{equation}
We now apply the sources of noise discussed above.
\begin{enumerate}
    \item The resulting state after the lossy channel is the following:
    \begin{equation}
    |\psi'_{in} \rangle = \bigotimes_{i=0}^{N-1} | x_i \sqrt{\eta} \alpha \rangle .
\end{equation}
    \item After applying the polarization misalignment, we obtain:
    \begin{equation}
    |\psi''_{in} \rangle = \bigotimes_{i=0}^{N-1} \left| x_i \cos{\theta_i} \sqrt{\eta} \alpha \right\rangle_{P} \left| -x_i \sin{\theta_i} \sqrt{\eta} \alpha \right\rangle_{P_{\perp}} .
    \end{equation}
    \item After the phase shift $\phi_i$ is applied on each mode, we get:
     \begin{equation}
    |\psi'''_{in} \rangle = \bigotimes_{i=0}^{N-1} \left| x_i \cos{\theta_i}e^{i\phi_i} \sqrt{\eta} \alpha \right\rangle_{P} \left| -x_i \sin{\theta_i} e^{i\phi_i} \sqrt{\eta} \alpha \right\rangle_{P_{\perp}} . \label{psi3}
    \end{equation}
\end{enumerate}
The state in \eqref{psi3} is the global state of the $N$ parties' modes, after the noisy and lossy channel and before entering the BBS network. We now evolve the modes through the $M$-input and $M$-output BBS network, according to the transformation in \eqref{transf}. We define the coefficients of the inverse transformation of the modes as: $f_{i,j}:=(-1)^{-\vec{i}\cdot\vec{j}}$. Here, we make the non-restrictive assumption that the $N$ modes sent by the parties correspond to the first $N$ inputs of the BBS network. A different choice would not alter the protocol's performance. The output state after the BBS network reads:
\begin{equation}
    |\psi_{out} \rangle = \bigotimes_{j=0}^{M-1} \left| \sqrt{\frac{\eta}{M}} \alpha \sum_{i=0}^{N-1} x_i f_{i,j} \cos{\theta_i}e^{i\phi_i} \right\rangle_{P} \left| -\sqrt{\frac{\eta}{M}} \alpha \sum_{i=0}^{N-1} x_i f_{i,j} \sin{\theta_i}e^{i\phi_i} \right\rangle_{P_{\perp}} .
\end{equation}
At this point, the relay performs a threshold measurement on each mode that returns a click in the corresponding detector if one or more photons are detected. We are interested in the probability that only detector $D_j$ clicks, i.e. $\Pr(\Omega_j|x_0,x_1,\dots,x_{N-1},\mathrm{KG})$. By including the effect of dark counts, we can express such probability as follows:
\begin{align}
    \Pr(\Omega_j|x_0,x_1,\dots,x_{N-1},\mathrm{KG}) & = p_d(1-p_d)^{M-1}\mbox{Tr}\left[ \rho_{out} \bigotimes_{k=0}^{M-1} \ketbra{0}{0}_k \right] + (1-p_d)^{M-1} \mbox{Tr}\left[ \rho_{out} (\mathds{1}_j-\ketbra{0}{0}_j) \bigotimes_{k\neq j} \ketbra{0}{0}_k \right] \nonumber \\
    & = (1-p_d)^{M-1} \mbox{Tr}\left[ \rho_{out} \mathds{1}_j \bigotimes_{k\neq j} \ketbra{0}{0}_k \right] - (1-p_d)^M\mbox{Tr}\left[ \rho_{out} \bigotimes_{k=0}^{M-1} \ketbra{0}{0}_k \right] \label{probDjKG}
\end{align}
where $\ketbra{0}{0}_k $ is the projector on the vacuum of the output mode $k$ for polarizations $P$ and $P_\perp$, as the detectors do not distinguish polarization, and $\rho_{out}=|\psi_{out} \rangle \langle \psi_{out}|$. We calculate both terms appearing in \eqref{probDjKG}. For the second term, we have:
\begin{align}
    \mbox{Tr}\left[ \rho_{out} \bigotimes_{k=0}^{M-1} \ketbra{0}{0}_k\right] & = \prod_{k=0}^{M-1} \exp\left[ -\left| \sqrt{\frac{\eta}{M}}\alpha \sum_{i=0}^{N-1} x_i f_{i,k} \cos{\theta_i}e^{i\phi_i} \right|^2 - \left| \sqrt{\frac{\eta}{M}}\alpha \sum_{i=0}^{N-1} x_i f_{i,k} \sin{\theta_i}e^{i\phi_i} \right|^2\right] \nonumber \\
    & =  \exp\left[-\frac{\eta}{M} \alpha^2 \sum_{k=0}^{M-1} \left( \left| \sum_{i=0}^{N-1} x_i f_{i,k} \cos{\theta_i}e^{i\phi_i} \right|^2 + \left| \sum_{i=0}^{N-1} x_i f_{i,k} \sin{\theta_i}e^{i\phi_i} \right|^2\right) \right]. \label{secondterm}
\end{align}
We now focus on the sum over $k$ in the last expression and use the fact that we fixed the angles $\theta_i$ and $\phi_i$ as discussed above. The sum over $k$ simplifies to: 
\begin{align}
     \sum_{k=0}^{M-1} &\left( \left| \sum_{i=0}^{N-1} x_i f_{i,k} \cos{\theta_i}e^{i\phi_i} \right|^2 + \left| \sum_{i=0}^{N-1} x_i f_{i,k} \sin{\theta_i}e^{i\phi_i} \right|^2\right)  \nonumber \\
      & = \sum_{k=0}^{M-1}  \left(  \left|x_0 \cos{\theta_0} + \cos{\theta_1}e^{i\phi}\sum_{i=1}^{N-1} x_i f_{i,k}  \right|^2+\left|x_0 \sin{\theta_0} + \sin{\theta_1}e^{i\phi}\sum_{i=1}^{N-1} x_i f_{i,k}  \right|^2 \right) \nonumber \\
      & = \sum_{k=0}^{M-1}  \left(  x_0^2 \cos^2{\theta_0} + \cos^2{\theta_1}\left|\sum_{i=1}^{N-1} x_i f_{i,k}  \right|^2+ x_0^2 \sin^2{\theta_0} + \sin^2{\theta_1}\left|\sum_{i=1}^{N-1} x_i f_{i,k}\right|^2 \right. \nonumber \\
        & + \left. 2 x_0 \cos{\theta_0}\cos{\theta_1} \cos{\phi} \sum_{i=1}^{N-1} x_i f_{i,k} + 2  x_0 \sin{\theta_0}\sin{\theta_1} \cos{\phi} \sum_{i=1}^{N-1} x_i f_{i,k} \right) \nonumber \\
   &  =  \sum_{k=0}^{M-1} \left( 1 + \left|\sum_{i=1}^{N-1} x_i f_{i,k}\right|^2 + 2 x_0 \cos{\theta} \cos{\phi}\sum_{i=1}^{N-1} x_i f_{i,k}  \right) \label{prob-calc1}
\end{align}
where we used that $x_i=\pm 1$, $\cos{\theta_i}^2+\sin{\theta_i}^2=1$ and $\cos{\theta_0}\cos{\theta_1}+\sin{\theta_0}\sin{\theta_1}=\cos(\theta_0-\theta_1)$ and where we defined $\theta:=\theta_0-\theta_1$. Consider the following lemma for the function $f_{i,k}$, which coincides with $f_{k,i}$ in \eqref{func}. 
\begin{lemma} \label{lemma:fki}
 For $f_{k,i}$ as defined by \eqref{func}, it holds: 
 \begin{equation}
     \sum_{i=0}^{M-1} f_{k,i}=M \,\delta_{k,0} 
\end{equation}
\end{lemma}

\begin{proof}
    To show the result of the lemma we first recall that the function $f_{k,i}$ is given by:
    \begin{eqnarray}
        f_{k,i}=(-1)^{\vec{k}\cdot \vec{i}},
    \end{eqnarray}
    where $\vec{k}$ and $\vec{i}$ are the binary vectors of length $s$ that represent the numbers $k$ and $i$ in binary representation. Then, the sum over $i$ of $f_{k,i}$ can be recast as:
    \begin{align}
    \sum_{i=0}^{M-1} f_{k,i}    &= \sum_{\vec{i}\in\{0,1\}^s} (-1)^{\vec{k}\cdot \vec{i}} \nonumber\\
        &= \sum_{c=0}^{\abs{\vec{k}}} (-1)^c\binom{|\vec{k}|}{c}2^{s-\abs{\vec{k}}}, \label{lemma-step}
    \end{align}
    where in the second equality we perform the sum over all the possible values $c$ of $\vec{k}\cdot \vec{i}$ and count how many distinct vectors $\vec{i}$ lead to the same scalar product $c=\vec{k}\cdot \vec{i}$. This number is given by the ways in which we can select $c$ bits equal to one in $\vec{k}$ (the binomial coefficient), which fixes the corresponding $c$ bits in $\vec{i}$ to be one and also fixes other $|\vec{k}|-c$ bits in $\vec{i}$ to be zero since they correspond to the ones in $\vec{k}$ that have not been selected. At this point, the vector $\vec{i}$ is almost all fixed, except for the bits that correspond to the $s-|\vec{k}|$ zero bits in $\vec{k}$. Since such bits in $\vec{i}$ can be arbitrary as they would not contribute to the scalar product, the total number of possibilities is given by $2^{s-|\vec{k}|}$.
    
    Now, we can simplify the expression in \eqref{lemma-step} as follows:
    \begin{align}
        \sum_{i=0}^{M-1} f_{k,i} &= \sum_{c=0}^{\abs{\vec{k}}} (-1)^c\binom{|\vec{k}|}{c}2^{s-\abs{\vec{k}}} \nonumber\\
        &= 2^{s-\abs{\vec{k}}} \sum_{c=0}^{\abs{\vec{k}}} \binom{|\vec{k}|}{c} (-1)^c (1)^{|\vec{k}|-c} \nonumber\\
        &= 2^{s-\abs{\vec{k}}} (1-1)^{|\vec{k}|} \nonumber\\
        &= M \delta_{k,0},
    \end{align}
    where we used the binomial formula in the third line and that $M=2^s$ together with the definition of Kronecker delta in the last line. This concludes the proof.
\end{proof}

By applying Lemma~\ref{lemma:fki} in \eqref{prob-calc1}, we can simplify the term with the cosines as follows:
\begin{align}
     \sum_{k=0}^{M-1} \left( \left| \sum_{i=0}^{N-1} x_i f_{i,k} \cos{\theta_i}e^{i\phi_i} \right|^2 + \left| \sum_{i=0}^{N-1} x_i f_{i,k} \sin{\theta_i}e^{i\phi_i} \right|^2\right)  & =  \sum_{k=0}^{M-1} \left( 1 + \left|\sum_{i=1}^{N-1} x_i f_{i,k}\right|^2 + 2 x_0 \cos{\theta} \cos{\phi}\sum_{i=1}^{N-1} x_i f_{i,k}  \right) \nonumber\\
     &= \sum_{k=0}^{M-1} \left( 1 + \left|\sum_{i=1}^{N-1} x_i f_{i,k}\right|^2 \right) + 2 x_0 \cos{\theta} \cos{\phi}\sum_{i=1}^{N-1} x_i M \delta_{i,0} \nonumber\\
     &=\sum_{k=0}^{M-1} \left( 1 + \left|\sum_{i=1}^{N-1} x_i f_{i,k}\right|^2 \right) \nonumber\\
     &=M + \sum_{k=0}^{M-1}\left|\sum_{i=1}^{N-1} x_i f_{i,k}\right|^2,   \label{prob-calc2}
\end{align}
where the sum with the Kronecker delta $\delta_{i,0}$ is identically zero since the index $i$ starts from one.

We now expand the square in the last expression and obtain:
\begin{align}
     \sum_{k=0}^{M-1} \left( \left| \sum_{i=0}^{N-1} x_i f_{i,k} \cos{\theta_i}e^{i\phi_i} \right|^2 + \left| \sum_{i=0}^{N-1} x_i f_{i,k} \sin{\theta_i}e^{i\phi_i} \right|^2\right)  &=M + \sum_{k=0}^{M-1}\left|\sum_{i=1}^{N-1} x_i f_{i,k}\right|^2 \nonumber\\
     &= M + \sum_{k=0}^{M-1} \sum_{i,i'=1}^{N-1}x_i x_{i'} f_{i,k}f_{i',k} \nonumber\\
     &= M +  \sum_{i,i'=1}^{N-1}x_i x_{i'} \sum_{\vec{k}\in\{0,1\}^s} (-1)^{(\vec{i}+\vec{i}')\cdot\vec{k}},   \label{prob-calc3}
\end{align}
where we remark that the result of Lemma~\ref{lemma:fki} cannot be directly applied to the innermost sum since $(\vec{i}+\vec{i}')$ is not a binary vector. However, we can use the Lemma to compute such a sum. In order to do so, we observe that the vector $(\vec{i}+\vec{i}')$ deviates from a binary vector only in the elements $r$ where $i_r =i'_r=1$, and we have $\vec{i}\cdot\vec{i}'$ many such elements. These elements do not contribute to the value of $(-1)^{(\vec{i}+\vec{i}')\cdot\vec{k}}$ regardless of the value of $k_r$. Hence, we can define shorter vectors $\vec{m}\in\{0,1\}^{s-\vec{i}\cdot\vec{i}'}$ and $\vec{l}\in\{0,1\}^{s-\vec{i}\cdot\vec{i}'}$ that correspond to the remaining $s-\vec{i}\cdot\vec{i}'$ elements of $\vec{i}+\vec{i}'$ and $\vec{k}$, respectively, where $i_r+i'_r\neq 2$. By definition, we have that: $(-1)^{(\vec{i}+\vec{i}')\cdot\vec{k}}=(-1)^{\vec{m}\cdot\vec{l}}$. Now, in order to replace the sum over $\vec{k}$ with a sum over $\vec{l}$, we must account for the fact that, for every fixed value of $\vec{l}$ and hence of $(-1)^{\vec{m}\cdot\vec{l}}$, there are $2^{\vec{i}\cdot\vec{i}'}$ vectors $\vec{k}$ such that $(-1)^{(\vec{i}+\vec{i}')\cdot\vec{k}}=(-1)^{\vec{m}\cdot\vec{l}}$. Therefore, we can recast the innermost sum in \eqref{prob-calc3} as follows:  
\begin{align}
    \sum_{\vec{k}\in\{0,1\}^s} (-1)^{(\vec{i}+\vec{i}')\cdot\vec{k}} &= 2^{\vec{i}\cdot\vec{i}'} \sum_{\vec{l}\in\{0,1\}^{s-\vec{i}\cdot\vec{i}'}} (-1)^{\vec{m}\cdot\vec{l}} \nonumber\\
    &= 2^{\vec{i}\cdot\vec{i}'} 2^{s-\vec{i}\cdot\vec{i}'}\delta_{\vec{m},\vec{0}} \nonumber\\
    &= M \delta_{\vec{m},\vec{0}} \nonumber\\
    &= M \delta_{\vec{i},\vec{i}'}\,\,, \label{prob-calc4}
\end{align}
where in the second equality we used Lemma~\ref{lemma:fki} since now $\vec{m}$ is a binary vector and in the fourth equality we used the fact that the delta $\delta_{\vec{m},\vec{0}}$ effectively implies that $\vec{i}=\vec{i}'$ over the whole set of $s$ elements since  $\vec{m}$ is given by the elements of $\vec{i}+\vec{i}'$ corresponding to the positions where the two vectors are not both equal to one.

Thus, by using \eqref{prob-calc4} in \eqref{prob-calc3}, we obtain:
\begin{align}
     \sum_{k=0}^{M-1} \left( \left| \sum_{i=0}^{N-1} x_i f_{i,k} \cos{\theta_i}e^{i\phi_i} \right|^2 + \left| \sum_{i=0}^{N-1} x_i f_{i,k} \sin{\theta_i}e^{i\phi_i} \right|^2\right)  &= M +  M \sum_{i,i'=1}^{N-1}x_i x_{i'} \delta_{\vec{i},\vec{i}'} \nonumber\\
     &=M \left(1 + \sum_{i=1}^{N-1} x^2_i \right)  \nonumber\\
     &= M \, N,   \label{prob-calc5}
\end{align}
where we used the fact that $x_i=\pm 1$. Finally, by employing \eqref{prob-calc5} in \eqref{secondterm}, we obtain:
\begin{equation}
      \mbox{Tr}\left[ \rho_{out} \bigotimes_{k=0}^{M-1} \ketbra{0}{0}_k\right] = e^{-N\eta \alpha^2}\label{secondterm2},
\end{equation}
which concludes the calculation of the second trace in \eqref{probDjKG}.

We now move on to calculate the first trace in \eqref{probDjKG}. In a similar manner to \eqref{secondterm}, we can write:
\begin{align}
    \mbox{Tr}\left[ \rho_{out} \mathds{1}_j \bigotimes_{k\neq j} \ketbra{0}{0}_k\right] &  = \prod_{k\neq j} \exp\left[ -\left| \sqrt{\frac{\eta}{M}}\alpha \sum_{i=0}^{N-1} x_i f_{i,k} \cos{\theta_i}e^{i\phi_i} \right|^2 - \left| \sqrt{\frac{\eta}{M}}\alpha \sum_{i=0}^{N-1} x_i f_{i,k} \sin{\theta_i}e^{i\phi_i} \right|^2\right] \nonumber \\
    & =  \exp\left[-\frac{\eta}{M} \alpha^2 \sum_{k\neq j} \left( \left| \sum_{i=0}^{N-1} x_i f_{i,k} \cos{\theta_i}e^{i\phi_i} \right|^2 + \left| \sum_{i=0}^{N-1} x_i f_{i,k} \sin{\theta_i}e^{i\phi_i} \right|^2\right) \right] \nonumber\\
    &=  \exp\left[-\frac{\eta}{M} \alpha^2 \sum_{k\neq j} C_k \right] \nonumber\\ 
    &= \exp\left[-\frac{\eta}{M} \alpha^2 \left(\sum_{k=0}^{M-1} C_k - C_j\right) \right] \nonumber\\
    &= e^{-N\eta \alpha^2} e^{(\eta/M) \alpha^2 C_j}   \label{firstterm},
\end{align}
where in the third line we defined:
\begin{equation}
    C_k:= \left| \sum_{i=0}^{N-1} x_i f_{i,k} \cos{\theta_i}e^{i\phi_i} \right|^2 + \left| \sum_{i=0}^{N-1} x_i f_{i,k} \sin{\theta_i}e^{i\phi_i} \right|^2 \label{Ck}\,,
\end{equation}
and we used \eqref{secondterm2} in the last line. With analogous calculations to those leading to \eqref{prob-calc1}, one can simplify $C_j$ as follows:
\begin{equation} 
    C_j = 1 + \left|\sum_{i=1}^{N-1} x_i f_{i,j}\right|^2 + 2 x_0 \cos{\theta} \cos{\phi}\sum_{i=1}^{N-1} x_i f_{i,j}.
\end{equation}
We now recall that in the post-processing of the protocol, party $A_i$ flips their $X$-basis outcome, $x_i$, if $f_{i,j}=(-1)^{\vec{i}\cdot\vec{j}}=-1$. For this, we identify the sum $\sum_{i=1}^{N-1} x_i f_{i,j}$ in the last expression as the sum of the post-processed $X$-basis outcomes of the parties (excluding $A_0$) and can label it as follows:
\begin{equation} \label{Sj}
    \sum_{i=1}^{N-1} x_i f_{i,j} =: S^j_{x_1,\dots,x_{N-1}}.
\end{equation}
This allows us to recast $C_j$ as follows:
\begin{equation}
    C_j = 1 + (S^j_{x_1,\dots,x_{N-1}})^2 + 2  S^j_{x_1,\dots,x_{N-1}} x_0 \cos{\theta} \cos{\phi}.
\end{equation}
By using the last expression in \eqref{firstterm}, we obtain the final form of the first trace in \eqref{probDjKG}:
\begin{align}
    \mbox{Tr}\left[ \rho_{out} \mathds{1}_j \bigotimes_{k\neq j} \ketbra{0}{0}_k\right]  &= e^{-N\eta \alpha^2} e^{(\eta/M) \alpha^2 \left(1 + (S^j_{x_1,\dots,x_{N-1}})^2 + 2  S^j_{x_1,\dots,x_{N-1}} x_0 \cos{\theta} \cos{\phi}\right)} \nonumber\\
    &= e^{-(MN-1)\eta\alpha^2/M} e^{\eta\alpha^2\left((S^j_{x_1,\dots,x_{N-1}})^2 + 2  S^j_{x_1,\dots,x_{N-1}} x_0 \cos{\theta} \cos{\phi}\right)/M} \label{firstterm2}.
\end{align}
Finally, by combining \eqref{secondterm2} and \eqref{firstterm2} in \eqref{probDjKG}, we obtain the following expression for the probability that only detector $D_j$ clicks, conditioned on the parties preparing coherent states $\ket{x_0\alpha},\dots,\ket{x_{N-1}\alpha}$ in a KG round: 
\begin{align}
    \Pr(\Omega_j|x_0,x_1,\dots,x_{N-1},\mathrm{KG}) = &(1-p_d)^{M-1} e^{-(MN-1)\eta\alpha^2/M} e^{\eta\alpha^2\left((S^j_{x_1,\dots,x_{N-1}})^2 + 2  S^j_{x_1,\dots,x_{N-1}} x_0 \cos{\theta} \cos{\phi}\right)/M} \nonumber\\
    &-(1-p_d)^M e^{-N\eta \alpha^2} \label{probDjKG-final},
\end{align}
where $S^j_{x_1,\dots,x_{N-1}}$ is given in \eqref{Sj} and $\theta=\theta_0 -\theta_1$.\\

\textbf{Computation of $\Pr(\Omega_j|\mathrm{KG})$}\\
We now calculate the probability that detector $D_j$ clicks in a KG round, i.e.
\begin{align}
    \Pr(\Omega_j|\mathrm{KG}) &= \frac{1}{2^N} \sum_{(x_0,\dots,x_{N-1})\in\{1,-1\}^{N}} \Pr(\Omega_j|x_0,x_1,\dots,x_{N-1},\mathrm{KG}) \nonumber \\
     &= -(1-p_d)^Me^{-N\eta \alpha^2} + \frac{(1-p_d)^{M-1}}{2^N} e^{-(MN-1)\eta\alpha^2/M} \nonumber\\
     &\cdot\sum_{(x_0,\dots,x_{N-1})\in\{1,-1\}^{N}} e^{\eta\alpha^2\left((S^j_{x_1,\dots,x_{N-1}})^2 + 2  S^j_{x_1,\dots,x_{N-1}} x_0 \cos{\theta} \cos{\phi}\right)/M}. \label{prDetjKG}
\end{align}
We denote the leftover sum in the last expression as $\Sigma$ for brevity. Then, we can simplify it as follows:
\begin{align}
    \Sigma &= \sum_{(x_1,\dots,x_{N-1})\in\{1,-1\}^{N-1}} \left(e^{\eta\alpha^2\left((S^j_{x_1,\dots,x_{N-1}})^2 + 2  S^j_{x_1,\dots,x_{N-1}} \cos{\theta} \cos{\phi}\right)/M} + e^{\eta\alpha^2\left((S^j_{x_1,\dots,x_{N-1}})^2 - 2  S^j_{x_1,\dots,x_{N-1}} \cos{\theta} \cos{\phi}\right)/M}\right) \nonumber\\
    &=\sum_{(x_1,\dots,x_{N-1})\in\{1,-1\}^{N-1}} e^{\eta\alpha^2(S^j_{x_1,\dots,x_{N-1}})^2/M}\left(e^{\eta\alpha^2 2  S^j_{x_1,\dots,x_{N-1}} \cos{\theta} \cos{\phi}/M} + e^{-\eta\alpha^2 2  S^j_{x_1,\dots,x_{N-1}} \cos{\theta} \cos{\phi}/M}\right) \nonumber\\
    &=2\sum_{(x_1,\dots,x_{N-1})\in\{1,-1\}^{N-1}} e^{\eta\alpha^2(S^j_{x_1,\dots,x_{N-1}})^2/M} \cosh\left(2\frac{\eta\alpha^2}{M}  S^j_{x_1,\dots,x_{N-1}} \cos{\theta}\cos{\phi}\right)
\end{align}
At this point, we define a vector $\vec{y}\in\{1,-1\}^{N-1}$ such that $y_i=x_i f_{i,j}$. Then, we can rewrite $S^j_{x_1,\dots,x_{N-1}}=\sum_i y_i$. Importantly, since we sum over all possible vectors $(x_1,\dots,x_{N-1})$, we reach all possible values for $\vec{y}$. This implies that we can recast the sum over $(x_1,\dots,x_{N-1})$ as a sum over all possible vectors $\vec{y}$. This has the important consequence that the probability of detector $D_j$ clicking is independent of $j$. With these considerations, we rewrite the last expression as follows:
\begin{align}
    \Sigma &= 2\sum_{\vec{y}\in\{1,-1\}^{N-1}} e^{\eta\alpha^2({\textstyle \sum_i y_i})^2/M} \cosh\left(2\frac{\eta\alpha^2}{M} ({\textstyle \sum_i y_i}) \cos{\theta}\cos{\phi}\right).
\end{align}
Now let us call $k$ the number of ones in the vector $\vec{y}$. We have that $\sum_i y_i=k-(N-1-k)=2k+1-N$. Since there are $\binom{N-1}{k}$ different vectors $\vec{y}$ that have a fixed number $k$ of ones, we can recast the last expression as follows:
\begin{align}
    \Sigma &= 2\sum_{k=0}^{N-1}\binom{N-1}{k} e^{\eta\alpha^2(2k+1-N)^2/M} \cosh\left(2\frac{\eta\alpha^2}{M} (2k+1-N) \cos{\theta}\cos{\phi}\right).
\end{align}
By inserting the last expression in \eqref{prDetjKG}, we obtain the final expression for the probability that detector $D_j$ clicks in a KG round:
\begin{align}
    \Pr(\Omega_j|\mathrm{KG}) = &-(1-p_d)^Me^{-N\eta \alpha^2} \nonumber\\
    &+ \frac{(1-p_d)^{M-1}}{2^{N-1}} e^{-(MN-1)\eta\alpha^2/M} \sum_{k=0}^{N-1}\binom{N-1}{k} e^{\eta\alpha^2(2k+1-N)^2/M} \cosh\left(2\frac{\eta\alpha^2}{M} (2k+1-N) \cos{\theta}\cos{\phi}\right), \label{p1}
\end{align}
where $\theta=\theta_0-\theta_1$. As discussed above, the probability that a specific detector clicks is independent of $j$, as expected given our symmetric channel model.\\

\textbf{Computation of $Q_{X_0,X_i}^{j}$}\\
The QBER is computed through \eqref{QBERcomp}, which we report here for clarity:
\begin{align}
    Q^j_{X_0,X_i} & =  \sum_{x_0 \neq x_{i}f_{i,j}} \frac{ \Pr\left( \Omega_j | x_0, x_{i},\mathrm{KG}\right)}{4 \Pr(\Omega_j|\mathrm{KG})} , \label{QBERcomp-app}
\end{align}
where the only quantity that still needs to be computed is $\Pr\left( \Omega_j | x_0, x_{i},\mathrm{KG}\right)$. By definition, we have:
\begin{align}
    \Pr(\Omega_j|x_0, x_{i},\mathrm{KG}) &= \frac{1}{2^{N-2}} \sum_{(x_1,\dots,\hat{x}_i,\dots,x_{N-1})\in\{1,-1\}^{N-2}} \Pr(\Omega_j|x_0,x_1,\dots,x_{N-1},\mathrm{KG}) \nonumber \\
     &= -(1-p_d)^Me^{-N\eta \alpha^2} + \frac{(1-p_d)^{M-1}}{2^{N-2}} e^{-(MN-1)\eta\alpha^2/M} \nonumber\\
     &\cdot\sum_{(x_1,\dots,\hat{x}_i,\dots,x_{N-1})\in\{1,-1\}^{N-2}} e^{\eta\alpha^2\left((S^j_{x_1,\dots,x_{N-1}})^2 + 2  S^j_{x_1,\dots,x_{N-1}} x_0 \cos{\theta} \cos{\phi}\right)/M}, \label{QBERcomp2}
\end{align}
where $(x_1,\dots,\hat{x}_i,\dots,x_{N-1})$ are $(N-2)$-dimensional vectors where the $i$-th element is removed. Then, we can define a vector $\vec{y}\in\{1,-1\}^{N-1}$ with $y_l=x_l f_{l,j}$ for $l \neq i$ and $y_i=0$, such that $S^j_{x_1,\dots,x_{N-1}}=\sum_l y_l + x_i f_{i,j}$. Since the sum in the last expression runs over all vectors $(x_1,\dots,\hat{x}_i,\dots,x_{N-1})$, we can reach all possible choices of $\vec{y}$, meaning that we can recast the sum as a sum over all possible choices of $\vec{y}$. With these considerations, we recast \eqref{QBERcomp2} as follows:
\begin{align}
    \Pr(\Omega_j|x_0, x_{i},\mathrm{KG}) &= -(1-p_d)^Me^{-N\eta \alpha^2} + \frac{(1-p_d)^{M-1}}{2^{N-2}} e^{-(MN-1)\eta\alpha^2/M} \nonumber\\
     &\cdot\sum_{\substack{\vec{y}\in\{1,-1\}^{N-1}:\\ y_i=0}} e^{\eta\alpha^2\left((\sum_l y_l + x_i f_{i,j})^2 + 2  (\sum_l y_l + x_i f_{i,j}) x_0 \cos{\theta} \cos{\phi}\right)/M}. \label{QBERcomp3}
\end{align}
We label the sum as $\Sigma'$ and focus on it:
\begin{align}
    \Sigma'&= \sum_{\substack{\vec{y}\in\{1,-1\}^{N-1}:\\ y_i=0}} e^{\eta\alpha^2\left((\sum_l y_l)^2 + 1+ 2x_i f_{i,j}\sum_l y_l + 2 x_i f_{i,j}x_0 \cos{\theta} \cos{\phi} +2 x_0 \cos{\theta} \cos{\phi}\sum_l y_l\right)/M} \nonumber\\
    &= e^{\eta\alpha^2\left(1+2x_0 x_i f_{i,j}\cos\theta\cos\phi\right)/M} \sum_{\substack{\vec{y}\in\{1,-1\}^{N-1}:\\ y_i=0}} e^{\eta\alpha^2\left((\sum_l y_l)^2 +  2\sum_l y_l(x_i f_{i,j} +  x_0 \cos{\theta} \cos{\phi})\right)/M}.
\end{align}
By replicating the argument in the calculation of $\Pr(\Omega_j|\mathrm{KG})$, we can replace the sum over $\vec{y}$ with a sum over $k$, which is the number of ones in $\vec{y}$:
\begin{align}
    \Sigma' &= e^{\eta\alpha^2\left(1+2x_0 x_i f_{i,j}\cos\theta\cos\phi\right)/M} \sum_{k=0}^{N-2}\binom{N-2}{k} e^{\eta\alpha^2(2k+2-N)^2/M} e^{2 \eta\alpha^2(2k+2-N)(x_i f_{i,j} +  x_0 \cos{\theta} \cos{\phi})/M}
\end{align}
By inserting this in \eqref{QBERcomp3}, we obtain the final expression for the probability that detector $D_j$ clicks, given that party $A_0$ ($A_i$) prepared coherent state $\ket{x_0 \alpha}$ ($\ket{x_i \alpha}$):
\begin{align}
    \Pr(\Omega_j|x_0, x_{i},\mathrm{KG}) &= -(1-p_d)^Me^{-N\eta \alpha^2} + \frac{(1-p_d)^{M-1}}{2^{N-2}} e^{-(MN-2-2x_0 x_i f_{i,j}\cos\theta\cos\phi)\eta\alpha^2/M} \nonumber\\
     &\cdot\sum_{k=0}^{N-2}\binom{N-2}{k} e^{\eta\alpha^2(2k+2-N)^2/M} e^{2 \eta\alpha^2(2k+2-N)(x_i f_{i,j} +  x_0 \cos{\theta} \cos{\phi})/M} . \label{QBERcomp4}
\end{align}
With \eqref{QBERcomp4} we can finally compute the QBER as follows:
\begin{align}
    Q^j_{X_0,X_i} &= \sum_{x_0 \neq x_{i}f_{i,j}} \frac{ \Pr\left( \Omega_j | x_0, x_{i},\mathrm{KG}\right)}{4 \Pr(\Omega_j|\mathrm{KG})} \nonumber\\
    &= \frac{-(1-p_d)^Me^{-N\eta \alpha^2}}{2 \Pr(\Omega_j|\mathrm{KG})} + \frac{(1-p_d)^{M-1}}{2^{N}\Pr(\Omega_j|\mathrm{KG})} e^{-(MN-2+2\cos\theta\cos\phi)\eta\alpha^2/M}\nonumber\\
     &\quad\cdot\sum_{k=0}^{N-2}\binom{N-2}{k} e^{\eta\alpha^2(2k+2-N)^2/M} \left(e^{2 \eta\alpha^2(2k+2-N)(1- \cos{\theta} \cos{\phi})/M}+e^{-2 \eta\alpha^2(2k+2-N)(1- \cos{\theta} \cos{\phi})/M}\right) \nonumber\\
     &= \frac{-(1-p_d)^Me^{-N\eta \alpha^2}}{2 \Pr(\Omega_j|\mathrm{KG})} + \frac{(1-p_d)^{M-1}}{2^{N-1}\Pr(\Omega_j|\mathrm{KG})} e^{-(MN-2+2\cos\theta\cos\phi)\eta\alpha^2/M}\nonumber\\
     &\quad\cdot\sum_{k=0}^{N-2}\binom{N-2}{k} e^{\eta\alpha^2(2k+2-N)^2/M} \cosh\left(2 \frac{\eta\alpha^2}{M}(2k+2-N)(1- \cos{\theta} \cos{\phi})\right). \label{QBERcomp-app2}
\end{align}
which is also independent of $j$ ($\Pr(\Omega_j|\mathrm{KG})$ is independent of $j$, see \eqref{p1}), as well as $i$, due to the symmetry of the considered noise model.\\

\textbf{Computation of $\Pr(\Omega_j|\beta_0,\beta_1,\dots,\beta_{N-1})$}\\
We now calculate the gains, i.e. the probability that only detector $D_j$ clicks in a PE round where the parties prepared phase-randomized coherent states with intensities $\beta_0,\beta_1,\dots,\beta_{N-1}$. We recall that the state \eqref{PEstates} sent by party $A_i$ can be equivalently described as follows:
\begin{equation}
      \rho_{a_i}(\beta_i)=\frac{1}{2\pi}\int_0^{2\pi} \, d\varphi_i \ketbra{\sqrt{\beta_i} e^{i\varphi_i}}{\sqrt{\beta_i} e^{i\varphi_i}},
\end{equation}
where $\beta_i\in\mathcal{S}_i$. Thus, the state sent by all parties reads:
\begin{align}
    \rho_{in}=\bigotimes_{i=0}^{N-1}\rho_{a_i}(\beta_i)=\frac{1}{(2\pi)^N}\int_{0}^{2\pi} \, d\varphi_0 \dots d\varphi_{N-1} \bigotimes_{i=0}^{N-1}\ketbra{\sqrt{\beta_i} e^{i\varphi_i}}{\sqrt{\beta_i} e^{i\varphi_i}}. \label{rhoin}
\end{align}
We now apply our channel model comprising a pure-loss channel and a polarization misalignment (we neglect the phase shift as the states are already phase-randomized). After going through the lossy and noisy channel, $\rho_{in}$ evolves to:
\begin{align}
    \rho'_{in}=\frac{1}{(2\pi)^N}\int_{0}^{2\pi} \, d\varphi_0 \dots d\varphi_{N-1}  \bigotimes_{i=0}^{N-1}\ketbra{\cos{\theta_i}\sqrt{\eta\beta_i}e^{i\varphi_i}}{\cos{\theta_i}\sqrt{\eta\beta_i}e^{i\varphi_i}}_P \otimes \ketbra{-\sin{\theta_i}\sqrt{\eta\beta_i}e^{i\varphi_i}}{-\sin{\theta_i}\sqrt{\eta\beta_i}e^{i\varphi_i}}_{P_\perp}.
\end{align}
The final step consists in evolving $\rho'_{in}$ through the BBS network. We obtain the following state:
\begin{align}
    \rho_{out}=\frac{1}{(2\pi)^N}\int_{0}^{2\pi} \, d\varphi_0 \dots d\varphi_{N-1} & \bigotimes_{k=0}^{M-1}\left|\sqrt{\frac{\eta}{M}}\sum_{i=0}^{N-1}f_{i,k}\cos{\theta_i}\sqrt{\beta_i}e^{i\varphi_i}\right\rangle \left\langle \sqrt{\frac{\eta}{M}}\sum_{i=0}^{N-1}f_{i,k}\cos{\theta_i}\sqrt{\beta_i}e^{i\varphi_i} \right|_P \nonumber \\ & \otimes\left|-\sqrt{\frac{\eta}{M}}\sum_{i=0}^{N-1}f_{i,k}\sin{\theta_i}\sqrt{\beta_i}e^{i\varphi_i}\right\rangle \left\langle -\sqrt{\frac{\eta}{M}}\sum_{i=0}^{N-1}f_{i,k}\sin{\theta_i}\sqrt{\beta_i}e^{i\varphi_i} \right|_{P_\perp},
\end{align}
which we remark is not anymore a product state of phase-randomized coherent states. Now, similarly to the calculation of $\Pr(\Omega_j|x_0,\dots,x_{N-1},\mathrm{KG})$, we can express each gain as follows:
\begin{align}
    \Pr(\Omega_j|\beta_0,\beta_1,\dots,\beta_{N-1})& = (1-p_d)^{M-1} \mbox{Tr}\left[ \rho_{out} \mathds{1}_j \bigotimes_{k\neq j} \ketbra{0}{0}_k \right] - (1-p_d)^M \mbox{Tr}\left[ \rho_{out} \bigotimes_{k=0}^{M-1} \ketbra{0}{0}_k \right] , \label{gains-app}
\end{align}
where $\ketbra{0}{0}_k$ is the projector on the vacuum of the output mode $k$ for polarizations $P$ and $P_\perp$, since the detectors do not distinguish polarization. We now evaluate the two terms in \eqref{gains-app}. Let us begin with the second, i.e.
\begin{align}
    \mbox{Tr}\left[ \rho_{out} \bigotimes_{k=0}^{M-1} \ketbra{0}{0}_k \right] &=\frac{1}{(2\pi)^N}\int_{0}^{2\pi} \, d\varphi_0 \dots d\varphi_{N-1} \nonumber\\
    &\quad\prod_{k=0}^{M-1} \abs{\braket{0|\sqrt{\frac{\eta}{M}}\sum_{i=0}^{N-1}f_{i,k}\cos{\theta_i}\sqrt{\beta_i}e^{i\varphi_i}}}^2 \abs{\braket{0|-\sqrt{\frac{\eta}{M}}\sum_{i=0}^{N-1}f_{i,k}\sin{\theta_i}\sqrt{\beta_i}e^{i\varphi_i}}}^2 \nonumber\\
    &=\frac{1}{(2\pi)^N}\int_{0}^{2\pi} \, d\varphi_0 \dots d\varphi_{N-1} \nonumber\\
    &\quad\prod_{k=0}^{M-1} \exp\left[-\abs{\sqrt{\frac{\eta}{M}}\sum_{i=0}^{N-1}f_{i,k}\cos{\theta_i}\sqrt{\beta_i}e^{i\varphi_i}}^2 - \abs{\sqrt{\frac{\eta}{M}}\sum_{i=0}^{N-1}f_{i,k}\sin{\theta_i}\sqrt{\beta_i}e^{i\varphi_i}}^2\right] \nonumber\\
    &=\int_{0}^{2\pi} \, \frac{d\varphi_0 \dots d\varphi_{N-1}}{(2\pi)^N} \exp\left[-\frac{\eta}{M}\sum_{k=0}^{M-1}\left(\abs{\sum_{i=0}^{N-1}f_{i,k}\cos{\theta_i}\sqrt{\beta_i}e^{i\varphi_i}}^2 + \abs{\sum_{i=0}^{N-1}f_{i,k}\sin{\theta_i}\sqrt{\beta_i}e^{i\varphi_i}}^2\right)\right] \nonumber\\
    &\equiv \int_{0}^{2\pi} \, \frac{d\varphi_0 \dots d\varphi_{N-1}}{(2\pi)^N} e^{-\frac{\eta}{M}\sum_{k=0}^{M-1} C_k}\label{gains-app2}.
\end{align}
Let us now focus on the sum of the terms labelled $C_k$. By expanding the squares in $C_k$ we obtain:
\begin{align}
    \sum_{k=0}^{M-1} C_k  &=  \sum_{k=0}^{M-1}  \left( \sum_{i=0}^{N-1} \abs{f_{i,k} \cos{\theta_i}\sqrt{\beta_{i}} e^{i\varphi_i}}^2  + \sum_{\substack{i,i'=0 \\ i \neq i'}}^{N-1} f_{i,k} f_{i',k} \cos{\theta_i} \cos{\theta_{i'}}\sqrt{\beta_i\beta_{i'}}e^{i(\varphi_i-\varphi_{i'})}  \right. \nonumber \\
   &\quad + \left. \sum_{i=0}^{N-1} \abs{f_{i,k} \sin{\theta_i} \sqrt{\beta_i}e^{i\varphi_i}}^2  + \sum_{\substack{i,i'=0 \\ i \neq i'}}^{N-1} f_{i,k} f_{i',k} \sin{\theta_i} \sin{\theta_{i'}}\sqrt{\beta_i\beta_{i'}}e^{i(\varphi_i-\varphi_{i'})} \right) \nonumber\\
   &=  \sum_{k=0}^{M-1}  \left( \sum_{i=0}^{N-1} \cos^2{\theta_i}\beta_{i}   + 2\sum_{\substack{i,i'=0 \\ i < i'}}^{N-1} f_{i,k} f_{i',k} \cos{\theta_i} \cos{\theta_{i'}}\sqrt{\beta_i\beta_{i'}}\cos(\varphi_i-\varphi_{i'})  \right. \nonumber \\
   &\quad + \left. \sum_{i=0}^{N-1} \sin^2{\theta_i} \beta_i  + 2\sum_{\substack{i,i'=0 \\ i < i'}}^{N-1} f_{i,k} f_{i',k} \sin{\theta_i} \sin{\theta_{i'}}\sqrt{\beta_i\beta_{i'}}\cos(\varphi_i-\varphi_{i'}) \right). \label{Ck-gains}
\end{align}
Now, we use the result in \eqref{prob-calc4} (derived from Lemma~\ref{lemma:fki}) to argue that:
\begin{align}
    \sum_{k=0}^{M-1}f_{i,k}f_{i',k} &= \sum_{k=0}^{M-1}(-1)^{(\vec{i}+\vec{i}')\cdot \vec{k}} \nonumber\\
    &=M \delta_{\vec{i},\vec{i}'}.
\end{align}
By applying this result in \eqref{Ck-gains}, and by noting that $\vec{i}$ and $\vec{i}'$ must differ in the sums that involve them, we are left with:
\begin{align}
     \sum_{k=0}^{M-1} C_k =  \sum_{k=0}^{M-1} \left( \sum_{i=0}^{N-1}\beta_{i} \cos{\theta_i}^2 + \sum_{i=0}^{N-1}\beta_i \sin{\theta_i}^2 \right) = M \sum_{i=0}^{N-1}\beta_i.
\end{align}
By using this result in \eqref{gains-app2}, we can directly integrate over the phases and obtain the following expression for the second term in \eqref{gains-app}:
\begin{equation}
     \mbox{Tr}\left[ \rho_{out} \bigotimes_{k=0}^{M-1} \ketbra{0}{0}_k \right]= e^{-\eta \sum_i \beta_i}. \label{gains-app5}
\end{equation}
Regarding the first term in \eqref{gains-app}, we can express it as follows:
\begin{align}
    \mbox{Tr}\left[ \rho_{out} \mathds{1}_j \bigotimes_{k\neq j} \ketbra{0}{0}_k\right] &= \frac{1}{(2\pi)^N}  \int_{0}^{2\pi} \, d\varphi_0 \dots d\varphi_{N-1} \; e^{-\frac{\eta}{M}\sum_{k=0,\,k\neq j}^{M-1} C_k} \nonumber\\
    &=\frac{1}{(2\pi)^N}  \int_{0}^{2\pi} \, d\varphi_0 \dots d\varphi_{N-1} e^{-\frac{\eta}{M}\left(\sum_{k=0}^{M-1} C_k - C_j\right)}  \nonumber \\
    & = e^{-\eta \sum_i \beta_i}  \frac{1}{(2\pi)^N}  \int_{0}^{2\pi} \, d\varphi_0 \dots d\varphi_{N-1} e^{\eta C_j/M}.  \label{gains-app3}
\end{align}
Now we calculate the coefficient $C_j$ by expanding its squares:
\begin{align}
    C_j &= \abs{\sum_{i=0}^{N-1}f_{i,j}\cos{\theta_i}\sqrt{\beta_i}e^{i\varphi_i}}^2 + \abs{\sum_{i=0}^{N-1}f_{i,j}\sin{\theta_i}\sqrt{\beta_i}e^{i\varphi_i}}^2 \nonumber\\
    &=\sum_{i=0}^{N-1} \beta_i + 2\sum_{\substack{i,i'=0 \\ i<i'}}^{N-1} f_{i,j}f_{i',j}(\cos\theta_i \cos\theta_{i'}+\sin\theta_i \sin\theta_{i'}) \sqrt{\beta_i \beta_{i'}}\cos(\varphi_i -\varphi_{i'}) \nonumber\\
    &=\sum_{i=0}^{N-1} \beta_i + 2\sum_{\substack{i,i'=0 \\ i<i'}}^{N-1} f_{i,j}f_{i',j}\cos(\theta_i- \theta_{i'}) \sqrt{\beta_i \beta_{i'}}\cos(\varphi_i -\varphi_{i'}).
\end{align}
Now we use the fact that $\theta_i=\theta_1$ for every $i\geq 1$. By splitting the second sum into two terms, where the first has $i=0$ fixed and in the second $i\geq 1$, we obtain:
\begin{align}
    C_j &= \sum_{i=0}^{N-1} \beta_i + 2\cos\theta \sum_{i=1}^{N-1}f_{i,j}\sqrt{\beta_0 \beta_{i}}\cos(\varphi_0 -\varphi_{i}) + 2\sum_{\substack{i,i'=1 \\ i<i'}}^{N-1} f_{i,j}f_{i',j} \sqrt{\beta_i \beta_{i'}}\cos(\varphi_i -\varphi_{i'})
\end{align}
where $\theta=\theta_0-\theta_1$. By using this expression in \eqref{gains-app3}, we obtain the following expression for the first term in \eqref{gains-app}:
\begin{align}
    \mbox{Tr}\left[ \rho_{out} \mathds{1}_j \bigotimes_{k\neq j} \ketbra{0}{0}_k\right] &= e^{-\eta(1-1/M) \sum_i \beta_i} \,\,I_j(\beta_0,\dots,\beta_{N-1}),  \label{gains-app4}
\end{align}
where we defined the integral:
\begin{align}
    I_j(\beta_0,\dots,\beta_{N-1}):= &\frac{1}{(2\pi)^N}  \int_{0}^{2\pi} \, d\varphi_0 \dots d\varphi_{N-1} \nonumber\\
    &\exp\left[\frac{2\eta}{M}\left(\cos\theta \sum_{i=1}^{N-1}f_{i,j}\sqrt{\beta_0 \beta_{i}}\cos(\varphi_0 -\varphi_{i}) + \sum_{\substack{i,i'=1 \\ i<i'}}^{N-1} f_{i,j}f_{i',j} \sqrt{\beta_i \beta_{i'}}\cos(\varphi_i -\varphi_{i'})\right)\right]. \label{integral}
\end{align}

By employing \eqref{gains-app4} and \eqref{gains-app5} in \eqref{gains-app}, we obtain the following compact expression for the gains:
\begin{align}
    \Pr(\Omega_j|\beta_0,\beta_1,\dots,\beta_{N-1})& = (1-p_d)^{M-1} e^{-\eta(1-1/M) \sum_i \beta_i} \,\,I_j(\beta_0,\dots,\beta_{N-1}) - (1-p_d)^M e^{-\eta\sum_i \beta_i} , 
\end{align}
where $I_j(\beta_0,\dots,\beta_{N-1})$ is given in \eqref{integral}.

Importantly, due to our symmetric channel model, the gains are independent of which detector $D_j$ clicks. To show this, we argue that the integral in \eqref{integral} is actually independent of $j$. To this aim, we label the function to be integrated in \eqref{integral} as follows: 
\begin{equation}
    F_j(\varphi_0, \dots,\varphi_{N-1}):= \exp\left[\frac{2\eta}{M}\left(\cos\theta \sum_{i=1}^{N-1}(-1)^{\vec{i}\cdot\vec{j}}\sqrt{\beta_0 \beta_{i}}\cos(\varphi_0 -\varphi_{i}) + \sum_{\substack{i,i'=1 \\ i<i'}}^{N-1} (-1)^{(\vec{i}+\vec{i}')\cdot\vec{j}} \sqrt{\beta_i \beta_{i'}}\cos(\varphi_i -\varphi_{i'})\right)\right]
\end{equation}
and observe that this function is periodic in each variable $\varphi_i$, with period $2\pi$. The only dependency of $F_j$ on $j$ comes from the $\pm 1$ signs inside the sums. We can reabsorb such signs by defining new integration variables $\Phi_i:=\varphi_i-\pi \cdot (\vec{i}\cdot\vec{j})$, which allow us to simplify the summands as follows:
\begin{align}
    (-1)^{\vec{i}\cdot\vec{j}}\cos(\varphi_0 -\varphi_{i}) &= \cos(\Phi_0 -\Phi_{i}), \\
    (-1)^{(\vec{i}+\vec{i}')\cdot\vec{j}}\cos(\varphi_i -\varphi_{i'}) &= (-1)^{\vec{i}'\cdot\vec{j}-\vec{i}\cdot\vec{j}}\cos(\varphi_i -\varphi_{i'}) =\cos(\Phi_i -\Phi_{i'}).
\end{align}
Then, by performing the change of variable $\Phi_i:=\varphi_i-\pi \cdot (\vec{i}\cdot\vec{j})$ in the integral and by using the fact that the function $F_j$ is periodic in each variable, we obtain:
\begin{align}
    I_j(\beta_0,\dots,\beta_{N-1}) &= \frac{1}{(2\pi)^N}  \int_{-\pi(\vec{i}\cdot\vec{j})}^{2\pi -\pi(\vec{i}\cdot\vec{j})} d\Phi_0 d\Phi_1 \dots d\Phi_{N-1} F_0(\Phi_0,\Phi_1,\dots,\Phi_{N-1}) \nonumber\\
    &=\frac{1}{(2\pi)^N}  \int_{0}^{2\pi} \, d\Phi_0 d\Phi_1 \dots d\Phi_{N-1} F_0(\Phi_0,\Phi_1,\dots,\Phi_{N-1}) \nonumber\\
    &= I_0(\beta_0,\dots,\beta_{N-1}), \label{integral2}
\end{align}
which confirms that $I_j$ is independent of $j$. The final formula for the gains is thus:
\begin{align}
    \Pr(\Omega_j|\beta_0,\beta_1,\dots,\beta_{N-1})& = (1-p_d)^{M-1} e^{-\eta(1-1/M) \sum_i \beta_i} \,\,I(\beta_0,\dots,\beta_{N-1}) - (1-p_d)^M e^{-\eta\sum_i \beta_i} , \label{channelgain}
\end{align}
where the integral:
\begin{align}
    I(\beta_0,\dots,\beta_{N-1})=  \int_{0}^{2\pi} \, \frac{d\varphi_0 \dots d\varphi_{N-1}}{(2\pi)^N} \exp\left[\frac{2\eta}{M}\left(\cos\theta \sum_{i=1}^{N-1}\sqrt{\beta_0 \beta_{i}}\cos(\varphi_0 -\varphi_{i}) + \sum_{\substack{i,i'=1 \\ i<i'}}^{N-1} \sqrt{\beta_i \beta_{i'}}\cos(\varphi_i -\varphi_{i'})\right)\right], \label{integral3}
\end{align}
is evaluated numerically in our simulations. Note that we freely relabelled the variables in the integral   back to $\varphi_i$.\\

\textbf{Computation of $\Pr(\Omega_j|n_0,\dots,n_{N-1})$}\\
Here we calculate the analytical expression of any yield $Y^j_{n_0,\dots,n_{N-1}}$, defined in Eq. \eqref{yields} as the probability that detector $D_j$ clicks given the hypothetical scenario in which party $A_i$ sent exactly $n_i$ photons.

We remark that in an experiment the parties cannot, in general, learn the exact value of each yield with the decoy-state analysis, but can derive upper bounds as shown in Appendix~\ref{section:yieldbound}. In the limit of an infinite number of decoy intensities, the yields' upper bounds would tend to the exact values computed here. 

To evaluate $Y^j_{n_0,\dots,n_{N-1}}$, we consider the scenario in which the parties send the state $\bigotimes_{i=0}^{N-1}|n_i \rangle$, where $|n_i\rangle$ is a Fock state of $n_i$ photons. The state can be written as
\begin{equation}
    |\xi (n_0,\dots,n_{N-1}) \rangle = \left( \prod_{i=0}^{N-1} \frac{(\hat{a}_i^{\dag})^{n_i}}{\sqrt{n_i!}} \right)\ket{0},
\end{equation}
where  $\hat{a}_i^{\dag}$ is the creation operator of the optical mode of party $A_i$ and $|0\rangle$ represents the vacuum state on all modes. We now introduce, step by step, the effect of all sources of noise and then apply the BBS network. 

The lossy channel transforms each party's mode according to:
\begin{equation}
    \hat{a}_i^\dag \rightarrow \sqrt{\eta}\hat{a}_i^\dag + \sqrt{1-\eta} \hat{l}_i^\dag,
\end{equation}
where $\hat{l}_i^\dag$ is the creation operator of the loss mode of party $A_i$. The input state $\ket{\xi}$ is transformed as follows:
\begin{align}\label{nstate1}
    |\xi' (n_0,\dots,n_{N-1}) \rangle & = \left( \prod_{i=0}^{N-1} \frac{(\sqrt{\eta}\hat{a}_i^\dag + \sqrt{1-\eta} \hat{l}_i^\dag)^{n_i}}{\sqrt{n_i!}} \right)|0\rangle \nonumber \\
    & = \left[ \prod_{i=0}^{N-1} \left( \sum_{k_i=0}^{n_i} \binom{n_i}{k_i}\frac{\eta^{\frac{k_i}{2}}(1-\eta)^{\frac{n_i-k_i}{2}}}{\sqrt{n_i!}} (\hat{a_i}^{\dag})^{k_i}(\hat{l}_i^\dag)^{n_i-k_i} \right)\right] | 0 \rangle \nonumber \\
    & = \sum_{k_0=0}^{n_0} \cdots \sum_{k_{N-1}=0}^{n_{N-1}} \binom{n_0}{k_0} \cdots \binom{n_{N-1}}{k_{N-1}} \frac{\eta^{\frac{\sum_{i}k_i}{2}}(1-\eta)^{\frac{\sum_{i}(n_i-k_i)}{2}}}{\sqrt{n_0! \cdots n_{N-1}!}}  \sqrt{(n_0-k_0)! \cdots (n_{N-1}-k_{N-1})!} \nonumber \\
    & \cdot \left[ \prod_{i=0}^{N-1} (\hat{a_i}^{\dag})^{k_i} \right]|0\rangle_{a_0,\dots,a_{N-1}}\otimes |n_0-k_0\rangle_{l_0} \otimes \cdots \otimes |n_{N-1}-k_{N-1}\rangle_{l_{N-1}},
\end{align}
where we just used the binomial expansion in the second line and where $a_i$ and $l_i$ are used to indicate the optical mode and the loss mode of party $A_i$, respectively.

We now note that the loss modes are not observed by the parties and thus need to be traced out. The density matrix $\rho'=|\xi'\rangle \langle \xi'|$ will thus have two sets of indices $(k_0, \dots k_{N-1})$ and $(k'_0, \dots k'_{N-1})$. However, it is immediate to see from Eq. \eqref{nstate1} that tracing out the loss modes will impose the conditions $k_i=k'_i \;\; \forall i$. Thus we are left with the state
\begin{align}
    \rho' & = \sum_{k_0=0}^{n_0} \cdots \sum_{k_{N-1}=0}^{n_{N-1}} \binom{n_0}{k_0}^2 \cdots \binom{n_{N-1}}{k_{N-1}}^2 \eta^{\sum_{i}k_i}(1-\eta)^{\sum_{i}(n_i-k_i)}\frac{(n_0-k_0)! \cdots (n_{N-1}-k_{N-1})!}{n_0! \cdots n_{N-1}!} \nonumber \\
    & \cdot\left[ \prod_{i=0}^{N-1} (\hat{a_i}^{\dag})^{k_i} \right]\ketbra{0}{0}_{a_0,\dots,a_{N-1}} \left[ \prod_{i=0}^{N-1} (\hat{a_i})^{k_i} \right] \nonumber \\
    & = \sum_{k_0=0}^{n_0} \cdots \sum_{k_{N-1}=0}^{n_{N-1}} \binom{n_0}{k_0} \cdots \binom{n_{N-1}}{k_{N-1}} \frac{\eta^{\sum_{i}k_i}(1-\eta)^{\sum_{i}(n_i-k_i)}}{k_0!\cdots k_{N-1}!} \left[ \prod_{i=0}^{N-1} (\hat{a_i}^{\dag})^{k_i} \right]\ketbra{0}{0}_{a_0,\dots,a_{N-1}} \left[ \prod_{i=0}^{N-1} (\hat{a_i})^{k_i} \right], \label{rhoprime-app}
\end{align}
where we used the fact that $\frac{(n_i-k_i)!}{n_i!}=\frac{1}{\binom{n_i}{k_i}k_i!}$ and where, from now on, for simplicity of notation we will neglect the explicit dependence of the state on $n_0,\dots,n_{N-1}$. 

We now introduce the polarization misalignment, while we skip the phase misalignment since its effect cancels out on tensor products of Fock states. The polarization misalignment acts on the creation operators of each mode as follows:
\begin{equation}
    \hat{a}^\dag_i \rightarrow \cos{\theta_i}\hat{a}^\dag_{i,P}-\sin{\theta_i}\hat{a}^\dag_{i,P_{\perp}},
\end{equation}
where we recall that in our channel model we set $\theta_i=\theta_1$ for $i \geq 1$, i.e. we only introduce a misalignment between the reference party $A_0$ and the other parties. For simplicity of notation we will omit the label $P$ and consider the polarization $P$ to be the input polarization and label the orthogonal polarization with $\perp$. By applying the above transformation to the creation operators in \eqref{rhoprime-app} and by using again the binomial expansion we obtain:
\begin{align}
    \prod_{i=0}^{N-1} & \left(\cos{\theta_i}\hat{a_i}^{\dag}-\sin{\theta_i}\hat{a}_{i,\,\perp}^{\dag}\right)^{k_i} \nonumber \\
    & =  \prod_{i=0}^{N-1} \left( \sum_{l_i=0}^{k_i} (-1)^{k_i-l_i}\binom{k_i}{l_i}(\cos{\theta_i})^{l_i}(\sin{\theta_i})^{k_i-l_i}(\hat{a_i}^{\dag})^{l_i}(\hat{a}_{i,\,\perp}^{\dag})^{k_i-l_i}\right) \nonumber \\
    & =  \sum_{l_0=0}^{k_0} \cdots \sum_{l_{N-1}=0}^{k_{N-1}} (-1)^{\sum_i (k_i-l_i)} \binom{k_0}{l_0} \cdots \binom{k_{N-1}}{l_{N-1}} (\cos{\theta_0})^{l_0}(\sin{\theta_0})^{k_0-l_0}(\cos{\theta_1})^{\sum_{i=1}^{N-1}l_i}(\sin{\theta_1})^{\sum_{i=1}^{N-1}(k_i-l_i)} \nonumber \\
    & \cdot \left[ \prod_{i=0}^{N-1}(\hat{a_i}^{\dag})^{l_i} (\hat{a}_{i,\,\perp}^{\dag})^{k_i-l_i} \right]. 
\end{align}
By using this expression in \eqref{rhoprime-app}, we obtain:
\begin{align}
    \rho'' & = \sum_{k_0=0}^{n_0} \cdots \sum_{k_{N-1}=0}^{n_{N-1}} \sum_{l_0=0}^{k_0} \cdots \sum_{l_{N-1}=0}^{k_{N-1}} \sum_{l'_0=0}^{k_0} \cdots \sum_{l'_{N-1}=0}^{k_{N-1}} \binom{n_0}{k_0} \cdots \binom{n_{N-1}}{k_{N-1}} \binom{k_0}{l_0} \cdots \binom{k_{N-1}}{l_{N-1}} \binom{k_0}{l'_0} \cdots \binom{k_{N-1}}{l'_{N-1}} \nonumber \\
    &\cdot  (-1)^{\sum_i (2k_i-l_i-l'_i)} \frac{\eta^{\sum_{i}k_i}(1-\eta)^{\sum_{i}(n_i-k_i)}}{k_0!\cdots k_{N-1}!} (\cos{\theta_0})^{l_0+l'_0}(\sin{\theta_0})^{2k_0-l_0-l'_0}(\cos{\theta_1})^{\sum_{i=1}^{N-1}(l_i+l'_i)}(\sin{\theta_1})^{\sum_{i=1}^{N-1}(2k_i-l_i-l'_i)} \nonumber \\
    & \cdot  \left[ \prod_{i=0}^{N-1}(\hat{a_i}^{\dag})^{l_i} (\hat{a}_{i,\,\perp}^{\dag})^{k_i-l_i} \right] \ketbra{0}{0}_{a_0,\dots,a_{N-1},a_{0,\perp},\dots,a_{N-1,\perp}} \left[ \prod_{i=0}^{N-1}(\hat{a_{i}})^{l'_i} (\hat{a}_{i,\,\perp})^{k_i-l'_i} \right].
\end{align}

We now let the state evolve through the optical setup of the BBS network. The resulting transformation of the incoming creation operators is given in \eqref{transf}, as proved in Appendix~\ref{sec:BBSnet}. This brings us to the following state of the output modes in the BBS network:
\begin{align}\label{yields:outstate}
    \rho_{out} & = \sum_{k_0=0}^{n_0} \cdots \sum_{k_{N-1}=0}^{n_{N-1}} \sum_{l_0=0}^{k_0} \cdots \sum_{l_{N-1}=0}^{k_{N-1}} \sum_{l'_0=0}^{k_0} \cdots \sum_{l'_{N-1}=0}^{k_{N-1}} \binom{n_0}{k_0} \cdots \binom{n_{N-1}}{k_{N-1}} \binom{k_0}{l_0} \cdots \binom{k_{N-1}}{l_{N-1}} \binom{k_0}{l'_0} \cdots \binom{k_{N-1}}{l'_{N-1}} \nonumber \\
    &\cdot  (-1)^{\sum_i (2k_i-l_i-l'_i)} \frac{\eta^{\sum_{i}k_i}(1-\eta)^{\sum_{i}(n_i-k_i)}}{k_0!\cdots k_{N-1}!} (\cos{\theta_0})^{l_0+l'_0}(\sin{\theta_0})^{2k_0-l_0-l'_0}(\cos{\theta_1})^{\sum_{i=1}^{N-1}(l_i+l'_i)}(\sin{\theta_1})^{\sum_{i=1}^{N-1}(2k_i-l_i-l'_i)} \nonumber \\
    & \cdot  \left[ \prod_{i=0}^{N-1} \left(\frac{1}{\sqrt{M}}\sum_{s=0}^{M-1} (-1)^{\vec{s}\cdot \vec{i}}\,  \hat{d}_s^\dag \,\right)^{l_i} \left(\frac{1}{\sqrt{M}}\sum_{s'=0}^{M-1} (-1)^{\vec{s}\,'\cdot \vec{i}}\,  \hat{d}_{s',\perp}^\dag \,\right)^{k_i-l_i} \right] \ketbra{0}{0}_{d_0,\dots,d_{N-1},d_{0,\perp},\dots,d_{N-1,\perp}} \nonumber \\
    & \cdot\left[ \prod_{i=0}^{N-1}\left(\frac{1}{\sqrt{M}}\sum_{q=0}^{M-1} (-1)^{\vec{q}\cdot \vec{i}}\,  \hat{d}_q \,\right)^{l'_i} \left(\frac{1}{\sqrt{M}}\sum_{q'=0}^{M-1} (-1)^{\vec{q}\,'\cdot \vec{i}}\,  \hat{d}_{q',\perp} \,\right)^{k_i-l'_i} \right].
\end{align}
From the definition of yields, $Y^j_{n_0,\dots,n_{N-1}}=\Pr(\Omega_j|n_0,\dots,n_{N-1})$, we can express them as follows by including dark counts (each detector has a probability $p_d$ of a dark count):
\begin{equation}
    Y^j_{n_0,\dots,n_{N-1}}=(1-p_d)^{M-1} \mbox{Tr}\left[ \rho_{out} \mathds{1}_{j} \bigotimes^{M-1}_{r\neq j} \ketbra{0}{0}_r\right] - (1-p_d)^M\mbox{Tr}\left[ \rho_{out} \bigotimes_{r=0}^{M-1} \ketbra{0}{0}_r\right], \label{exactyields-app}
\end{equation}
where the identity operator and the projector on the vacuum are defined on both modes of polarization, since the detectors cannot distinguish them. We note that calculating the second trace in \eqref{exactyields-app} is trivial: projecting all modes onto the vacuum forces all indexes to be equal to zero, thus yielding the result: 
\begin{equation}
    \mbox{Tr}\left[ \rho_{out} \bigotimes_{r=0}^{M-1} \ketbra{0}{0}_r\right] = (1-\eta)^{\sum_i n_i}. \label{secondtrace-app}
\end{equation}
In order to calculate the first trace in \eqref{exactyields-app}, we would need to expand the sums over the detectors' creation modes using multinomial expansions. However, since we need to project onto the vacuum state in all modes except modes $d_j$ and $d_{j,\perp}$, this operation will force all the terms in the multinomial expansion to vanish, except for the terms containing $\hat{d}_j$ or $\hat{d}_{j,\perp}$. Therefore, the reduced state of $\rho_{out}$ after projecting onto the vacuum all modes except the $j$-th mode, $\rho^{(j)}_{out}:=\braket{0_1 \dots,0_{j-1},0_{j+1},\dots,0_{M-1}|\rho_{out}|0_1 \dots,0_{j-1},0_{j+1},\dots,0_{M-1}}$, reads:
\begin{align}\label{yields:outstatered}
    \rho_{out}^{(j)} &= \sum_{k_0=0}^{n_0} \cdots \sum_{k_{N-1}=0}^{n_{N-1}} \sum_{l_0=0}^{k_0} \cdots \sum_{l_{N-1}=0}^{k_{N-1}} \sum_{l'_0=0}^{k_0} \cdots \sum_{l'_{N-1}=0}^{k_{N-1}} \binom{n_0}{k_0} \cdots \binom{n_{N-1}}{k_{N-1}} \binom{k_0}{l_0} \cdots \binom{k_{N-1}}{l_{N-1}} \binom{k_0}{l'_0} \cdots \binom{k_{N-1}}{l'_{N-1}} \nonumber \\
    &\cdot  (-1)^{\sum_i (2k_i-l_i-l'_i)} \frac{\eta^{\sum_{i}k_i}(1-\eta)^{\sum_{i}(n_i-k_i)}}{k_0!\cdots k_{N-1}!} (\cos{\theta_0})^{l_0+l'_0}(\sin{\theta_0})^{2k_0-l_0-l'_0}(\cos{\theta_1})^{\sum_{i=1}^{N-1}(l_i+l'_i)}(\sin{\theta_1})^{\sum_{i=1}^{N-1}(2k_i-l_i-l'_i)} \nonumber \\
    & \cdot  \left[ \prod_{i=0}^{N-1}\left(\frac{(-1)^{\vec{j}\cdot \vec{i}}}{\sqrt{M}} \right)^{k_i}\left(\hat{d}_j^{\dag} \right)^{l_i} \left(  \hat{d}_{j,\perp}^\dag\right)^{k_i-l_i} \right] \ketbra{0}{0}_{d_j,d_{j,\perp}} \left[ \prod_{i=0}^{N-1}\left(\frac{(-1)^{\vec{j}\cdot \vec{i}}}{\sqrt{M}} \right)^{k_i}\left(\hat{d}_j \right)^{l'_i} \left(  \hat{d}_{j,\perp}\right)^{k_i-l'_i} \right] \nonumber \\
    &=  \sum_{k_0=0}^{n_0} \cdots \sum_{k_{N-1}=0}^{n_{N-1}} \sum_{l_0=0}^{k_0} \cdots \sum_{l_{N-1}=0}^{k_{N-1}} \sum_{l'_0=0}^{k_0} \cdots \sum_{l'_{N-1}=0}^{k_{N-1}} \binom{n_0}{k_0} \cdots \binom{n_{N-1}}{k_{N-1}} \binom{k_0}{l_0} \cdots \binom{k_{N-1}}{k_{N-1}} \binom{k_0}{l'_0} \cdots \binom{k_{N-1}}{l'_{N-1}} \nonumber \\
    &\cdot  (-1)^{\sum_i (2k_i-l_i-l'_i)} \frac{\eta^{\sum_{i}k_i}(1-\eta)^{\sum_{i}(n_i-k_i)}}{M^{\sum_i k_i} k_0!\cdots k_{N-1}!} (\cos{\theta_0})^{l_0+l'_0}(\sin{\theta_0})^{2k_0-l_0-l'_0}(\cos{\theta_1})^{\sum_{i=1}^{N-1}(l_i+l'_i)}(\sin{\theta_1})^{\sum_{i=1}^{N-1}(2k_i-l_i-l'_i)} \nonumber \\
    & \cdot  \left[\left(\hat{d}_j^{\dag} \right)^{\sum_i l_i} \left(  \hat{d}_{j,\perp}^\dag\right)^{\sum_i (k_i-l_i)} \right] \ketbra{0}{0}_{d_j,d_{j,\perp}} \left[ \left(\hat{d}_j\right)^{\sum_i l'_i} \left(  \hat{d}_{j,\perp}\right)^{\sum_i (k_i-l'_i)} \right],
\end{align}
where we used the fact that $\left( (-1)^{\vec{j}\cdot\vec{i}} \right)^{2 \sum_i k_i}= 1$. We observe that, as expected, the yields do not depend on $j$, i.e. on the detector that clicked, due to our symmetric channel model.

We can now compute the first trace in \eqref{exactyields-app} by simply taking the trace of $\rho^{(j)}_{out}$. We note that this forces the identity $\sum_i l_i = \sum_i l'_i$ on the indexes, allowing us to obtain the following expression:
\begin{align}
    \mbox{Tr}\left[ \rho_{out} \mathds{1}_{j} \bigotimes^{M-1}_{r\neq j} \ketbra{0}{0}_r\right] = \Tr[\rho^{(j)}_{out}] = Q(n_0,\dots,n_{N-1}) \label{firsttrace-app},
\end{align}
where we defined:
\begin{align}
    &Q(n_0,\dots,n_{N-1}) = \nonumber\\
    &\sum_{k_0=0}^{n_0} \cdots \sum_{k_{N-1}=0}^{n_{N-1}}\,\, \sum_{(l_0,\dots,l_{N-1},l'_0,\dots,l'_{N-1})\in \mathcal{L}(k_0,\dots,k_{N-1})} \binom{n_0}{k_0} \cdots \binom{n_{N-1}}{k_{N-1}} \binom{k_0}{l_0} \cdots \binom{k_{N-1}}{l_{N-1}} \binom{k_0}{l'_0} \cdots \binom{k_{N-1}}{l'_{N-1}} \nonumber \\
    &\frac{\eta^{\sum_{i}k_i}(1-\eta)^{\sum_{i}(n_i-k_i)}}{M^{\sum_i k_i} k_0!\cdots k_{N-1}!} (\cos{\theta_0})^{l_0+l'_0}(\sin{\theta_0})^{2k_0-l_0-l'_0}(\cos{\theta_1})^{\sum_{i=1}^{N-1}(l_i+l'_i)}(\sin{\theta_1})^{\sum_{i=1}^{N-1}(2k_i-l_i-l'_i)} \left(\sum_i l_i\right)!\left(\sum_i (k_i - l_i)\right)!  \label{Q-app},
\end{align}
where the summation set is defined as
\begin{equation}
    \mathcal{L}(k_0,\dots,k_{N-1}):= \left\{ (l_0,\dots,l_{N-1},l'_0,\dots,l'_{N-1}) : 0 \leq l_i \leq k_i,\; 0 \leq l'_{i} \leq k_i, \; \sum_{i=0}^{N-1} l_i = \sum_{i=0}^{N-1} l'_i \right\}. \label{summationset}
\end{equation}
By using \eqref{secondtrace-app} and \eqref{firsttrace-app} in \eqref{exactyields-app}, we obtain the final expression for the yields in our channel model:
\begin{equation}
    Y^j_{n_0,\dots,n_{N-1}}=(1-p_d)^{M-1} Q(n_0,\dots,n_{N-1})  - (1-p_d)^M (1-\eta)^{\sum_i n_i},  \label{exact-yields}
\end{equation}
where $Q(n_0,\dots,n_{N-1})$ is defined in \eqref{Q-app} and we emphasize once again that the yields are independent of $j$.

\section{Numerical simulations}\label{section:numsim}

In this section we provide more details about the numerical simulations presented in Sec.~\ref{sec:simulations} of the manuscript.

In our simulations, we set a polarization and a phase misalignment between the reference party $A_0$ and each other party of $2\%$. This means that the parameters $\theta_i$ and $\phi_i$, introduced in Appendix~\ref{sec:channelmodel} to describe the polarization rotation and the phase mismatch of party $A_i$, are set to: $\phi_0=0$, $\phi_{i \geq 1}=\phi$, $\theta_{i \geq 1}=\theta_1$  and $\phi=\theta_0 - \theta_1 =\arcsin{\sqrt{0.02}}$. We compute the protocol's key rate for three values of $p_d$, i.e. the probability of a dark count in a given detector, namely: $p_d=10^{-8}, \, 10^{-9}$ and $10^{-10}$. 

As for the decoy-state analysis, we consider two decoy intensities for each party, $\beta_0$ and $\beta_1$, and use the analytical bounds derived in Sec~\ref{sec:decoy} to compute the upper bound \eqref{phase-error-rate-bound} on the phase error rate. The decoy intensity $\beta_0$ is fixed to $\beta_0=0.5$, which we verified is a close-to-optimal value for all loss parameters, while $\beta_1=0$ is optimal. In Sec.~\ref{sec:simulations} we also plot the key rate in the case in which the exact value of the yields is known, which corresponds to the limit where the parties have an infinite number of decoy intensities. The exact values of the yields are computed for our channel model in Appendix~\ref{sec:channelmodel} and are reported in \eqref{exact-yields}. We then replaced the exact yields  $Y_0,\dots,Y_{N-1}$ in the phase error rate bound \eqref{phase-error-rate-bound}, in place of the yields' bounds $\bar{Y}_{0},\dots,\bar{Y}_{N-1}$.

The key rates in Fig.~\ref{fig:rate} are optimized over the signal amplitude $\alpha$ for all values of loss and is computed for $N=3,\, 4$ and $5$ parties. In Fig.~\ref{fig:optalpha} we provide the optimal values of $\alpha$ for every loss, both when we used the yields bounds obtained with two decoys and when we used the exact expressions of the yields from the channel model. By comparing the optimal values of $\alpha$ in the two cases, we deduce that tighter bounds on the yields would allow for a higher optimal value of $\alpha$. This is explained by the fact that having tighter bounds on the yields in the phase error rate bound \eqref{phase-error-rate-bound} allows the yields' coefficients in that expression to grow, i.e. $\alpha$ to grow, without increasing the phase error rate bound. In turn, greater values of $\alpha$ can increase the key rate due to a higher chance of having a detector click (see bottom plot in Fig.~\ref{fig:rate}). Therefore, we deduce that increasing the number of decoy intensities used by each party would lead to better yields' bounds and hence to a significantly improved key rate.

In order to reduce the number of yields that are non-trivially bounded in \eqref{phase-error-rate-bound}, we remark that the polarization and phase angles $\theta_i$ and $\phi_i$ are the same for all parties except for reference party $A_0$. Moreover, the signal and decoy intensities are the same for all parties as well as the losses. Therefore, the channel model is symmetric under the permutation of parties $A_1,A_2,\dots,A_{N-1}$. This implies, in particular, that the yields in \eqref{phase-error-rate-bound} satisfy:
\begin{equation}
    Y_{n_0,n_1, \dots,n_{N-1}} =Y_{n_0,\sigma(n_1,\dots,n_{N-1})},
\end{equation}
where $\sigma(n_1,\dots,n_{N-1})$ represents a permutation of the indexes $n_1,\dots,n_{N-1}$. The permutational symmetry of the yields in our channel model implies that, in computing the phase error rate bound \eqref{phase-error-rate-bound} for a cutoff $\overline{n}=4$ (above which every yield is bounded by one), we only need to bound the following yields for $N=3$: $Y_{0,0,0}$, $Y_{2,0,0}$, $Y_{0,2,0}$, $Y_{4,0,0}$, $Y_{0,4,0}$ $Y_{1,1,0}$, $Y_{0,1,1}$, $Y_{2,2,0}$, $Y_{0,2,2}$, $Y_{1,3,0}$, $Y_{0,1,3}$, $Y_{1,1,2}$.

Similarly, for $N=4$ we only bound the yields: $Y_{0,0,0,0}$, $Y_{2,0,0,0}$, $Y_{0,2,0,0}$, $Y_{4,0,0,0}$, $Y_{0,4,0,0}$, $Y_{1,1,0,0}$, $Y_{0,1,1,0}$, $Y_{2,2,0,0}$, $Y_{0,2,2,0}$, $Y_{1,3,0,0}$, $Y_{0,1,3,0}$, $Y_{1,1,2,0}$, $Y_{0,1,1,2}$, $Y_{1,1,1,1}$.

And for $N=5$ we only bound the yields: $Y_{0,0,0,0,0}$, $Y_{2,0,0,0,0}$, $Y_{2,0,0,0,0}$, $Y_{4,0,0,0,0}$, $Y_{0,4,0,0,0}$ $Y_{1,1,0,0,0}$, $Y_{0,1,1,0,0}$, $Y_{2,2,0,0,0}$, $Y_{0,2,2,0,0}$, $Y_{1,3,0,0,0}$, $Y_{0,1,3,0,0}$, $Y_{1,1,2,0,0}$, $Y_{0,1,1,2,0}$, $Y_{1,1,1,1,0}$, $Y_{0,1,1,1,1}$.

\begin{figure}[th!]
    \centering
    \begin{minipage}{0.495\textwidth}
        \centering
        \textbf{\large Two decoys}\par\medskip
			\includegraphics[width=1\linewidth,keepaspectratio]{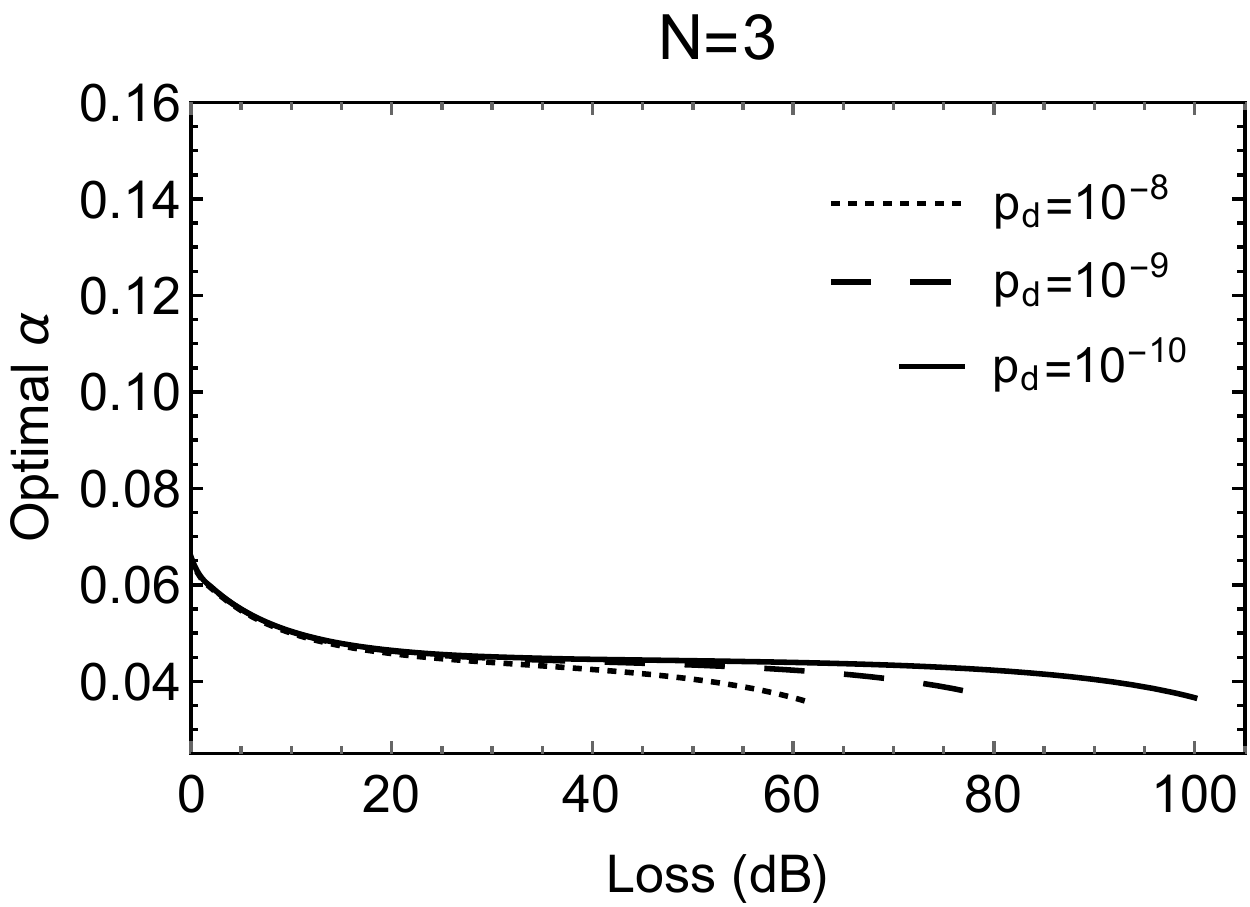}
			\includegraphics[width=1\linewidth,keepaspectratio]{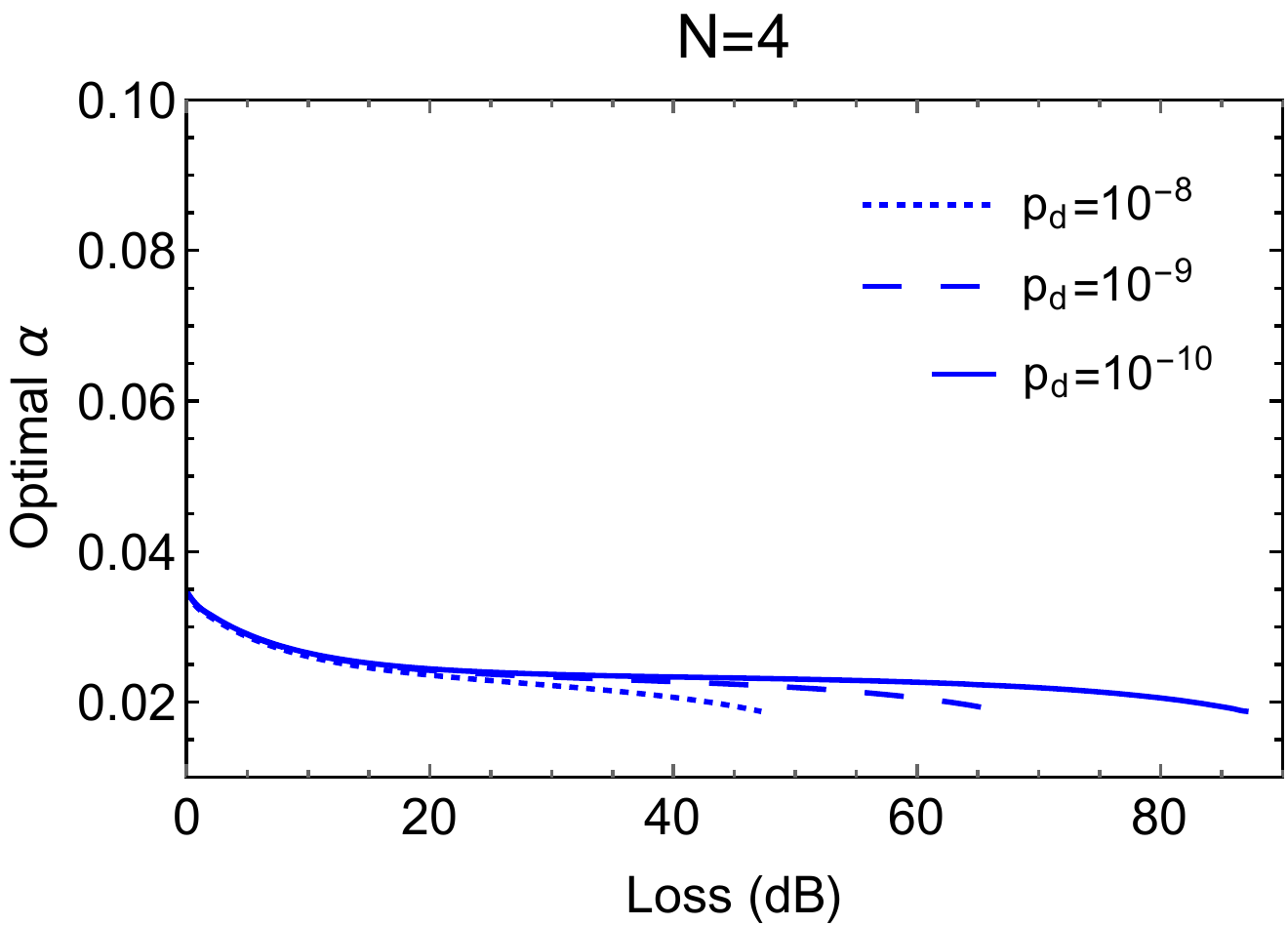}
			\includegraphics[width=1\linewidth,keepaspectratio]{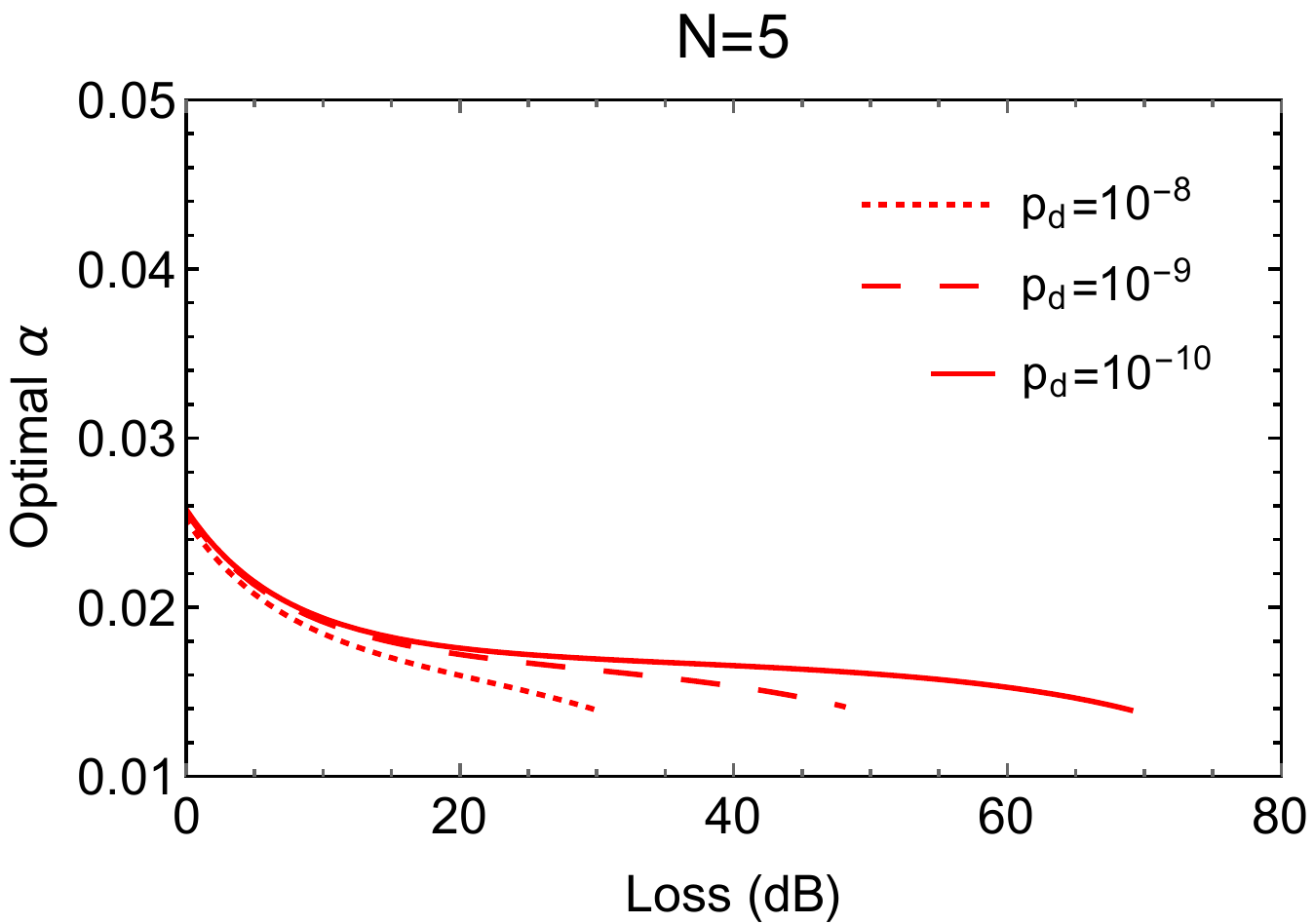}
	\end{minipage}\hfill\vline\hfill
	\begin{minipage}{0.495\textwidth}
        \centering
        \textbf{\large Exact yields}\par\medskip
        \includegraphics[width=1\linewidth,keepaspectratio]{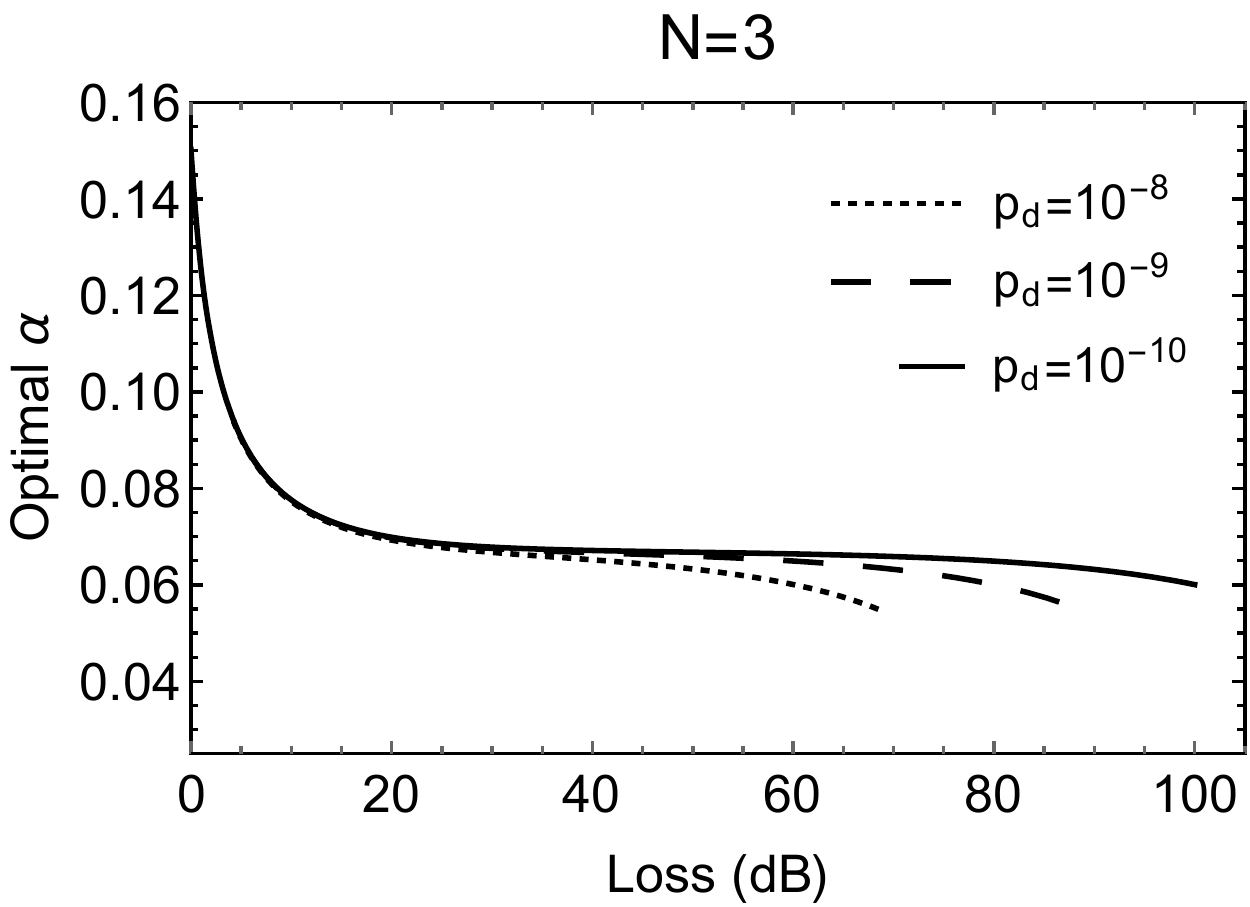}
		\includegraphics[width=1\linewidth,keepaspectratio]{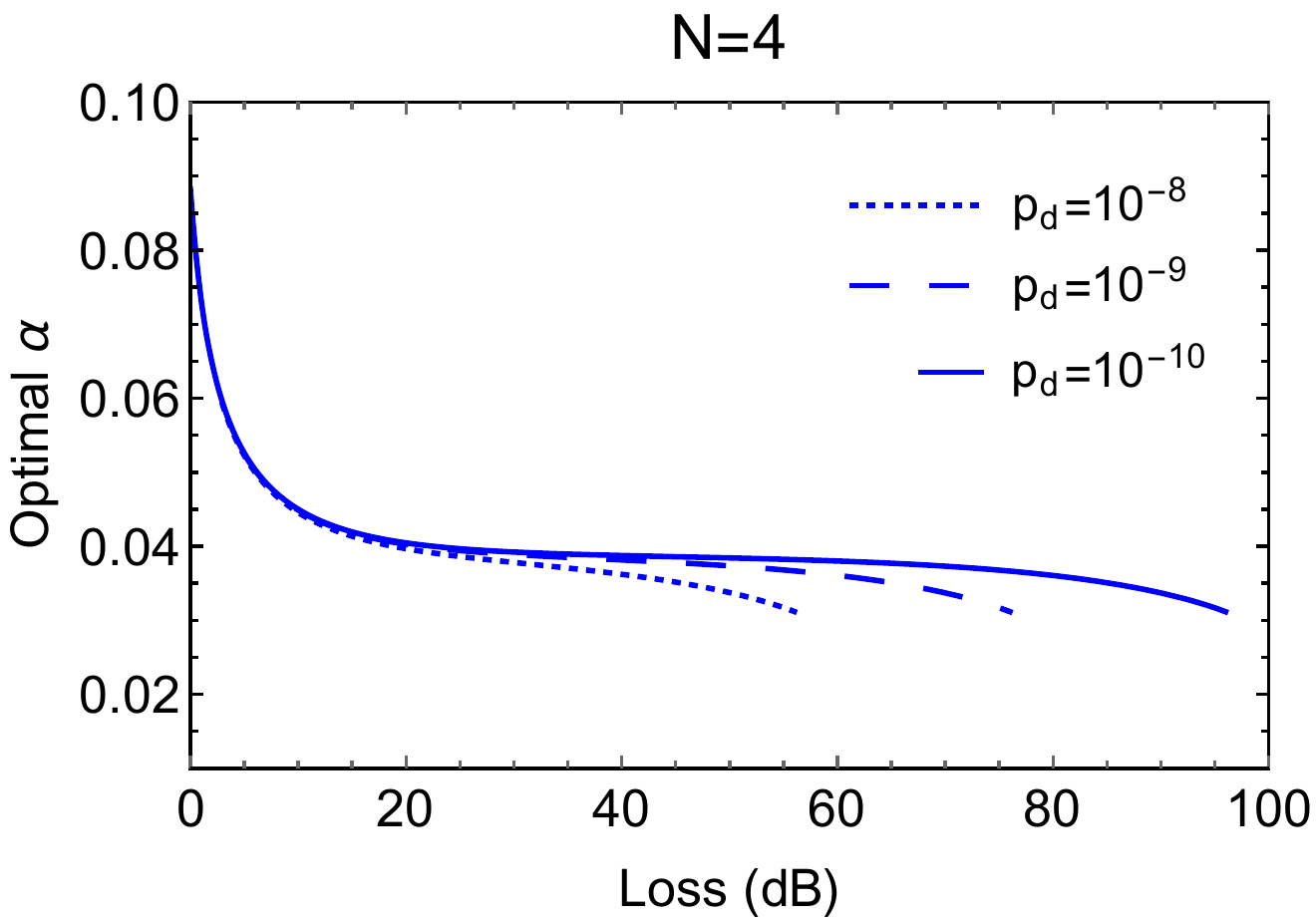}
		\includegraphics[width=1\linewidth,keepaspectratio]{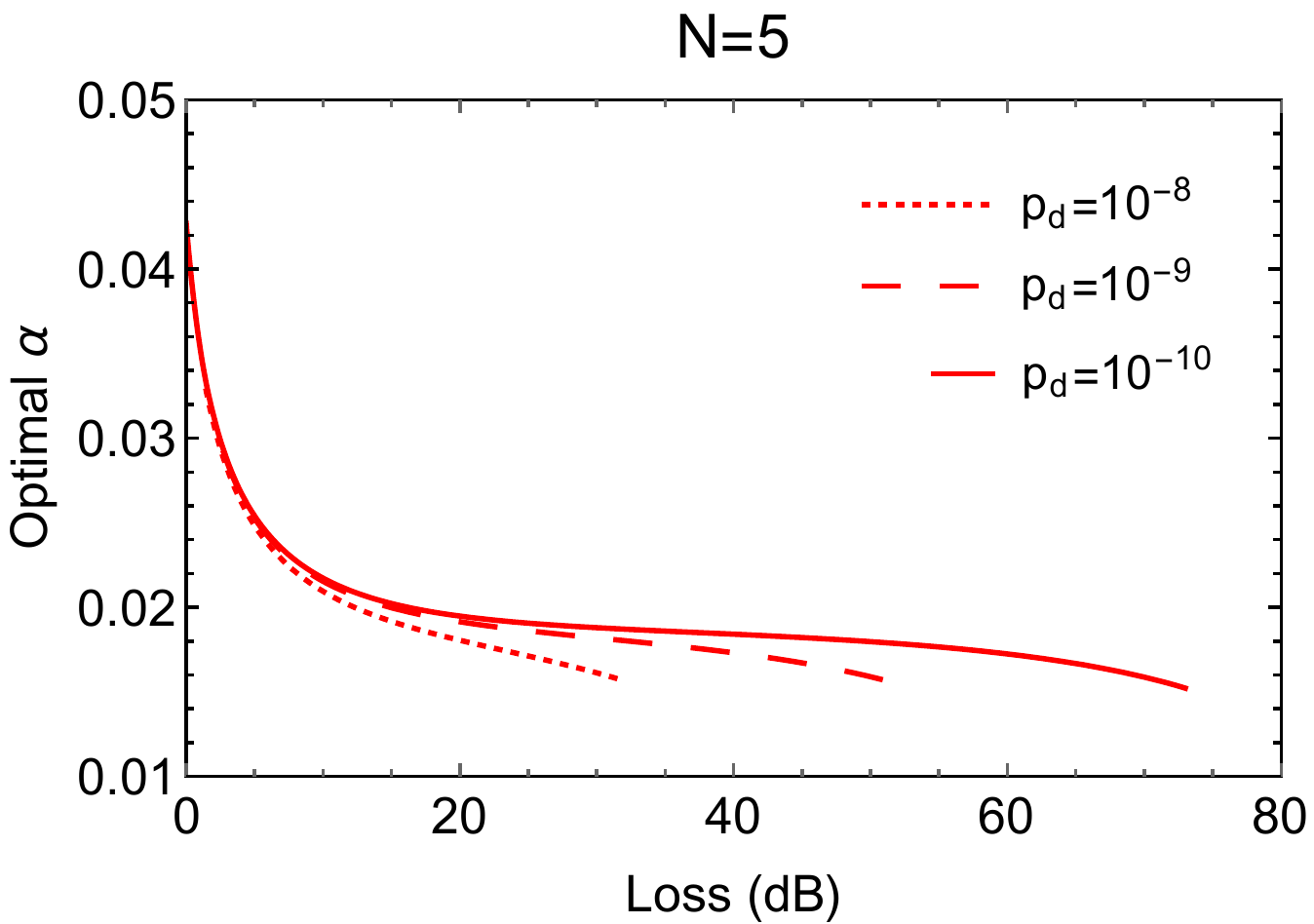}
	\end{minipage}
\caption{The optimal value of the signal amplitude ($\alpha$) that maximizes the key rate plotted in Fig.~\ref{fig:rate}, for different values of the dark count probability ($p_d$) and number of parties ($N$). On the left, the key rate is computed by using the analytical bounds on the yields \eqref{upbound} in the phase error rate bound \eqref{phase-error-rate-bound}, while the plots on the right use the exact expressions of the yields for our channel model \eqref{exact-yields}. We observe that a tighter bound on the yields allows for a higher value of $\alpha$ and leads to a higher key rate (see Fig~\ref{fig:rate}).}
\label{fig:optalpha}
\end{figure}

\end{document}